\documentclass[]{article}

\usepackage{graphicx}
\usepackage{algorithmic}
\usepackage{algorithm}
\usepackage{amsmath}
\usepackage{amsthm}
\usepackage{amssymb}
\usepackage{mathrsfs}
\usepackage{framed}

\DeclareMathOperator*{\argmax}{argmax}

\DeclareSymbolFont{AMSb}{U}{msb}{m}{n}
\DeclareMathSymbol{\N}{\mathbin}{AMSb}{"4E}
\DeclareMathSymbol{\Z}{\mathbin}{AMSb}{"5A}
\DeclareMathSymbol{\R}{\mathbin}{AMSb}{"52}
\DeclareMathSymbol{\Q}{\mathbin}{AMSb}{"51}
\DeclareMathSymbol{\I}{\mathbin}{AMSb}{"49}
\DeclareMathSymbol{\C}{\mathbin}{AMSb}{"43}

\newcommand{\NWREG}{\mathit{NWREG}}
\newcommand{\NREG}{\mathit{NREG}}
\newcommand{\cP}{\mathcal{P}}
\newcommand{\cB}{\mathcal{B}}

\newcommand{\commentout}[1]{}
\newcommand{\ucP}{\overline{\cP}^+}
\newcommand{\cQ}{\mathscr{Q}}
\newcommand{\sd}{\mathit{S.D.}}
\newcommand{\cont}{\mathit{cont}}
\newcommand{\chec}{\mathit{check}}
\newcommand{\back}{\mathit{back}}
\newcommand{\regret}{\mathit{reg}}

\newtheorem{definition}{Definition}

\newtheorem{axiom}{Axiom}

\newtheorem{myTheorem}{Theorem}

\newtheorem{lemma}{Lemma}

\newtheorem{example}{\texttt{[Example]}}
\newtheorem{proposition}{Proposition}
\DeclareSymbolFont{AMSb}{U}{msb}{m}{n}
\DeclareMathSymbol{\N}{\mathbin}{AMSb}{"4E}
\DeclareMathSymbol{\Z}{\mathbin}{AMSb}{"5A}
\DeclareMathSymbol{\R}{\mathbin}{AMSb}{"52}
\DeclareMathSymbol{\Q}{\mathbin}{AMSb}{"51}
\DeclareMathSymbol{\I}{\mathbin}{AMSb}{"49}
\DeclareMathSymbol{\C}{\mathbin}{AMSb}{"43}

\newcommand{\citeyear}{\cite}

\begin{document}
%joe8: for papers, I use ``Joseph Y.''
%\author{Joe Halpern, Samantha Leung}
%joe16
%\author{Joseph Y. Halpern  Samantha Leung}
\author{
Joseph Y. Halpern\\
Cornell University\\
halpern@cs.cornell.edu 
\and Samantha Leung\\
Cornell University\\
samlyy@cs.cornell.edu}

%\author{}
%joe1*: I think this confounds two issues: we have an approach to
%decision making that involves minimax weighted regret, which makes
%sense whether or not you update.  We also have a technique for updating
%the weights, which has nothing to do with making decisions.  Put
%another way, we have two contributions here: a new way of representing
%uncertainty: weighted sets of probailty, and a new approach to decision
%making: using minimax weighted expected regret.  So here's a clunky
%title that makes that clear, which we can then modify 
\title{Weighted Sets of Probabilities and Minimax Weighted Expected
Regret: 
New Approaches for Representing Uncertainty and Making Decisions
\thanks{The authors thank Joerg Stoye for useful comments. Work supported in part by NSF grants IIS-0534064, IIS-0812045, and
IIS-0911036, by AFOSR grants
FA9550-08-1-0438 and FA9550-09-1-0266, and by ARO grant W911NF-09-1-0281.}}

\maketitle

\begin{abstract}
%joe8: I'm not sure what this is doing here!
%$$ f \succeq_M^{\cP^+,U} g
%$$
%$$ \succeq_{M,E}^{\cP^+,U}
%$$
%joe1*: The inadequacies have nothing to do with MER!  Rewrote
%To address the inadequacies of measure-by-measure updating for
%multiple-prior minimax expected regret (MER) preferences, we propose
%that a multiple prior model should include weights indicating the
%significance of each prior. We describe the natural way of updating such
%a model, and provide an axiomatization for likelihood updating when
%initial preferences are MER. 
We consider a setting where an agent's uncertainty is represented by
%joe2
%a set of probability distributions, rather than a single distribution.
%Probability-by-probability updating of such a set of distributions upon 
a set of probability measures, rather than a single measure.
Measure-by-measure updating of such a set of measures upon 
acquiring new information is well-known to suffer from problems;
agents are not always able to learn appropriately.  To deal with these
problems, we propose using \emph{weighted sets of probabilities}: a
representation where each measure is associated 
with a \emph{weight}, which denotes its significance.  
We describe a natural approach to updating in such a situation and a
natural approach 
to determining the weights.  
We then show how this representation can be used in decision-making, by modifying a standard approach to decision
making---minimizing expected regret---to obtain \emph{minimax weighted
expected regret} (MWER).  We provide an axiomatization that characterizes 
preferences induced by MWER
%joe2
%, as well as an axiomatization that
%characterizes conditional preferences induced by MWER and the natural
%updating approach. 
both in the static and dynamic case.
\end{abstract}

\section{Introduction}
%joe: This is a perfectly good start, but doesn't sound like AAMAS.
%The paper really is more appropriate for UAI than AAMAS, but if we're
%going to submit to AAMAS, let's make it sound like it's about
%autonomous agents
%From deciding between crispy fries and a bland salad, to forming an
%investment portfolio, to military planning --  
%our decisions can significantly impact our lives and those of others. 
%More often than not, there are uncertainties in our decision problems.
Agents must constantly make decisions; these decisions are typically made in a setting with uncertainty.
%joe: it's technically the subjective here; rewrote slightly
%uncertainty can be quantified by a probability distribution over the
%possible states.  
%If one was betting on the toss of a fair coin or fair dice, the
%the uncertainty can be quantified  by a probability distribution over the
%possible states.  
For decisions based on the outcome of the toss of a fair coin, the
uncertainty can be well characterized by probability.
%joe2
%However, what is the probability of you getting cancer if you ate fries
However, what is the probability of you getting cancer if you eat fries
at every meal?  
%joe2
%What is the probability if you have salads instead? 
What if you have salads instead? 
Even experts would not agree on a single probability. 
%joe: cut; we can use the space for other things
%This type of uncertainty, which cannot be captured 
%by a probability
%distribution, is termed \emph{ambiguity}.  

% I would start with a discussion of the unreasonableness of modeling
% uncertainty using probability and maximizing expected utility if you
% have no (or very little) information  about the probability.   Gilboa
% has been writing a lot about that recently.  See his papers
% "Rationality of belief (with Postlethwaite and Schmeidler), Questions
% in Decision Theory, and "Ambiguity and the Bayesian paradigm" (with
% Marinacci).     

%joe: I downplayed this, because it's not clear what the story is.  If you
%were a Bayesian, you would simply say that people make mistakes.  They
%*should* represent their beliefs by a single distributin and use MUE,
%but they don't.  If this were the story, we would have to argue that
%MWER is a good *desscriptive* tool; it predicts how people behave.
%While MWER does not give these paradoxes, we have nothing to say about
%whether it is descriptively accurate
\commentout{
Not only is it sometimes unnatural to represent uncertainty using a
single probability distribution; but doing so (and assuming 
expected utility maximization) often predict behavior that contradict
intuition.  
The most well-known of these contradictions are probably the Allais
Paradox \cite{Allais1953} and the Ellsberg Paradox \cite{Ellsberg1961}.   
%The Allais Paradox involves a decision problem where common choices
%contradict maximizing expected utility; the Ellsberg Paradox attacks
%the assumption that an agent believes in a single prior.  
%As a more close-to-home example, Gilboa and Marinacci
%\cite{GilboaMarinacci2011} describe a couple who tries to learn their
%probabilities of getting a heart attack within the next 10 years, in
%order to make decisions about purchasing health insurance. 
%After consulting reputable sources, the husband gets probabilities
%ranging from $25$ to $50$ percent. What should his belief be if we must
%represent it as a single probability? 
}
Representing uncertainty by a single probability measure and making
decisions by maximizing expected utility leads to further problems.  
%sam8: changed the example from quality control problem to robot
%delivery problem 
Consider the following stylized problem, which serves as a running example in this paper.
The baker's delivery robot, T-800, is delivering $1,000$ cupcakes from the bakery to a banquet. 
Along the way, T-800 takes a tumble down a flight of stairs and breaks some of the cupcakes. 
The robot's map indicates that this flight of stairs must be either ten feet or fifteen feet high. 
%sam20 replace numbers with words
%For simplicity, assume that a fall of ten feet results in $1$ broken
%cupcake, while a fall of fifteen feet results in $10$ broken cupcakes.  
For simplicity, assume that a fall of ten feet results in one broken
cupcake, while a fall of fifteen feet results in ten broken cupcakes.  

%A business owner (the decision maker, agent) is contracted to produce
%%joe2: added commas in numbers; did this globally
%%$1000$ items, and will be rewarded when she delivers the items, but
%$1,000$ items, and will be rewarded when she delivers the items, but
%punished if she delivers a batch with too many broken cakes.  
%She has recently switched a raw material supplier, and does not know
%whether the supplier is reliable, and provides good quality materials,
%or unreliable, and provides bad quality materials.
%%joe: we should acknolwedge that it's
%%more natural to assume that good quality results in a distribution over
%%outcomes, not a unique outcome.
%For simplicity, assume that using good quality raw materials results in
%%joe3
%%$1$ broken cake in the batch of $1,000$, while using bad quality raw
%one broken cake in the batch of $1,000$, while using bad quality raw
%materials results in ten broken cakes.  

\begin{table}
	\centering
		\begin{tabular}{|l|l|l|l|l|l|l|} \hline
&	$1$ broken & $10$ broken \\ \hline
		$\cont $  				& 10,000
		& -10,000				\\ \hline 
			$\back$ 				& 0									& 0 \\ \hline 
			$\chec	$					&	5,001							& -4,999 \\ \hline 
		\end{tabular}
	\caption{Payoffs for the robot delivery problem. 
Acts are in the leftmost column.
The remaining two columns describe the outcome for the two sets of states that matter.} 
	\label{tab:QC}
\end{table}

T-800's choices and their consequences are summarized in Table \ref{tab:QC}.
Decision theorists typically model decision problems with states, acts,
and outcomes: 
the world is in one of many possible states,
and the decision maker chooses an \emph{act}, a function mapping states to outcomes.
%A natural state space in this problem is $\{good,broken\}^{1000}$,
A natural state space in this problem is $\{$\emph{good,broken}$\}^{1000}$, where each state is a possible state of the cupcakes. 
However, all that matters about the state is the number of broken cakes, so we can further restrict to states with either one or ten broken cakes. 

T-800 can choose among three acts: $\cont$: continue the delivery attempt; $\back$: go back for new cupcakes; 
or $\chec$: open the container and count the number of broken cupcakes, and then decide to continue or go back, depending on the number of broken cakes.
%sam20 use 'one' 'ten' instead of numbers
%The client will tolerate $1$ broken cupcake, but not $10$ broken cupcakes. 
The client will tolerate one broken cupcake, but not ten broken cupcakes. 
Therefore, if T-800 chooses $\cont$, it obtains a utility of $10,000$ if there is only one broken cake, 
but a utility of $-10,000$ if there are ten broken cakes. 
If T-800 chooses to go $\back$, then it gets a utility of $0$. 
Finally, checking the cupcakes costs $4,999$ units of utility but is reliable, so if T-800 chooses $\chec$, 
it ends up with a utility of $5,001$ if there is one broken cake, and a utility of $-4,999$ if there are ten broken cakes.

If we try to maximize expected utility, we must assume some probability
over states.  What measure should be used?  
%You talk about ``hypotheses'' below, but never introduce them.
%I introduced them here
There are two hypotheses that T-800 entertains: (1) the stairs are ten
%joe8
%feet tall and  
%(2) the stairs are fifteen feet tall. 
feet high and  
(2) the stairs are fifteen feet high. 
Each of these places a different probability on states. 
%joe8
%If the stairs are ten feet tall, we can take all of the $1,000$ states
If the stairs are ten feet high, we can take all of the $1,000$ states
where there is exactly one broken cake 
to be equally probable, and take the remaining states to have
%joe8
%probability 0; if the stairs are fifteen feet tall, we can take all of
probability 0; if the stairs are fifteen feet high, we can take all of
the $C(1000,10)$ states where there are exactly ten 
broken cakes to be equally probable, and take the remaining states to
have probability 0.  One way to model T-800's uncertainty about the height of the stairs is to take each hypothesis to be equally likely.  
However, not having any idea about which hypothesis holds is very different 
%joe2: no reason to italicize this
%from believing that all hypotheses are \emph{equally likely}.  
from believing that all hypotheses are equally likely.  
%joe: I don't want to talk about ``priors''.  In our setting, we just
%have sets of distributions, which are constantly being updated.  There
%are no priors and posteriors.
%In this case, a $50/50$ prior gives an expected utility maximizing
%strategy of paying for the test, but placing $51\%$ probability on $bad$
%quality raw material gives an expected utility maximizing strategy of 
%selling the batch at a $\back$ price without testing.  
%How can one justify one choice or the other?
It is easy to check that 
taking each hypothesis to be equally likely makes $\chec$ the act
that maximizes utility, but taking the probability that the stairs are fifteen feet high to be $.51$ makes $\back$ the act that maximizes expected
utility, and taking the probability that the stairs are ten feet high to be $.51$ makes $\cont$ the act that maximizes expected utility.  
What makes any of these choices the ``right'' choice?

%joe: I'm very suspicious about the claim in the last sentence of this
%paragraph.  It seems to me more of a problem with using sets
%of distributions and doing measure-by-measure updating.  I
%can easily come up with a model where you learn perfectly well if
%information is represented by a single distribution.  
%sam: ES were talking about using a single, improproper beta
%distribution. 
%joe2: If you were to mention this (which you shouldn't) you should let
%the reader in on this fact
\commentout{
Optimal portfolio problems are other examples where
Bayesian\footnote{Here, as is in most of the Decision Theory literature,
Bayesianism implies (1) using states to model all uncertainty, and a
single probability distribution to represent the agent's beliefs; (2)
using Bayes rule to update beliefs; (3) maximizing expected utility with
respect to the beliefs. \cite{Gilboa2004}} models produce
counterintuitive decisions. 
In empirical economics, stock returns are often assumed to have a stable
distribution.  
A novice investor is uncertain about the true distribution, thus the
optimal portfolio devotes less in stocks.  
Each trading period, the investor learns more about the distribution. 
Common sense dictates that the proportion of stocks in the optimal portfolio should increase as the investor becomes more confident. 
However, Bayesian models fail to capture this; the optimal proportion of
stocks does \emph{not} increase as more data is gathered
\cite{Epstein2007}. 
}

%joe: Again, our goal isn't to explain behavior.  If it were, we would
%have to rewrite the papre
%These examples suggest that a single probability distribution cannot
%explain intuitive behavior when there are risks with unknown odds, and
%that a different decision model is needed. 
It is easy to construct many other examples where 
a single probability measure does not capture uncertainty, and does
not result in what seem to be reasonable decisions, when combined with
expected utility maximization.
%joe3: removed paragraph break
% 
%joe: Again, you're confusing two issues: the representation of
%uncertainty, and the decision rule.  Let's focus on the representation
%issues first.
%To this end, a number of multiple-prior models were devised to address
%the inadequacies of the Bayesian model.  
%Two popular decision models are maxmin expected utility (MMEU)
%\cite{GilboaSchmeidler1989} and minimax expected regret (MER)
%\cite{Hayashi2006}. 
A natural alternative, which has often been considered in the
literature, is to represent the agent's uncertainty by a \emph{set} of
probability measures.  
%joe: note that here you're talking about the representation, not the
%decision rule.
%For example, in the quality control problem, the agent's beliefs could be
%represented by two probability distributions-- one assigning uniform
%joe3: removed para break
%
For example, in the delivery problem, the agent's beliefs could be
represented by two probability measures, $\Pr_1$ and $\Pr_{10}$, one
for each hypothesis. 
Thus, $\Pr_1$ 
%joe1
%assigns uniform probability to all states with exactly 1 defect, and
%joe3
%assigns uniform probability to all states with exactly 1 broken cake,
assigns uniform probability to all states with exactly one broken cake,
%joe3
%and $\Pr_{10}$ assigns uniform probability to all states with exactly
%10 
and $\Pr_{10}$ assigns uniform probability to all states with exactly ten 
%joe2
%defects.   
broken cakes.

%joe: YOU NEED TO SLOW DOWN HERE.  
But this representation also has problems. Consider the delivery
example again.  
%sam20 added "In a more realistic scenario", since in our example it is
%indeed the case that we know it is either 1 or 10 broken cupcakes 
%joe20: I prefer the original wording; to my ear, the addition sounds a
%bit clunky
%In a more realistic scenario, why should T-800 be sure
Why should T-800 be sure
%joe3
%that there is exactly either one broken cake or 10 broken cakes?
that there is exactly either one broken cake or ten broken cakes?
%joe2
%Of course, we can replace the hypothesis so that it says that the
Of course, we can replace these two hypotheses by hypotheses that say
that the probability of a cake being broken is either $.001$ or $.01$, but this
doesn't solve the problem.  Why should the agent be sure that the
probability is either exactly $.001$ or exactly $.01$?  Couldn't it also
be $.0999$? Representing uncertainty by a set of measures still
places a sharp boundary on what measures are considered possible
and impossible.

A second problem involves updating beliefs.
How should beliefs be updated if they are represented by a set of
probability measures? 
The standard approach for updating a single measure is by
conditioning.  
The natural extension of conditioning to 
sets of measure is measure-by-measure updating:
conditioning each measure on the information (and also removing
measures that give the information probability 0).

%joe
%However, measure-by-measure updating does not always capture our learning
%capabilities.  
%Let us illustrate. 
However, measure-by-measure updating can produce some rather
counterintuitive outcomes.
In the delivery example, suppose that a passer-by tells T-800 the information $E$: the first 100 cupcakes are good.
Assuming that the passer-by told the truth, intuition tells us that there
%joe3
%is now more reason to believe that there is only 1 broken cake in
%the
is now more reason to believe that there is only one broken cupcake.

However, $\Pr_1 \mid E$ places uniform probability on all states where
the first 100 cakes are good, and there is exactly one broken cake among
the last 900.   
Similarly, $\Pr_{10}\mid E$ places uniform probability on all states
%joe3
%where the first 100 items are good, and there are exactly 10 broken
where the first 100 cakes are good, and there are exactly ten broken
cakes among the last 900.  
%joe: this misses the point.  There are no probabilities on hypotheses
%Clearly, after conditioning on $E$, the probabilities that there is 1
%broken cake is still 1 with with $
%item' and  `10 defects' do not change at all.  
%In fact, even if the worker had told the business owner that the first
%$990$ items are all good, the probabilities are still unaffected! 
$\Pr_1 \mid E$ still places probability $1$ on there being one broken
cake, just like $\Pr_1$,  $\Pr_{10} \mid E$ still places probability $1$
on there being ten broken cakes.  There is no way to capture the fact
%joe3
%that the agent now views the hypothesis $\Pr_{10}$ as less likely -- even
that T-800 now views the hypothesis $\Pr_{10}$ as less likely, even
if the passer-by had said instead that the first $990$ cakes are all good! 

%joe: new paragraph, where I lay out the key contributions of the
%paper.  This is a critical paragraph, and was missing before
%sam4:
Of course, both of these problems would be alleviated if we placed
%Of course, both of these probabilities would be alleviated if we placed
a probability on hypotheses, but, as we have already observed, this leads
to other problems.  In this paper, we propose an intermediate approach:
representing uncertainty using \emph{weighted sets of probabilities}.
That is, each probability measure is associated with a weight.  These
%joe7
%weights can be viewed as probabilities, but they do not quite act like
%probabilities in general.  For example, if we had a countable set of
%hypotheses, we could assign them all weight 1 (so that, intuitively, if
%they are all viewed as equally likely), but there is no uniform
%measure on a countable set.   
weights can be viewed as probabilities; 
indeed, if the set of probabilities is finite, we can normalize them so
that they are effectively probabilities.  Moreover, in one important
setting, we update them in the same way that we would update
probabilities, using likelihood (see below).  On the other hand, 
these weights do not act like probabilities if the set of probabilities is
infinite.  For example, if we had a countable set of
hypotheses, we could assign them all weight $1$ (so that, intuitively,
they are all viewed as equally likely), but there is no uniform
measure on a countable set. 

More importantly, when it comes to decision making, we use the weights
quite differently from how we would use second-order probabilities on
probabilities.  Second-order probabilities would let us define a
probability on events (by taking expectation) and maximize expected
utility, in the usual way.  Using the weights, we instead define a 
novel decision rule, \emph{minimax weighted expected regret (MWER)}, that
has some rather nice properties, which we believe will make it widely 
applicable in practice.  
%joe7*: describing MWER before we do likelihood updating; that seems to
%me a better story.  Moved the discussion of likelihood updating that
%used to be here below
%We define a decision rule that can be used with weighted sets of
%probabilities, which we call \emph{minimax  weighted expected regret
%(MWER)}.   
If all the weights are $1$, then MWER is just the standard
\emph{minimax expected regret (MER)} rule (described below).  
%joe7
%As the weighted set of measures converges to a single measure, MWER
If the set of probabilities is a singleton, then MWER agrees with
(subjective) expected utility maximization (SEU).  More interestingly
perhaps, if the weighted set of measures converges to a single measure
%joe7: added
%MWER converges to subjective expected utility maximization (SEU).  Thus, the
(which will happen in one important special case, discussed below),
MWER converges to SEU.  Thus, the weights give 
%joe2
%
%us a smooth and natural way of interpolating between MER and SEU.  
us a smooth, natural way of interpolating between MER and SEU.  

%sam16 Added statement of motivation and advantages as per UAI2012 reviews
In summary, weighted sets of probabilities allow us to represent ambiguity (uncertainty about the correct probability distribution).
Real individuals are sensitive to this ambiguity when making decisions, and the MWER decision rule takes this into account.
Updating the weighted sets of probabilities using likelihood allows the initial ambiguity to be resolved as more information about the true distribution is obtained.

\commentout{
We propose a representation that gives the `best of both worlds' of the
probabilistic and multiple prior models.  
We maintain a set of likelihoods of priors, and use the likelihoods as weights indicating how seriously a certain prior should be taken. 
Note that these likelihoods are not to be interpreted as a second-order
distribution. 
This approach allows us to model ambiguity (via the presence of multiple priors) as well as confidence (via the weights on each prior). 

The natural way to update such a belief representation on an event $E$
would be to condition the individual measures to reflect the
elimination of possible states, and also update the relative likelihoods
based on the event.  
Assuming that there is a stable true state, and that some prior with
positive weight assigns positive probability to this true state, this
updating method captures our intuitions for learning -- in the quality
control example, the evidence that some of the items are good increases
the likelihood that the raw material was of good quality. 
%joe20
%We will refer to the above natural updating method on the weighted
We refer to this updating method on the weighted
representation as \emph{likelihood updating}.  

Although \emph{likelihood updating} is orthogonal to the decision rule
being used, for concreteness, we focus on applying this weighted
representation and updating technique to minimax expected regret (MER)
preferences.  
The main reason we focus on applying \emph{likelihood updating} to MER is
because the weighting of priors is naturally achieved for MER by
multiplying each expected regret by the likelihood of the prior.  
Multiplying the expected regret from a particular prior by a larger
weight amplifies the regret and makes the prior more significant.  
%joe10*: I don't think you've defined MMEU yet
The same is not true for the MMEU representation -- multiplying a
positive payoff by a larger weight makes the payoff better, but
multiplying a negative payoff by a larger weight makes the payoff
worse. 
}
 
We now briefly explain MWER, by first discussing MER.  
MER is a probabilistic variant of the minimax regret decision rule proposed by
Niehans \citeyear{Niehans1948} and Savage \citeyear{Savage1951}.  
Most likely, at some point, we've second-guessed ourselves and thought
``had I known this, I would have done that instead''.  
%joe: choice -> act
%That is, in hindsight, we regret not choosing the choice that turned out
That is, in hindsight, we regret not choosing the act that turned out
to be optimal for the realized state, called the \emph{ex post} optimal
act.
The \emph{regret} of an act $a$ in a state $s$ is the difference (in utility)
%joe2
%between the ex post optimal act and $a$.
%Of course, often one does not know the true state at the time of 
between the ex post optimal act in $s$ and $a$.
Of course, typically one does not know the true state at the time of 
decision. 
%joe
%hence one can only minimize \emph{anticipated} regret. 
%joe: added the rest of the explantion
Therefore the regret of an act is the worst-case regret, taken over all states.
%joe
%Minimax Regret captures this concept and chooses an act that minimizes
%the maximum regret across the states.  
The \emph{minimax regret} rule orders acts by their regret.

%joe: slow down!
%Because Minimax Regret is not concerned with how likely each state is, 
%it is also called \emph{priorless} minimax regret.  
%MER, on the other hand, is a multiple-prior model, with each prior
%being a distribution over the states. 
The definition of regret applies if there is no probability on states.
If an agent's uncertainty is represented by a single probability
measure, then we can compute the \emph{expected regret} of an act $a$:
just multiply the regret of an act $a$ at a state $s$ by the probability
%joe2
%of $s$, and then sum.  It is well known that ordering on states by
%expected regret is identical to the ordering by expected utility
%\cite{Hayashi2006}.  
of $s$, and then sum.  It is well known that the order on acts
%sam4
induced by minimizing expected regret is identical to that induced by
maximizing expected utility 
%induced by expected regret is identical to that induced by expected utility
(see \cite{Hayashi2006} for a proof).
%Just as we can compute expected utility for a given prior, we can
%compute the expected regret of an act, with respect to a prior, by
%taking the weighted sum of the state-wise regrets with respect to a
%given prior. 
%Each prior thus gives an expected regret for each act. 
If an agent's uncertainty is represented by a set $\cP$ of probabilities,
then we can compute the expected regret of an act $a$ with respect to
each probability measure $\Pr \in \cP$, and then take the worst-case expected
regret.  The MER (Minimax Expected Regret) rule orders acts according to
%joe2
%their worst-case expected regret, preferring the 
their worst-case expected regret, preferring 
the act that minimizes the worst-case regret. 
If the set of measures is the set of \emph{all} probability measures on
states, then it is not hard to show that MER induces the same order on acts as
(probability-free) minimax regret.  Thus, MER generalizes both minimax
regret (if $\cP$ consists of all measures) and expected utility
maximization (if $\cP$ consists of a single measure).
%The maximum expected regret always occurs at one of the vertices of this
%simplex, corresponding to a probability mass on a single state. 
%Furthermore, subjective expected utility maximization (with respect to a
%single prior) is also a special case of MER. 
%When there is only a single prior, minimizing expected regret with
%respect to that prior distribution is equivalent to maximizing expected
%utility with respect to the same distribution. 

MWER further generalizes MER. If we start with a \emph{weighted} set of
measures, then we can compute the weighted expected regret for each one (just
%sam4
multiply the expected regret with respect to $\Pr$ by the weight of $\Pr$) and compare acts 
%multiply the expected regret of $\Pr$ by the weight of $\Pr$) and compare acts 
by their worst-case weighted expected regret.

%sam8 talk about Sarver's regret preferences
%joe8: this will be totally meaningless to those who haven't read
%Sarver's paper
%For those familiar with the \emph{regret representation} from Sarver
%\cite{Sarver2008}, MWER may appear similar to the regret representation,
%since both put a multiplicative weight on a regret quantity.  
%However, Sarver's regret representation is fundamentally different from
%MWER. Most importantly, in Sarver's regret representation, regret is a
%factor only when comparing two sets of acts.
%The ranking of individual acts is just expected utility maximization.
Sarver \citeyear{Sarver2008} also proves a representation theorem
that involves putting a multiplicative weight on a regret quantity.  
However, his representation is fundamentally different from MWER.
In his representation, regret is a factor only when comparing two \emph{sets}
of acts; 
the ranking of individual acts is given by expected utility maximization.
%joe8: added
By way of contrast, we do not compare sets of acts.

%joe
%It is standard in decision theory to axiomatize a model via a 
It is standard in decision theory to axiomatize a decision rule by means
of a representation theorem. 
%A representation theorem, or an axiomatization, typically states that
%if preferences satisfy certain axioms, then they must be representable
%using the aforementioned model.  
%The theorem would also say the converse: that the model in question
%satisfies those axioms.   
For example, Savage \citeyear{Savage1954} showed that if an agent's preferences
$\succeq$ satisfied several axioms, such as completeness and
transitivity, then the agent is behaving as if she is maximizing expected
utility with respect to some utility function and probabilistic belief.  
%joe deleted
%These axioms never refer to parts of the representation, such as the
%utility function or the probabilistic beliefs in the case of Savage, as
%doing so would defeat the purpose of the representation theorem. 

If uncertainty is represented by a set of probability measures,
then we can generalize expected utility maximization to \emph{maxmin
%joe2: I changed MEU to MMEU globally.  I find MEU confusing, since it's
%sometimes used for maximing expected utility
expected utility (MMEU)}.   
MMEU compares acts by their worst-case expected utility, taken over all
measures.  
%Similarly, static MMEU was 
MMEU has been axiomatized by Gilboa and Schmeidler
\citeyear{GilboaSchmeidler1989}.
%joe7: moved below
%A dynamic version of MMEU with measure-by-measure updating 
%% and measure-by-measure updating on MMEU have
%has been axiomatized by Jaffray \cite{Jaffray1994}, Pires \cite{Pires2002},
%and Siniscalchi \cite{Siniscalchi2011}.   
%joe8: we didn't mention
%As we mentioned, 
MER was axiomatized by Hayashi \citeyear{Hayashi2006} and Stoye
\citeyear{StoyeRegret}. 
%joe
%However, we are not aware of any axiomatization for conditional multiple
%prior MER preferences. 
%Thus, we fill this gap and provide a representation theorem for
%\emph{likelihood updating} on MER;
We provide an axiomatization of MWER.
%joe7
%as well as an axiomatization for dynamic MWER with likelihood
%updating. 
We make use of ideas introduced by Stoye \citeyear{StoyeRegret} in his
axiomatization of MER, but the extension seems quite nontrivial.

%In particular, we show that \emph{likelihood updating} on MER is
%characterized by adding a single axiom, M-DC, to the axiomatization of
%MER.  

%sam18 Added reference to Chateauneuf and Faro
%joe18: I moved this down until after we discuss representation theorems
%and updating.  It seems strange to refer to them in the discusion of CF
%before we mention them in the context of MMEU

%joe2
%If there is a true probability measure describing the world, and the
%distribution is stable over time, then there is a natural way of
%joe7: moved material on dynamic updating here
We also consider a dynamic setting, where beliefs are updated by new
information.  
If observations are generated according to a probability measure that is
stable over time, then, as we suggested above,
there is a natural way of updating the weights given observations, using
ideas of likelihood.  The idea is straightforward.  
After receiving some information $E$, we update each probability $\Pr
\in \cP$ to $\Pr \mid E$, and take its weight to be $\alpha_{\Pr} =
\Pr(E) / \sup_{\Pr' \in \cP} \Pr'(E)$. 
%sam20 added clarification
If more than one $\Pr \in \cP$ gets updated to the same $\Pr \mid E$,
the $\sup$ of all such weights is used. 
Thus, the weight of $\Pr$ after
observing $E$ is modified by taking into account the likelihood of
observing $E$ assuming that $\Pr$ is the true probability.
We refer to this method of updating weights as \emph{likelihood updating}. 

If observations are generated by a stable measure 
%joe3
%(e.g. flipping a biased coin) then, 
(e.g., we observe the outcomes of repeated flips of a biased coin) then, 
%joe2
%as the agent makes more and more observations, then the weighted set of
as the agent makes more and more observations, the weighted set of
probabilities of the agent will, almost surely, look more 
and more like a single measure.  The weight of the measures in
%joe20
%$\cP$ closest to the measure generating the observations will converge to
$\cP$ closest to the measure generating the observations converges to
1, and the 
%joe20
%weight of all other measures will converge to 0.  This would not be
weight of all other measures converges to 0.  This would not be
the case if uncertainty were represented by a set of probability measures and
we did measure-by-measure updating, as is standard.  
%joe7:
%joe20
%As we mentioned above, this means that MWER will converge to SEU.
As we mentioned above, this means that MWER converges to SEU.

%joe7: moved from above
We provide an axiomatization for dynamic MWER with likelihood updating.
%joe7*: isn't this true?
%(which also gives as a special case an axiomation for dynamic MER with
%measure-by-measure updating, something that had not been done earlier).
We remark that 
a dynamic version of MMEU with measure-by-measure updating 
has been axiomatized by Jaffray \citeyear{Jaffray1994}, Pires
\citeyear{Pires2002}, 
and Siniscalchi \citeyear{Siniscalchi2011}.

Likelihood updating is somewhat similar in spirit to an updating method
implicitly proposed by Epstein and Schneider \citeyear{Epstein2007}.  They
also represented uncertainty by using (unweighted) sets of probability
measures.  They choose a threshold $\alpha$ with $0 < \alpha < 1$,
update by conditioning, and eliminate all measures whose 
%joe2
relative
likelihood
does not exceed the threshold.  This approach also has the property
that, {over time, all that is left in} $\cP$ are the measures closest
%joe2
%to the true distribution; all other distributions are eliminated.  
to the measure generating the observations; all other measures are
eliminated.   
However, it has the drawback that it introduces a
new, somewhat arbitrary, parameter $\alpha$.

%joe1: cut all this
\commentout{
%M-DC, short for Menu-Dependent Dynamic Consistency, connects the unconditional and conditional preferences. 
Notably, there is no mention of weights in M-DC, and yet M-DC guarantees
that the weights are likelihoods.  

Finally, in consideration of applying \emph{likelihood updating} on MER in
sequential decision making, we note that MWER with likelihood updating is
dynamically inconsistent (due to changing preferences, the path actually
followed may be suboptimal from the perspective of the initial, and
possibly every, decision point). 
This problem can be circumvented by using backward induction -- the agent is made to correctly anticipate her future choices. 
%We cite Siniscalchi's results \cite{Siniscalchi2011} to characterize preferences from consistent planning (backward induction with a specific tie-breaking rule) using MWER.
}

%joe: cut this from here; basically said it above (although I said
%much less)
\commentout{
\emph{likelihood updating} has similar properties to the learning method
Here we will consider a simplification of their method and refer to it
as Likelihood Threshold.  
Although Likelihood Threshold was designed for an extension of MMEU
called recursive multiple priors \cite{EpsteinSchneider2003}, and
specifically to learn a memoryless randomizing mechanism\footnote{The
method in \cite{Epstein2007} also deals with more complex scenarios
where there may be ambiguous elements in an environment that can never
be learned no matter how many observations are made. Here we focus on
the case where all parameters are learnable.}, the technique can also be
used to define conditional preferences for MER. 
Likelihood Threshold compares the relative likelihood of the priors, and
completely ignores a prior if its likelihood is lower than some
predefined threshold. 
proposed by Epstein and Schneider \citeyear{Epstein2007}.  

One drawback of Likelihood Threshold is the arbitrariness of the
threshold $\alpha$ -- it is unclear what $\alpha$ should be for a
particular problem, and whether it should be constant throughout an
entire problem.  
On the other hand, since a hard division is not needed, \emph{likelihood updating} does not require a specific threshold value.

\emph{likelihood updating} and Likelihood Threshold differ only when there is an intermediate amount of information. 
In the absence of information, updating is irrelevant. 
When there is abundant information, \emph{likelihood updating} results
in the convergence of all the weight to the single true prior, meaning
that all other prior distributions would be disregarded.  
The same occurs for Likelihood Threshold.
In the intermediate case, we argue that \emph{likelihood updating} is `smooth', that slight perturbations in the information result in slight perturbations in the preferences, while Likelihood Threshold may not be.  
We believe that the smoothness is a desirable quality.
}
%joe: \end{commentout}

%joe18: moved discussion of CF here
%joe18: rewrote, to give them credit for using weighted sets of
%probabilities.  
%Chateauneuf and Faro \citeyea{ChateauneufFaro2009} defines and
%xiomatizes a generalization of MMEU that parallels our generalization
%of MER.   
Chateauneuf and Faro \citeyear{ChateauneufFaro2009} also consider
weighted sets of probabilities (they model the weights using what they
call \emph{confidence functions}), although they impose more constraints on
the weights than we do.   They then define and provide a representation
of a generalization of MMEU using weighted sets of probabilities that
parallels our generalization of MER.   
%joe18: no need for this
%In particular, they use a confidence function that assigns confidence
%levels to each probability measure, just as we use weights to represent
%confidence in each probability measure.  
%Assuming that all utilities are nonnegative, the expected utility with
%respect to each probability measure is then scaled according to the
%confidence of the measure, before the maximin criterion is applied.  
%joe18: 
%Chateauneuf and Faro does not discuss the updating of the confidence
%functions and probability measures. 
Chateauneuf and Faro do not discuss the dynamic situation; specifically,
they do not consider how weights should be updated in the light of new
information.

%joe1: we should fill this in as soon as we know:
%sam2: added
The rest of this paper is organized as follows. 
Section~\ref{sec:defCP} introduces the weighted sets of probabilities representation, and Section~\ref{sec:MWER} introduces the MWER decision rule. 
Axiomatic characterizations of static and dynamic MWER are provided in Sections~\ref{sec:charMWER} and \ref{sec:charMWERL}, respectively. 
We conclude in Section~\ref{sec:conclusion}.

%joe: Again, MWER is not a ``model''.  And again, we want to separate
%the representation of uncertainty from the decision rule.
\section{Weighted Sets of Probabilities}
\label{sec:defCP}
%joe: took the next paragraph from my paper.
A set $\cP^+$ of \emph{weighted 
probability measures} on a set $S$ consists of pairs
%joe3
%$(\Pr,\alpha_{\Pr})$, where $\alpha_{\Pr}$ is a weight in $[0,1]$ and
$(\Pr,\alpha_{\Pr})$, where $\alpha_{\Pr} \in [0,1]$ and
%sam18 `probability on S' might sound incorrect to some people?
%$\Pr$ is a probability on $S$.%
$\Pr$ is a probability measure on $S$.%
\footnote{In this paper, for ease of exposition, we take the state space
$S$ to be finite, and assume that all sets are measurable.  We can
easily generalize to arbitrary measure spaces.}
Let $\cP = \{\Pr: \exists \alpha (\Pr,\alpha) \in \cP^+\}$.
We assume that,
for each $\Pr \in \cP$, there is exactly one $\alpha$ such that
$(\Pr,\alpha) \in \cP^+$.  We denote this number by $\alpha_{\Pr}$, and view
%joe3
%it as the weight of $\Pr$.
it as the \emph{weight of $\Pr$}.
We further assume for convenience that weights have 
been normalized so that there is at least one 
measure $\Pr \in \cP$ such that $\alpha_{\Pr} = 1$.%
\footnote{While we could take weights to be probabilities, and
normalize them so that they sum to 1, if $\cP$ is finite, this 
runs into difficulties if we have an infinite number of measures in
$\cP$.  For example, if we are tossing a coin, and $\cP$ includes all
probabilities on heads from  $1/3$ to $2/3$, using a uniform
probability, we would be forced to assign each individual probability
measure a weight of 0, which would not work well in the definition of
MWER.}
We remark that, just as we do, Chateaunef and Faro
\citeyear{ChateauneufFaro2009} take weights to be in the interval $[0,1]$.
They impose additional requirements on the weights.  For example, they
require that the weight of a convex combination of two probability
measures is at least as high as the weight of each one.  This does not
seem reasonable in our applications.  For example, an agent may know
that one of two measures is generating his observations, and 
give them both weight 1, while giving all other distributions weight 0.

%joe3: some glue
As we observed in the introduction, one way of updating weighted sets of
probabilities is by using likelihood updating. 
We use $\cP^+ \mid E$ to denote the result of applying likelihood
%joe2*
%updating to $\cP^+$.  Formally, if $\Pr(E) > 0$ for some $\Pr \in \cP$, we 
updating to $\cP^+$.  
%joe2*: 
Define $\ucP(E) = \sup\{\alpha_{\Pr}\Pr(E): \Pr \in \cP\}$;
%sam20: changed notation from \alpha_{\Pr \mid E} to (\alpha_{\Pr \mid
%E})\mid E 
%joe20*: this seems really clunky; you've got two ``\mid E''s!  How about
%\alpha_{\Pr,E}, \alpha_{\Pr}^E or \alpha_{\Pr}|E.  I slightly prefer
%the first or second.
%joe21: I assume that you made this change globally, although I didn't check
if $\ucP(E) > 0$, set $\alpha_{\Pr,E}= \sup_{\{\Pr'\in\cP :
\Pr'\mid E = \Pr\mid 
E\}} \frac{\alpha_{\Pr'}\Pr'(E)}{\ucP(E)}$.  Note that given a
measure  $\Pr \in \cP$, there may be
several distinct measures $\Pr'$ in $\cP$ such that $\Pr'\mid E = \Pr
\mid E$.  Thus, we take the weight of $\Pr \mid E$ to be the $\sup$ of the
possible candidate values of  $\alpha_{\Pr,E}$.  By dividing by
$\ucP(E)$, we guarantee that $\alpha_{\Pr,E} \in [0,1]$, and that
there is some measure $\Pr$ such that $\alpha_{\Pr,E} = 1$,
as long as there is some pair $(\alpha_{\Pr},\Pr) \in \cP$ such that 
$\alpha_{\Pr} \Pr(E) = \ucP(E)$.
%joe17*: the last clause is not necessary true unless we assume that
%there is always some pair $(\alpha_{\Pr},\Pr) \in \cP$ such that 
%$\alpha_{\Pr} \Pr(E) = \sup_[\Pr' \in \cP}{\alpha_{\Pr'} \Pr'(E)\}.  We
%can add this assumption (basically, it says that \cP^+ is upward closed,
%which is weaker than the assumption that we had earlier that it's
%closed.  You correctly convinced me that the latter assumption was
%inappropriate, but this weaker assumption seems more reasonable.  The
%question is whether we can show that the representation we construct in
%our representation theorem has this property.  I'm pretty sure that we
%can, but we need to check details.
If $\ucP(E) > 0$, 
we take $\cP^+ \mid E$ to be 
%sam20 : added |E
$$\{(\Pr\mid E, \alpha_{\Pr,E} ): \Pr \in \cP\}.$$ 
If $\ucP(E) = 0$, then $\cP^+ \mid E$
is undefined.    

%joe2*: moved the technical claim about order of updating (which is
%interesting, but comes out of the blue, after the intuition.  The
%technical claim itself also need more motivation 

%joe2: Added intuition, 
In computing $\cP^+ \mid E$, we update not just the probability
measures in $\cP$, but also their weights.  The new weight combines the
old weight with the likelihood.  
%joe2: no need for italics after the first time
%Clearly, if all priors assign the same probability to the event $E$,
%then \emph{likelihood updating} and measure-by-measure updating coincide.  
Clearly, if all measures in $\cP$ assign the same probability to the
event $E$, then likelihood updating and measure-by-measure
updating coincide.   
%joe2
%This is expected, since such an observation $E$ does not give us 
This is not surprising, since such an observation $E$ does not give us 
%joe2
%information about which prior is more likely than the others. 
information about the relative likelihood of measures. 
%sam: how about measure-by-measure updating? Sometimes
%appropriate even when likelihood updating isn't? 
%joe2: not clear; note, by the way, that I removed the italics.  You
%should only italicize the first occurrence
%We stress that using \emph{likelihood updating} is appropriate only if  
We stress that using likelihood updating is appropriate only if  
%joe2
%the true measure (which may be a point mass on a certain state) is
the measure generating the observations is
assumed to be stable.   
%joe2
%For example, \emph{Likelihood updating} is not appropriate when a coin
%of potentially \emph{different} bias is selected by an adversary at
%every round. 
For example, if observations of heads and tails are generated by coin
tosses, and a coin of possibly different bias is tossed in each round,
then likelihood updating would not be appropriate.

%joe2: if it's going to have a formal proof, we should label it a
%proposition.  I'm not sure that this result is necessary, but I don't
%mind including it.  But if we include, it needs more motivation!  I've
%added some here.
%joe3*: if we need space, this material (to the end of the section) can
%be cut.
%joe11
%It is well known that, when performing conditioning on a single probability
It is well known that, when conditioning on a single probability
measure, the order that information is acquired is irrelevant; the same
observation easily extends to sets of probability measures.  As we now
show, it can be further extended to weighted sets of probability measures.
\begin{proposition}
%joe2
%Likelihood updating is consistent in the sense that for any
Likelihood updating is consistent in the sense that for all
$E_1,E_2\subseteq S$, $(\cP^+ \mid E_1)\mid E_2 = (\cP^+ \mid E_2)\mid
E_1 = \cP^+ \mid (E_1\cap E_2)$,
%joe2
provided that  $\cP^+ \mid (E_1\cap E_2)$ is defined.
\end{proposition}
\begin{proof}
By standard results, $(\Pr\mid E_1)\mid E_2 = (\Pr\mid E_2)\mid E_1 =
\Pr\mid (E_1\cap E_2)$. 
%joe8
%Since the weight for each $(\Pr\mid E_1)$ is proportional to
%$\alpha_{\Pr}\Pr(E_1)$, the weight for each $(\Pr\mid E_1)\mid E_2$ is
Since the weight of the measure $\Pr\mid E_1$ is proportional to
$\alpha_{\Pr}\Pr(E_1)$, the weight of $(\Pr\mid E_1)\mid E_2$ is
%joe2: this is hard to parse.  
%proportional to $\alpha_{\Pr}\Pr(E_1)\Pr\mid E_1(E_2) =
proportional to $\alpha_{\Pr}\Pr(E_1)\Pr(E_2\mid E_1) =
\alpha_{\Pr}\Pr(E_1\cap E_2)$.  
%joe8
%Likewise, the weight for each $(\Pr\mid E_2)\mid E_1$ is proportional to
Likewise, the weight of $(\Pr\mid E_2)\mid E_1$ is proportional to
%joe2
%$\alpha_{\Pr}\Pr(E_2)\Pr\mid E_2(E_1)= \alpha_{\Pr}\Pr(E_1\cap E_2)$.  
$\alpha_{\Pr}\Pr(E_2)\Pr(E_1 \mid E_2)= \alpha_{\Pr}\Pr(E_1\cap E_2)$.  
%joe8
%Since in all these cases the largest weight is normalized to be $1$,
%all corresponding weights for $\cP^+\mid (E_1\cap E_2)$, $(\cP^+ \mid
%E_1)\mid 
Since, in all these cases, the $\sup$ of the weights is normalized to $1$, the
%corresponding weights for $\cP^+\mid (E_1\cap E_2)$, $(\cP^+ \mid E_1)\mid
weights of corresonding measures in $\cP^+\mid (E_1\cap E_2)$, $(\cP^+
\mid E_1)\mid 
E_2$ and $(\cP^+ \mid E_2)\mid E_1$ must be equal.  
\end{proof}

%joe: misplaced.  We don't need the AA framework to define MWER, and
%shouldn't assume it until we need.  You're jumping too quickly to the
%axiomatization.  SLOW DOWN and introduce the notions
\section{MWER}\label{sec:MWER}

We now define MWER formally.  
Given a set $S$ of states and a set $X$ of outcomes,
%joe2*: NO!  We do *not* want to introduce lotteries here (you'll notice
%that none of the definitions in this section use them)
%let $Y$ be the set of measures over $X$. 
an \emph{act} $f$ (over $S$ and $X$) is a function mapping $S$ to $X$.
%to lotteries in $Y$.  Thus, the set $L$ of acts is $Y^S$.  
%joe2: I'm wiling to allow X to consist of all lotteries (so
%Anscombe-Aumann becomes a special case).  But the it isn't finite.
%For %simplicity in this paper, we take $S${ and $X$ }to be finite.
For simplicity in this paper, we take $S$ to be finite.
%joe2: removed paragraph break
%
%joe: this is nontstandard (although it's been done this way before).
%It's perhaps better to take X to consists of
%lotteries.  That will give you Anscome-Aumann, for example.  
%distributions $P_{a,s}\in \Delta(X)$.  
%joe2: used u
%Associated with each outcome $x\in X$ is a utility: $U(x)$ is the
Associated with each outcome $x\in X$ is a utility: $u(x)$ is the
utility of outcome $x$.  We call a tuple $(S,X,u)$ a
\emph{(non-probabilistic) decision problem}.
%joe2: again, this is premature
%Therefore each lottery $y\in Y$ also has a utility: $u(y) = \sum_{x\in
%X}U(x)y(x)$.  
%joe2: removed paragraph break
%
%The final component in a decision problem is a finite, nonempty menu
%$M\subseteq A$ of feasible acts. 
%joe2*: this should be rewritten.  We don't want M to be part of the
%decision problem, since we want to get a family of orderings
To define regret, we need to assume that we are also given a set $M
%\subseteq L$ of feasible acts, called the menu.  
\subseteq X^S$ of feasible acts, called the \emph{menu}.
%joe2*: again, I we shouldn't assume AA here.  Cut all this
%Although the menu is finite, the decision maker can choose any mixture
%of acts within the menu. 
%Mixtures of acts, for example $pa+(1-p)b$, is generated by performing
%$a$ with probability $p$ and $b$ otherwise. 
%joe2: added; we need to explain this here, not in the next section.
The reason for the menu is that, as is well known (and we will
%joe3
%demonstrate by example shortly) the regret can depend on the menu.
demonstrate by example shortly), regret can depend on the menu. %sam4
%Thus, regret is defined relative to a tuple
%$(S,X,u,M)$; we call such a tuple a \emph{(non-probabilistic) decision
%problem}. 
%joe1*: again, you're russhing too fast
%joe2: removed paragrap break
%
%The regret of an act $f\in L$ with respect to $(S,X,u,M)$ is 
%joe2*: Do you want/need to assume that f is in M?
%sam: I think so.
%joe3: ``any'' is ambiguous, and should be avoided
%For any menu $M$ and act $f\in M$, the regret of $f$ with respect to
Moreover, we assume that every menu $M$ has utilities bounded from
%joe11: this shouldn't be a footnote
%above.\footnote{ 
%We assume that for every menu $M$, $\max_{s\in S} \sup_{f\in M}u(f(s))$
above.  That is, 
we assume that for all menus $M$, $\sup_{g\in M}u(g(s))$
is finite.  This ensures that the regret of each act is well defined.%
%joe11: 
%This is a weaker assumption than the one made in Stoye
%\citeyear{Stoye2011},
%in which each menu $M$ is the convex hull of a finite number of acts.} 
%in which each menu $M$ is the convex hull of a finite number of acts.} 
\footnote{
Stoye \citeyear{Stoye2011} assumes that, for each menu $M$, there is a
finite set $A_M$ of acts such that $M$ consists of all the convex
combinations of the acts in $A_M$.  Our assumption is clearly much
weaker than Stoye's.}
For a menu $M$ and act $f\in M$, the regret of $f$ with respect to
$M$ and decision problem 
$(S,X,u)$ 
%joe17: added next 5 lines
in state $s$ 
is
$$\regret_M(f,s) = \left(\sup_{g\in M}u(g(s))\right) - u(f(s)).$$
That is, the regret of $f$ in state $s$ (relative to menu $M$) is the
difference between $u(f(s))$ and the highest utility possible in state
$s$ (among all the acts in $M$).
%joe17: 
The regret of $f$ with respect to $M$ and decision problem $(S,X,u)$ is
the worst-case regret over all states:
$$
%joe2: all sentences need to end with a period, even if the sentence
%ends with a displayed formula
%\max_{s\in S}\left( \max_{g\in M}u(g(s)) - u(f(s)) \right)
%joe17: simplified 
%\max_{s\in S}\left( \sup_{g\in M}u(g(s)) - u(f(s)) \right).
\max_{s\in S}\regret_M(f,s).
%\end{align}
$$
%joe2: added
%joe3
%We usually denote this as $\regret_M^{(S,X,u)}(f)$.  We usually omit the
We denote this as $\regret_M^{(S,X,u)}(f)$,  and usually omit the
superscript $(S,X,u)$ if it is clear from context.
%joe2
%
If there is a probability measure $\Pr$ over the states, then we can
consider the \emph{probabilistic decision problem} $(S,X,u,\Pr)$.  
%joe2
%expected regret of $f\in L$ with respect to the 
The \emph{expected regret} of $f$ with respect to $M$ is
%joe2: this is a definition, so should be italicized
%probabilistic decision problem $(S,X,u,M,\Pr)$ is 
%joe2: again
%\begin{align}
%\sum_{s\in S}\Pr(s)\left( \max_{h\in M}u(h(s)) - u(f(s)) \right)
$$
\regret_{M,\Pr}(f) = 
%joe17
%\sum_{s\in S}\Pr(s)\left( \sup_{g\in M}u(g(s)) - u(f(s)) \right).
\sum_{s\in S}\Pr(s)\regret_M(f,s).
%\end{align}
$$
%joe2: 
%
%Moreover, if there is a set $\cP$ of probability measures over the
If there is a set $\cP$ of probability measures over the
%joe3: removed %
%states, then we consider the $\cP$-decision problem %$(S,X,u,\cP)$.
states, then we consider the $\cP$-decision problem $(S,X,u,\cP)$.
The maximum expected regret of $f\in M$ with respect to $M$ and 
%sam4 changed to /in M from /in X^S 
%joe2
%$(S,X,u,M,\cP)$ is 
$(S,X,u,\cP)$ is 
$$
%\begin{align}
\regret_{M,\cP}(f) = \sup_{\Pr\in \cP} \left( \sum_{s\in S}\Pr(s)
%joe17
%\left( \sup_{g\in M}u(g(s)) - u(f(s)) \right)  \right). 
\regret_M(f,s)   \right). 
%\end{align}
$$
%joe2
%
Finally, if beliefs are modeled by weighted probabilities $\cP^+$, then
we consider the $\cP^+$-decision problem $(S,X,u,\cP^+)$.
The maximum weighted expected regret of $f\in M$ with respect to 
%sam4
%changed to /in M from /in X^S 
%joe1
%the \emph{$\cP^+$ decision problem} $(S,X,u,M,\cP^+)$ is 
$M$ and $(S,X,u,\cP^+)$ is 
%joe2
%\begin{align}
$$
%joe11: formatted for UAI
%\begin{array}{ll}
%joe17
%&\regret_{M,\cP^+}(f) \\
%= &\sup_{\Pr\in \cP}\left( \alpha_{\Pr} \sum_{s\in
%S}\Pr(s)\left( \sup_{g\in M}u(g(s)) - u(f(s)) \right) \right).
\regret_{M,\cP^+}(f) 
= \sup_{\Pr\in \cP}\left( \alpha_{\Pr} \sum_{s\in
S}\Pr(s)\regret_M(f,s)  \right).
%joe11
%\end{array}
$$

%joe2
%The MER decision rule is thus defined for all $f,g\in L$ as
%\begin{align}\tag{MER}
%& f\succeq_{M,\cP} g \Leftrightarrow \\
%\notag & \max_{\Pr\in \cP} \left( \sum_{s\in S}\Pr(s)\left( \max_{h\in
%M}u(h(s)) - u(f(s)) \right)  \right) \\ 
%\notag & \leq \max_{\Pr\in \cP} \left( \sum_{s\in S}\Pr(s)\left(
%joe2
%\max_{h\in M}u(h(s)) - u(g(s)) \right)  \right) 
%\end{align}

%\max_{h\in M}u(h(s)) - u(g(s)) \right)  \right).
The MER decision rule is thus defined for all $f,g\in X^S$ as
$$
f\succeq_{M,\cP}^{S,X,u} g \mbox{ iff } \regret_{M,\cP}^{(S,X,u)}(f) \le
\regret_{M,\cP}^{(S,X,u)}(g). 
$$ 
That is, $f$ is preferred to $g$ if the maximum expected regret of $f$
is less than that of $g$.  We can similarly define
$\succeq_{M,\regret}$, 
%joe3
%(the order inducing by nonprobabilistic regret), 
%sam11 removed () around S,X,u
$\succeq_{M,\Pr}^{S,X,u}$, and
$\succeq_{M,\cP^+}^{S,X,u}$ by replacing $\regret_{M,\cP}^{(S,X,u)}$ by
$\regret_M^{(S,X,u)}$, 
$\regret_{M,\Pr}^{(S,X,u)}$, and $\regret_{M,\cP^+}^{(S,X,u)}$,
respectively.  Again, we usually omit the superscript $(S,X,u)$ 
%sam8: added
and subscript $\Pr$ or $\cP^+$, and just write $\succeq_M$, if it is
clear from context. 

%joe2*: now unnecessary ...
\commentout{
And the MWER decision rule is defined for all $f,g\in L$ as 
%\begin{align}\tag{MWER}
$$\begin{array}{{}
& f\succeq_{M,\cP^+} g \Leftrightarrow \\
\notag & \max_{(\Pr,\alpha_{\Pr})\in \cP^+}\left( \alpha_{\Pr} \sum_{s\in S}\Pr(s)\left( \max_{h\in M}u(h(s)) - u(f(s)) \right) \right)\\
\notag & \leq \max_{(\Pr,\alpha_{\Pr})\in \cP^+}\left( \alpha_{\Pr}
%joe2
%\sum_{s\in S}\Pr(s)\left( \max_{h\in M}u(h(s)) - u(g(s)) \right) \right)  
%\end{align}
\sum_{s\in S}\Pr(s)\left( \max_{h\in M}u(h(s)) - u(g(s)) \right)
\right).
$$  
}}
\commentout{
With likelihood updating, the conditional MWER preferences are defined as:
\begin{definition}{(MWER with \emph{likelihood updating})} 
\begin{align}\tag{MWER}\label{eq:MWER}
&f\succeq_{M,E} g \\
\notag & \Leftrightarrow \max_{P\in\mathscr{P}}\left( \alpha_{P} P(E)\sum_{s\in E}\left[ \left(\max_{a' \in M}u(a'(s)) - u(f(s))\right) P(s|E) \right] \right) \\
\notag &\leq \max_{P\in\mathscr{P}}\left(\alpha_{P} P(E) \sum_{s\in E}\left[\left(\max_{a' \in M}u(a'(s)) - u(g(s)) \right)P(s|E) \right] \right)
\end{align}
\end{definition}
}

To see how these definitions work, consider
%joe3
%the quality-control example from the introduction, there are $1,000$
the delivery example from the introduction.  There are $1,000$
%joe3
%states with 1 broken cake, and $C(1000,10)$ states with 10 broken
states with one broken cake, and $C(1000,10)$ states with ten broken
cakes.  The regret of each action in a state depends only on the number
of broken cakes, and is given in Table~\ref{tab:qcRegrets}.  It is easy
to see that the action that minimizes regret is $\chec$, with $\cont$
and $\back$ having equal regret.  If we
represent uncertainty using the two probability measures $\Pr_1$ and
$\Pr_{10}$, the expected regret of each of the acts with respect to $\Pr_1$
(resp., $\Pr_{10}$) is just its regret with respect to states with one
(resp. ten) broken cakes.  Thus, the action that minimizes maximum expected %sam4 added "maximum" expected 
regret is again $\chec$.

%\ref{tab:qcRegrets}. 
\begin{table}
	\centering
		\begin{tabular}{|l|l|l|l|l|l|l|} \hline
&	\multicolumn{2}{|c|}{$1$ broken cake}		& \multicolumn{2}{|c|}{$10$ broken cakes}\\ \hline
& Payoff	& Regret & Payoff & Regret \\ \hline
			$\cont$  			&10,000 	& 0	 					&-10,000			& 10,000				\\ \hline
			$\back$ 				& 0 			& 10,000			&0				& 0 \\ \hline 
			$\chec$				&5,001		&	4,999				&-4,999			& 4,999 \\ \hline 
		\end{tabular}
%joe17
%	\caption{Payoffs and regrets for delivery example}
	\caption{Payoffs and regrets for delivery example.}
	\label{tab:qcRegrets}
\end{table}

%joe2: cut; I don't know what it means that ``either prior turns out to
%be true''
%Since the maximum (unweighted) regret of test is the lowest, a MER agent
%would choose to $\chec$, which minimizes regret if either prior turns out
%to be true.  
%joe2: cut this; I'm not sure what point you're trying to make
%If the penalty for delivering too many broken cakes is lowered
%(below $4,999$), then the regret from delivering a bad batch is less than
%that of wasting utility to test a good batch, and the MER agent would
%instead deliver without testing. 

%joe2*: I cut this paragraph.  It breaks the flow, and we could equally
%well write paragraphs pointing out problems with MER compared to MMEU.
%I wouldn't get into that 
\commentout{
As an aside, this reaction towards change in penalty seems more
reasonable than that of MMEU.  
For MMEU, the expected utilities associated with each measure is simply the utilities associated with each class of states.
A MMEU agent would choose to sell the batch at $\back$ price, since this option maximizes the minimum expected utility over the two measures. 
However, the preference for selling the batch at $\back$ price remains even if the penalty for delivering a batch of broken cakes is decreased to merely $\epsilon$ -- 
the MMEU agent doesn't seem to care about the lost opportunity to gain $10,000$ utility, while we, and the MER agent, would certainly care.
}

%%%%%%%% menu dependence
%These regret-based decision rules are menu-dependent -- they are
%sensitive to changes in the menu $M$.  

%joe2: rewrote
As we said above, the ranking of acts based on MER or MWER can change if
the menu of possible choices changes.
%For example, in the quality-control problem, $\chec$ has lower maximum
%expected regret ($4999$) than $\cont$ ($10,000$), so MER prefers $\chec$
%over $\cont$.  
%However, if we introduce a \emph{new choice} that gives twice as much
%gains and losses as $\cont$, resulting in the setting described in
For example, suppose that we introduce a new choice in the
delivery problem, whose gains and losses are twice those of $\cont$,
resulting in the payoffs and regrets described in %sam4
%resulting in the setting described in
Table \ref{tab:menuDepEx}. 
%joe2
%the regrets of the choices become those
%depicted in Table \ref{tab:menuDepReg}.  
%$\cont$ will have lower maximum expected regret ($10,000$) than $\chec$
%The regret of each of the acts in the new setting is given in 
%Table \ref{tab:menuDepEx}.  $sam4: redundant
In this new setting,
%Table \ref{tab:menuDepReg}.  In this new setting,
$\cont$ has a lower maximum expected regret ($10,000$) than $\chec$
($14,999$), so MER prefers $\cont$ over $\chec$.  Thus, the introduction
of a new choice can affect the relative order of acts according to MER
(and MWER), even though other acts are preferred to the new choice. %i.e. it's an irrelevant alternative
By way of contrast, the decision rules MMEU and SEU are
\emph{menu-independent}; the relative order of acts according to MMEU
and SEU is not affected by
the addition of new acts.

\begin{table}
	\centering
		\begin{tabular}{|l|l|l|l|l|l|l|} \hline
&	\multicolumn{2}{|c|}{$1$ broken cake}		& \multicolumn{2}{|c|}{$10$ broken cakes}\\ \hline
& Payoff	& Regret & Payoff & Regret \\ \hline
		$\cont$  				& 10,000 &10,000		& -10,000		&10,000		\\ \hline 
			$\back$ 				& 0			 &20,000		& 0 			&0\\ \hline 
			$\chec$				&	5,001		&14,999		& -4,999 	& 4,999\\ \hline 
			$new$					& 20,000 & 0 				& -20,000 & 20,000 \\ \hline
		\end{tabular}
\caption{Payoffs and regrets for the delivery problem with a new choice
%joe17
%		added}  
		added.}  
	\label{tab:menuDepEx}
\end{table}

%sam3: combined tables
%\begin{table}
%	\centering
%		\begin{tabular}{|l|l|l|l|l|l|l|} \hline
%%joe2
%%		&	$1$ defect		& $10$ defects \\ \hline
%&	$1$ broken cake		& $10$ broken cakes \\ \hline
%			$\cont$  				& 10,000	 								& 10,000				\\ \hline
%			$\back$ 				& 20,000							& 0 \\ \hline 
%			$\chec$						&	14,999							& 4,999 \\ \hline 
%  		$new$	& 0 & 20,000 \\ \hline
%		\end{tabular}
%	\caption{Regrets for quality control example with a new choice added}
%	\label{tab:menuDepReg}
%\end{table}

%On the other hand, decision rules like maximizing expected utility and
%MEU rank acts independently of the menu.  
%joe2*: 
%joe3*: cut this for abstract; it's a bit of a distraction
\commentout{
Interestingly, we can define a
menu-independent version of MER (or MWER) in one special case of
interest: if there is a ``best'' act; that is, if there is an act $f^*$
such that $f^*(s)$ is at least as good as $g(s)$ for all states $s$ and
acts $g$.  In this case, we can compute the regret of an act $f$ with
respect to $f^*$.  Taking $\regret_a(f) = \max_{s\in S}(u(f^*(s) - u(f(s)) ) )$ gives us an absolute notion of regret.  
(The subscript $a$ stands for ``absolute''.)  We can analogously define
$\regret_{a,\Pr}$, $\regret_{a,\cP}$, and $\regret_{a,\cP^+}$.  Note that 
$\regret_a(f) = \regret_{M}(f)$ if $M$ is a menu that includes $f^*$ and
$f$.  Similarly, for all the other variants, the absolute notion of
regret is the same as the regret with respect to a menu that includes $f^*$.
Conceptually, in a setting where a best possible act $f^*$ exists, we
can think of the agent as always comparing his outcome to the best possible
outcome, what he would have gotten had he been able to do $f^*$ (which
for some reason may not be a feasible act).  
}

%%%%%%%%% updating 
%joe2: rewrote added some more glue
We next consider a dynamic situation, where the agent acquires
information.  Specifically, in the context of the delivery problem, %sam4: added "the"
suppose that T-800 learns $E$---the first $100$ items are good.
%Next, let's look at MWER conditional preferences with \emph{likelihood
%updating}.  
Initially, suppose that T-800 has no reason to believe that one
hypothesis is more likely than the other, so assigns both hypotheses
weight $1$.  
%However, if we were told that the first $100$ items are good, we would
%probably feel that it is now more likely that there is only $1$
%broken cake, rather than $10$.  
%Intuitively, $P_{1}$ explains the event better than $\Pr_{10}$, since it
%is unlikely that the first 100 items are all good, if there are 10
%broken cakes out of 1,000.
%Quantitatively, 
Note that 
%joe8
%$P_{1}(E) = 0.9$, $\Pr_{10}(E) = C(900,10)/C(1000,10) \approx 0.35$.  
$P_{1}(E) = 0.9$ and $\Pr_{10}(E) = C(900,10)/C(1000,10) \approx 0.35$.  
%joe2: simplified discussion;  I think this is overkill
\commentout{
Now suppose the agent's initial beliefs were a set of weighted probabilities. 
It is natural for this set to be $Pr_1$ and $Pr_2$ with weight 1 on both
measures. 
If the agent were to use only measure-by-measure updating, conditioning
each measure on $E$, then the regrets are those shown in Table
\ref{tab:qcWeightedRegrets}.  
The MWER (or MER, since measures are effectively unweighted) agent would prefer paying for the $\chec$, since it has the lowest maximum expected regret ($4999$) across $Pr_1$ and $Pr_2$. 
\begin{table}
	\centering
		\begin{tabular}{|p{0.4cm}|l|l|l|l|l|l|} \hline
											&		\multicolumn{3}{|c|}{$\Pr_{1}|E$} 												& \multicolumn{3}{|c|}{$\Pr_{10}|E$ 	} \\ \hline
											&	\small{max reg}	& \small{$\Pr_{1}(E)$}	& \small{weighted}	& \small{max reg}	& \small{$\Pr_{10}(E)$} & \small{weighted }\\ \hline
			\small{$\cont$}  				& 0 					& $ 0.9$								& 0 								& 10,000& $ 0.35$				& 3500						\\ \hline
			\small{$\back$} 				& 10,000			& $0.9$									& $ 9,000$						& 0 			& $ 0.35$				& 0								\\ \hline
			\small{$\chec$	}				&	4999	& $0.9$								& $\approx 4500$			& 4999	& $ 0.35$				& $\approx 1750$	\\ \hline
		\end{tabular}
%joe17
%	\caption{Unweighted and weighted regrets}
	\caption{Unweighted and weighted regrets.}
	\label{tab:qcWeightedRegrets}
\end{table}

Now suppose the beliefs are updated using \emph{likelihood updating} instead. 
As seen in Table \ref{tab:qcWeightedRegrets}, once the regrets are
weighted by the relative likelihoods of the measures, the maximum regret
for $\cont$ is less than the maximum regret for $\chec$. 
Therefore, the MWER agent would prefer to sell the batch at $\cont$ price. 

Conditional on the first $100$ items being good, the choice made by the MWER agent when \emph{likelihood updating} is used seems more intuitive than the choice resulting from measure-by-measure updating. 
Knowing that some of the items are good increases our confidence that
the batch has fewer broken cakes.  
}% end \commentout
%joe8: removed paragraph break
%
%Thus $\cP^+ \mid E = ((\Pr_1\mid E, 1), (\Pr_{10} \mid E, C(900,10)/(.9
Thus, $$\cP^+ \mid E = \{({\Pr}_1\mid E, 1), ({\Pr}_{10} \mid E, C(900,10)/(.9
C(1000,10))\}.$$ 
%joe8: First, I suspect you mean the weight of \Pr_{10}|E, not \Pr_{10};
%also, given your computation above, I think you mean 0.4, not 0.04.
%Finally, I'm not even sure that is worth saying. 
%where the weight of $\Pr_{10}$ is approximately $0.04$. 

We can also see from this example that MWER interpolates between MER and
expected utility maximization.  
%joe3
%Suppose the worker tells the agent that the first $N$ units out of $1,000$
Suppose that a passer-by tells T-800 that the first $N$ cupcakes are good.  
If $N=0$, MWER with initial weights $1$ is the same as MER. 
On the other hand, if $N\geq 991$, then the likelihood of $\Pr_{10}$ is
$0$, and the only measure that has effect is $\Pr_{1}$, which means
minimizing maximum weighted expected regret is just maximizing expected
utility with respect to $\Pr_{1}$.  
If $0 < N < 991$, then the likelihoods (hence weights) of $\Pr_1$ and
%joe2: we need to normalize; I'm also not sure where .99^N comes from
%$\Pr_{10}$ are $\frac{1000-N}{1000}$ and $\frac{C(1000-N,10)}{C(1000,10)}
%\approx 0.99^N$ respectively.  
%Thus, as $N$ increase, the weight on $P_1$ relative to that of $\Pr_{10}$
%increases. 
$\Pr_{10}$ are $1$ and $\frac{C(1000-N,10)}{C(1000,10)}\times
\frac{1000}{1000-N} < 
%\frac{1000-N}{1000} <  sam2: should be the reciprocal
((999-N)/999)^9$.
%((1000-N)/1000)^N$. sam2
%joe3: removed para break
%
Thus, as $N$ increases, the weight of $\Pr_{10}$ goes to 0, while the weight of
$\Pr_1$ stays at 1.  

\section{An axiomatic characterization of MWER}\label{sec:charMWER}

%joe2*: Now you need to SLOW DOWN!  The typical AAMAS reader will have
%no clue what an AA style framework is!
%joe11
%We now show how MWER can be characterized axiomatically, by means of a
%representation theorem.  That is, we are looking for a collection of
We now provide a representation theorem for MWER. That is, we provide 
a collection of properties (i.e., axioms)
that hold of MWER such that a preference order on acts
that satisfies these properties can be viewed as arising from MWER.
%joe11
%To get such a characterization, we restrict to what is known in the 
To get such an axiomatic characterization, we restrict to what is known
in the  
%joe3
%literature as the \emph{Anscombe-Aumann} framework
literature as the \emph{Anscombe-Aumann} (AA) framework
\cite{AnscombeAumann1963}, where outcomes are restricted to lotteries.  
%sam15 I think we mean 'decision theory' 
This framework is standard in the decision theory literature; axiomatic
%This framework is standard in the game theory literature; axiomatic
characterizations of SEU \cite{AnscombeAumann1963}, MMEU
\cite{GilboaSchmeidler1989}, and MER \cite{Hayashi2006,StoyeRegret} have
%joe3
%already been obtained in this framework; we draw on these results to
already been obtained in the AA framework.  We draw on these results to
obtain our axiomatization.
%joe9: It's strange to say this before we've even mentioned the axioms
%Axioms 1-9, 11, 12 are from existing literature.

%We use an Anscombe Aumann \cite{AnscombeAumann1963} style framework,
%where \emph{subjective} uncertainty is represented by having a set of
%states, and \emph{objective} uncertainty is represented by
%\emph{lotteries}, which are probability measure over outcomes.  
%The acts in this framework are functions mapping states to lotteries.  
%These acts are also called horse-roulette acts -- they combine
%subjective uncertainty (result of a horse race) with objective
%uncertainty (result of a roulette spin).  
%sam3 allow arbitrary $Y$
Given a set $Y$ (which we view as consisting of \emph{prizes}), 
%Given a finite set $Y$ (which we view as \emph{prizes}), 
a \emph{lottery} over $Y$ is just a probability with finite support on $Y$.  Let $\Delta(Y)$
%joe3
%consist of all probabilities over $Y$.  In the Anscombe-Aumann
consist of all finite probabilities over $Y$.  In the AA
framework, the set of outcomes has the form $\Delta(Y)$.  
So now acts are functions from $S$ to $\Delta(Y)$.
%joe9
(Such acts are sometimes called \emph{Anscombe-Aumann acts}.)
We can think of a lottery as modeling objective uncertainty, while a
probability on states models subjective uncertainty; thus, in the AA
framework we have both objective and subjective uncertainty.
The technical advantage of considering such a set of outcomes is that we
can consider convex combinations of acts.  If $f$ and $g$ are acts,
define the act $\alpha f + (1-\alpha)g$ to be the act that maps a state
$s$ to the lottery $\alpha f(s) + (1-\alpha)g(s)$.  
%joe2*: Why do we need this assumption.  All the axioms for MWER hold
%without it.  Since not making this assumption makes the axioms
%stronger, it can't hurt not to assume it.
%sam2: I think the current statements of Mixture Continuity (Archimedean
%property) and Ambiguity Aversion would not be meaningful without
%assuming that all mixtures of acts are allowed.  
%joe3*: Of course, we have to assume that the relevant acts are in M,
%but that's weaker than assuming that M is convex.  Does Stoye assume
%that M is convex?  Let's discuss this.  In any case, reinstated earlier
%version.  There's no logical reason to insist that the agent should be
%able to choose convex combinations (although I suppose that you can
%justify this by saying that the agent can always randomize) 
%Although the menu is finite, the decision maker can choose any mixture
%of acts within the menu. 
%Hence the actual set of feasible acts are closed under convex combinations. 
%sam3: removed assumption that agent can randomize arbitrarily, so commented out the following few lines.
%For ease of exposition, we follow Stoye \cite{StoyeRegret} and restrict
%to menus that are closed under convex combinations, so that if $f$ and
%$g$ are acts in $M$, then so is $pf + (1-p)g$ for all $p \in [0,1]$.
%joe3:
%(We are implicitly assuming that the agent can always randomize between
%feasible acts.)

In this setting, we assume that there is a utility function $U$ on
prizes in $Y$.  The utility of a lottery $l$ is just the expected
%sam3: sum over support only
utility of the prizes obtained, that is, $$u(l) = \sum_{\{y \in Y \colon
l(y)>0\}} l(y) 
%utility of the prizes obtained, that is, $$u(l) = \sum_{y \in Y} l(y)
%joe3
%u(y).$$  This makes sense since $l(y)$ is the probability of getting
U(y).$$  This makes sense since $l(y)$ is the probability of getting
prize $y$ if lottery $l$ is played.  The expected utility of an act $f$
with respect to a probability $\Pr$ is then just $u(f) = \sum_{s \in S}
\Pr(s) u(f(s))$, as usual.  
%joe8: removed paragraph break
%
%Moreover, we assume that there are at least two prizes $y_1$ and $y_2$
We also assume that there are at least two prizes $y_1$ and $y_2$
in $Y$, with different utilities $U(y_1)$ and $U(y_2)$. 

%sam13 
%Given $U$ and a set weight set $\cP^+$, we can 

%sam3: removed assumption of best and worst act, as per discussion
%joe3
%Note that in this framework there is a best (and a worst) act.  
%Note that in the AA framework, there is a best (and a worst) act.  
%joe3
%Clearly there is a best and worst lottery: the lottery that places
%Since we have assumed that $Y$ is finite, there is clearly a best and a
%worst lottery: the lottery that places probability 1 
%on the prize with the highest (resp., lowest) utility.   
%A constant act maps each state in $s$ to the same lottery.  
%Thus, the best (resp., worst) act is the constant act that maps all
%states $s$ to the best (resp., worst) lottery.
%joe3: added
%joe8: moved below: it's not good to have lots of short, choppy
%paragraphs.  I'm not sure it's all that standard.  
%As is standard in the literature, we use $l^*$ to denote a constant act
%that maps all states to $l$. 

%joe2: cut all this; AAMAS folks won't know what Savage-style acts are,
%and decision theorists already know why we work with AA acts.  In any
%case, AA acts are a special case of Savage acts, as I pointed out above.
\commentout{
We work with an Anscombe Aumann style framework (rather than Savage
style acts that maps states to outcomes) because we would like to encode
both subjective and objective probabilities.  
In particular, by pushing objective probabilities down to the level of the lotteries, we can model an agent that resolves objective (known) probabilities by maximizing expected utility, but is pessimistic when dealing with subjective (unknown) probabilities. 
For example, if beliefs are modeled by a set of weighted probabilities,
there are three levels of uncertainty.  
The lowest level are known, objective probabilities of each lottery over the outcomes. 
It happens that a subset of the axioms that characterize MER also guarantees that the agent's preferences over constant acts (which can be treated as lotteries) are representable by maximizing expected utility with respect to a subjective utility function.
The second level are subjective probabilities represented by the measures over the states. 
Finally, the weights on the measures represent nonprobabilistic confidence levels on each measure. 
The axiomatizations of MER provided by Hayashi \citeyear{Hayashi2006} and
%joe3
%Stoye \cite{StoyeRegret} are also in Anscombe Aumann style frameworks.  
Stoye \citeyear{StoyeRegret} also use the AA framework.  

%joe2: this is redundant
There is a finite set $S$ of states of the world, a finite set $X$ of
pure outcomes, and $Y$ is the set of all lotteries over $X$. 
%joe2* If you want to assume that S and X have at least two elements,
%then you need axioms that force this; it's anot an assumption that you
%make upfront.  In fact, Axiom 3 gives you the nontriviality you need.
$S$ and $X$ each have at least two elements. 
{Since $X$ is finite, there is a best and worst outcome.}
We also assume that there is a topology on $Y$, so that we can define a notion of convergence.
There is a set $L$ of possible acts. 
An act in $L$ is a function mapping each state in $S$ to a lottery in $Y$. 
An act $f$ is \emph{constant} if $f(s)$ is constant over $S$. 
The agent can choose from a nonempty, finite menu $M\subseteq L$, and can also randomize: for example $pf + (1-p)g$ is the act generated by performing $f$ with probability $p$ and performing $g$ with probability $1-p$.
More precisely, this act maps each state $s$ to the lottery $pf(s) +
(1-p)g(s)$.

%joe:
%Point out that our assumptions imply that there's a maximum and minimum
%act.  
In this section we provide an axiomatization of the MWER decision rule. 
That is, we show that any family of binary relations
$\{\succeq_{M}\}_{M\subseteq L}$ on $L$ satisfy the following axioms if and
only if the family can be represented as preferences induced by MWER. 
Our axioms are based on those from Stoye \citeyear{StoyeRegret,Stoye2011}. 
All acts appearing in the axioms without quantification are universally quantified. 
%joe: all of this discussion will be meaningless to most readers.
%Either say less (just treat the whole thing as a black box) or say
%more.  A reasonable compromise might be to discuss one or two axioms in
%some detail to give the reader a sense of what the axioms look like.
%But just giving names and saying they're common is not useful.  
%Sam: add more explanation if we need to fill up space
%Firstly, there are axioms that capture the `minimax principle', which
%are common for both MMEU and MER.  
%An example of these axioms are `Independence of never strictly optimal
%alternatives', which means adding an act that is not strictly better
%than all acts in any state to the menu does not affect the ranking of
%the acts in the menu.  
%Another example is `Ambiguity Aversion', which says that when the agent is
%indifferent between two acts, she weakly prefers a mixture of them
%(hedging).  
%
%Secondly, there are axioms that distinguish MMEU and MER. 
%For example, the `Independence' axiom holds for MER but not MMEU. 
%Independence says that mixing the entire menu with an act does not affect the ranking of the acts in the menu. 
%On the other hand, `Independence of Irrelevant Alternatives', synonymous with menu-independence, holds for MMEU but not MER. 
}
%joe2: \end{commentout}
%joe4

%joe9:
%sam11: joe9 wrote X,Y,U at places here, I think they should be S,Y,U
Given a set $Y$ of prizes, a utility $U$ on prizes, a state space $S$,
and a set $\cP^+$ of weighted probabilities on $S$, we can define a
family $\succeq_{M,\cP^+}^{S,\Delta(Y),u}$ of preference orders on
Anscombe-Aumann acts determined by weighted regret, one per
menu $M$, as discussed above, where $u$ is the utility function on
lotteries determined by $U$.  For ease of exposition, we usually write 
$\succeq_{M,\cP^+}^{S,Y,U}$ rather than
$\succeq_{M,\cP^+}^{S,\Delta(Y),u}$. 
%joe9*: Sam, please rewrite using this notation.
%joe13: cut sentence, added paragraph break
%We want to axiomatize (i.e., prove a representation theorem) for this family.

We state the axioms in a way that lets us clearly distinguish
the axioms for SEU, MMEU, MER, and MWER.  The axioms are universally
quantified over acts $f$, $g$, and $h$,  menus $M$ and $M'$, and $p \in
(0,1)$.   
%joe2*: I thought at first you needed the next assumption, but then I
%decided you didn't.  Please confirm 
%sam2: I think we do need this assumption. Otherwise, if say f is
%strictly optimal wrt M, then the definition of regret won't be what we
%want. 
%joe8
%We assume that when we write $f \succeq_M g$, then $f, g \in M$.
We assume that $f,g \in M$ when we write $f \succeq_M g$.%
\footnote{Stoye~\citeyear{Stoye2011} assumed that menus were convex, so
  that if $f, g \in M$, then so is $pf + (1-p)g$.  We do not make this
  assumption, although our results would still hold if we did (with the
  axioms slightly modified to ensure that menus are convex).  While it
  may seem reasonable to think that, if $f$ and $g$ are feasible for an
  agent, then so is $pf + (1-p)g$, this not always the case.  For
  example, it may be difficult for the agent to randomize, or it may be
  infeasible for the agent to randomize with probability $p$ for some
  choices of $p$ (e.g., for $p$ irrational).}
%joe8: moved from above 
We use $l^*$ to denote a constant act that maps all states to $l$.  

\begin{axiom}{(Transitivity)}\label{axiom:T}
$f\succeq_M g\succeq_M h\Rightarrow f\succeq_M h$.
\end{axiom}
\begin{axiom}{(Completeness)}
\label{axiom:c} 
$f\succeq_M g \text{ or } g\succeq_M f$.
\end{axiom}
\begin{axiom}{(Nontriviality)}\label{axiom:nt}
$f\succ_M g$ for some acts $f$ and $g$ and menu $M$.
\end{axiom}
\begin{axiom}{(Monotonicity)}\label{axiom:M}
%f\succeq_M g\end{align} 
%\begin{align}\tag{MON}u(f,s)\geq u(g,s),\forall s\in S \Rightarrow
%joe2*: this is not meaningful, since we don't have utility.  What you
%can say is f(s) \succeq g(s) for all $s \in S$, but for that to make
%sense, you have to assume a later axiom that says that utilities are
%menu independent.  I got around the problem by not assming menu
%independence 
%If $u(f,s)\geq u(g,s)$ for all  $s\in S$, then $f\succeq_M g$.
%sam2: maybe it's better to define a "choice conditional on states" as
%Stoye does? 
%joe3*: No!  This is a terrible idea.  I don't even know what choice
%conditional on states means  I can't follow your definition in the
%footnote, and you definitely, positively can't introduce the definition
%in a footnote without explanation and motivation!
%If $f(s)\succeq_M g(s)$ for all  $s\in S$, then $f\succeq_M g$.%
%If $f \succeq_{M\mid s} g$ for all  $s\in S$, then $f\succeq_M g$.%
If $(f(s))^* \succeq_{\{(f(s))^*,(g(s))^*\}} (g(s))^*$ for all  $s\in S$, then
$f\succeq_M g$.%
%sam7 commented out footnote since it's no longer necessary 
%\footnote{Recall that we identify a lottery like $f(s)$ or $g(s)$ with a
%constant act.}
%joe2*: this is problematic if we are going to reuqire that if we write
%f \succeq M g, then f and g have to be in M.  If that's the case, we
%could write f(s) %\succeq_{f(s),(s)} g(s), although it's a bit ugly.
%Or we could move the fact that lotteries are menu independent back.
%sam2: new footnote to explain the choice conditional on states
%joe3*: cut.  This is meaningless.
%\footnote{$\succeq_{M\mid s}$ is the preference over menu $M$ if the agent
%knew that $s$ has occurred.  
%The state-independent utility assumption, which is standard in the
%literature, imposes that this conditional preference is equal to the
%preferences over the constant acts.} 
%This is a notion of state-independent utility.
%\footnote{Note that $f(s)$ and $g(s)$ are lotteries, not acts.  To make
%sense of this axiom, we identify a lottery with the constant act which
%returns that lottery in all states.}
\end{axiom}
%joe2*: undeleted this; it implies continuity (I'll send you a proof when
%I get a chance); it's not at all obvious that it is implied by
%continuity.  In any case, it's more standard
%\begin{align}\tag{MON}u(f,s)\geq u(g,s),\forall s\in S \Rightarrow
\begin{axiom}{(Mixture Continuity)}\label{axiom:arch}
%\begin{align}\tag{MC}
If $f\succ_M g\succ_M h$, then 
%sam4 since p is already universally quantified, use q and r to avoid possible confusion
%sam3 I don't think we'd want to force pf+(1-p)h to be in M?
there exist $ q,r\in(0,1)$ such that $$qf+(1-q)h \succ_{M\cup\{qf+(1-q)h\}} g \succ_{M \cup \{ rf+(1-r)h\}} rf+(1-r)h.$$ 
%\end{align} 
\end{axiom}
%joe2*: 
%Axioms \ref{axiom:T}--\ref{axiom:M} are standard.
%holds also for MMEU and subjective utility maximization. 
Menu-independent versions of Axioms \ref{axiom:T}--\ref{axiom:arch} are standard.
Clearly (menu-independent versions of) Axioms 1, 2, 4, and 5 hold for MMEU, MER, and SEU; Axiom 3 is
assumed in all the standard axiomatizations, and is used to get a
unique representation.

%joe2*: put discussion of convergence up here.  Note that there's an
%obvious notion of convergence for lotteries.  After I put this in, I
%realize that the Archimedean property (what you're claling mixture
%independence) implies continuity, and it's a more standard axiom, so
%I'm using htat instead.  I saved my changes, in case they end up being
%useful.
\commentout{
For the next axiom, we say that for lotteries $l$ and $l_n$, $n = 1, 2, 3,
\ldots$, we say that $l_n$ converges to $l$, written $l_n \rightarrow
l$, if $l_n(y) \rightarrow l(y)$ for all $y \in Y$ (i.e., the
probability $l_n(y)$ of prizes $y$ converges to $l(y)$ for all prizes
$y$);  we say that $f_n \rightarrow f$ if, for all $s$, we have $f_n(s)
\rightarrow f(s)$.  
\begin{axiom}{(Sequential Continuity)}\label{axiom:sc}
%\begin{align}\tag{SC}
If $f_n \succeq_M g$ for all $n\in \N$ and  $f_n\rightarrow f$
 then $f\succeq_M g,$
%\end{align}
%\end{axiom}
Although sequential continuity holds for MMEU, MER, and SEU, it is
typically not assumed, because it follows from other axioms.  Since
these other axioms do not hold for MWER, we must explicitly assume it here.
}
%joe2:\end{commentout}

\begin{axiom}{(Ambiguity Aversion)}\label{axiom:AA}
%\begin{align}\tag{AA}
%sam3: probably don't want to force pf+(1-p)g to be in M, although it doesn't matter by INA?
$$f\sim_M g\Rightarrow pf+(1-p)g \succeq_{M\cup\{pf+(1-p)g\}} g.$$
%$f\sim_M g\Rightarrow pf+(1-p)g \succeq_M g$.
%\end{align}
\end{axiom}
Ambiguity Aversion says that the decision maker weakly prefers to hedge
her bets.  
%joe2*: added.  This is an important story, and it absolutely must be
%discussed.  Note that it follows from this that if we took the absolute
%notion of regret, where we assume a best-possible act, we can drop
%Axioms 6 and 8; they follow from 7
It also holds for MMEU, MER, and SEU, and is assumed in the
axiomatizations for MMEU and MER.  It is not assumed for the
axiomatization of SEU, since it follows from the Independence
%joe3
%axiom, discussed next.  This axiom also holds for MWER, except we must be
axiom, discussed next.  Independence also holds for MWER, provided that
we are
careful about the menus involved.  Given a menu M and an act $h$, let
$pM + (1-p)h$ be the menu $\{pf + (1-p)h: p \in M\}$.  

%joe8
%\begin{axiom}\label{axiom:IND}{(Independence)}\\
\begin{axiom}\label{axiom:IND}{(Independence)}
%\begin{align}\tag{IND}
%joe5: it's not a questions of looks; I prefer (and it's more standard)
%changed back
%$f\succeq_M g$ $\Leftrightarrow$ $pf + (1-p)h \succeq_{pM+(1-p)h} 
%joe8: used $$
%$f\succeq_M g$ iff $pf + (1-p)h \succeq_{pM+(1-p)h} 
$$f\succeq_M g \mbox{ iff } pf + (1-p)h \succeq_{pM+(1-p)h} 
%sam4: <=> seems to look better here than iff 
%to use <=> for as a logical symbol, not an abbrevation for iff 
%joe8
%pg+(1-p)h$ 
pg+(1-p)h.$$
%joe2: don't need this; we've already universally quantified
%,\forall p\in(0,1) 
%\end{align}
\end{axiom}
Independence holds in a strong sense for
SEU, since we can ignore the menus.  The menu-independent version of
Independence is easily seen to imply Ambiguity Aversion.
Independence does not hold for MMEU.  

%joe2*: this is the wrong story.  Independence is the basic axiom.  If
%you want to discuss C-independence, then it shoud be as the weakening
%needed to get a represenation theorem for MMEU/regret.  In any case,
%this is misplaced, since it should go where you state the axiom.
%inomprehensible to most readers.  
%Independence can be thought of as an analog of Certainty Independence
%(C-Independence) from \cite{GilboaSchmeidler1989} for MMEU. Like
%C-Independence, Independence says that ranking between two acts does not
%change when both acts are mixed with a third act; but unlike
%C-Independence, the menu used to compare the original acts is different
%from that used to compare the mixtures. Moreover, in C-Independence the
%third act must be constant. 

%joe2*: You don't have an axiom that says lotteries are menu
%independent, but it seems to me that you need one (otherwise the next
%axiom isn't even well defined.
Although we have menu independence for SEU and MMEU, we do not have it
for MER or MWER.  The following two axioms are weakened versions of menu
independence that do hold for MER and MWER.  
%joe3
%\begin{axiom}{(Menu independence for constant acts):)} If $f$ and $g$ are
\begin{axiom}{(Menu independence for constant acts)} If $l^*$ and $(l')^*$ are
constant acts, then $l^* \succeq_M (l')^*$ iff $l^* \succeq_{M'} (l')^*$.
\end{axiom}
%joe8
\noindent
In light of this axiom, when comparing constant acts, we omit the menu.

%joe2*: moved out of the scope of the axiom.  Definitions shouldn't be
%given inside axioms.  In any case, this definition can't be what you want.
%First, it has to be relative to M.  Second, and more importantly, I
%strongl suspect you want \succeq, not \succ here.  No act in M can be
%strictly better in any state than all acts in M (since it can't be
%strictly better than itself).  I've put down what I think you need.
%joe3*: removed your notation, and defined what you need never strictl
%optimal 
%An act $h$ is \emph{never-optimal  relative to $M$} if, for all states $s
%\in S$, there is some $f \in M$ such that $f \succ_{\mid s} h$.
%sam3
An act $h$ is \emph{never strictly optimal  relative to $M$} if, for all
%joe8
%states $s \in S$, there is some $f \in M$ such that $f(s) \succeq h(s)$.
states $s \in S$, there is some $f \in M$ such that $(f(s))^* \succeq
(h(s))^*$. 
%An act $h$ is \emph{never strictly optimal  relative to $M$} if, for all
%states $s \in S$, there is some $f \in M$ such that $f(s) \succeq h(s)$.
%this axiom holds for all h; every act h is not strictly potentially
%optimal in M \cup {h}.  I suspect you want \succeq rather than \succ
%here, although I'm worried about the ``strictly''.  Finally, note that
%since I haven't assumed that 
%M is convex, if you really need this assumption, then you need to take
%the convex closure here.  
%sam2: I guess we don't need to take the convex closure, since the
%vertices will be sufficient 
%joe3*: but if you assume that menus must be convex, you need to take
%the convex closure
%An act $f$ is strictly potentially optimal if there exists an $s\in S$
%such that $f(s) \succ g(s)$ for all $g\in M$.  
%joe3
%\begin{axiom}{(Independence of Never-Optimal Alternatives (INA))}
\begin{axiom}{(Independence of Never Strictly Optimal Alternatives (INA))}
\label{axiom:INA}
%joe2*: this can't be right
%Let $h$ be not strictly potentially optimal in $M\cup \{h\}$. then 
%sam3
If every act in $M'$ is never strictly optimal relative to $M$, then  
%If $h$ is never strictly optimal relative to $M$, then  
%\begin{align}\tag{INA} 
%joe3*
%$f\succeq_M g$ iff $f\succeq_{M\cup\{h\}} g$.
%sam3
$f\succeq_M g$ iff $f\succeq_{M\cup M'} g.$
%, where $M'$ is the convex closure of $M \cup \{h\}$.
% sam2: it is 'if and only if', right? We can add any act to the menu as
% long as it isn't strictly optimal. 
%$f\succeq_M g$, then $f\succeq_{M\cup\{h\}} g$
%\end{align}
\end{axiom}
%INA captures the limited menu independence of minimax regret, MER, and MWER. 
%Although adding acts to the menu, in general, can affect regret-based
%preferences, adding acts that are not strictly potentially optimal to
%the menu does not. 

%joe3: you actually don't use it
%We use AX-MWER to denote the set of Axioms
%\ref{axiom:T}-\ref{axiom:INA} since, as we now show, these axioms
%characterize MWER.  

%sam7: added new axiom so that all menus considered have bounded utilities
% joe: there must be some utility a such that $u(f(s)) \le a$
%for all states s and acts f \in M.  When we're doing the axiomatization, that
%means that there is some lottery l such that f(s) \preceq l for all f
%\in M and states s.  This is a weaker assumption that assuming that M is
%the convex closure of some finite set of acts, 
\begin{axiom}{(Boundedness of menus)}
\label{axiom:boundedness}
For every menu $M$, there exists a lottery $\overline{l}\in \Delta(Y)$
such that for all $f\in M$ and $s\in S$, 
%joe8
%\begin{align*}
%f(s) \preceq \overline{l}
%\end{align*}
$(f(s))^* \preceq \overline{l}^*$.
\end{axiom}
% explanation of the axiom
%joe8
\noindent
%joe13
%The boundedness axiom guarantees that the utilities of the acts in a
%menu are always bounded above.  
%This is required to prevent acts having infinite regret (and hence
%incomparable using a regret-based decision rule). 
The boundedness axiom enforces the assumption that we made earlier that
every menu has utilities that are bounded from above.  Recall 
that this assumption is necessary for regret to be finite.

%joe20-
%We will now proceed to present our representation theorem for MWER. 
We now present our representation theorem for MWER. 
Roughly, the representation theorem states that a family of preferences
satisfies Axioms~\ref{axiom:T}--\ref{axiom:boundedness} if and only if
it has a MWER representation with respect to some utility function and
weighted probabilities. 
In the representation theorem for SEU \cite{AnscombeAumann1963}, not only is the utility function unique (up to affine transformations, so that we can replace $U$ by $aU + b$, where $a>0$ and $b$ are constants), but the probability is unique as well.
Similarly, in the MMEU representation theorem of Gilboa and Schmeidler \cite{GilboaSchmeidler1989}, the utility function is unique, and the set of probabilities is also unique, as long as one assume that the set is convex and closed.  

%sam20 added discussion on properties required for uniqueness
%joe20: slight rewrite
%There are analogous properties that give us uniqueness of the set of
%weighted probabilities in the representation theorem for MWER.  
%For these properties, it is much more natural to consider a different
To get uniqueness in the representation theorem for MWER, we need to 
consider a different 
representation of weighted probabilities.  
%joe20: slight rewrite: why 2^S?  
%Define a \emph{sub-probability measure space} $(S,2^{S}, \mathbf{p})$,
%where $\mathbf{p}$ is a \emph{sub-probability measure} with range in
%$[0,1]$ and is additive, but $\mathbf{p}(S)$ can be less than $1$. 
%We will identify a weighted probability distribution $(\Pr, \alpha)$
%with the sub-probability measure $\alpha \Pr$. 
%It is easy to see that we can equivalently define weighted regrets with
%respect to a set $C$ of sub-probability measures.  
Define a \emph{sub-probability measure} $\mathbf{p}$ on $S$ to be like a
probability measure (i.e., a function mapping measurable subsets of
$S$ to $[0,1]$ such that $\mathbf{p}(T \cup T') = 
\mathbf{p}(T) +  \mathbf{p}(T')$ for disjoint sets $T$ and $T'$),
without the requirement that $\mathbf{p} = 1$.
We can identify a weighted probability distribution $(\Pr, \alpha)$
with the sub-probability measure $\alpha \Pr$.  (Note that given a
sub-probability measure $\mathbf{p}$, there is a unique pair
$(\alpha,\Pr)$ such that $\mathbf{P} = \alpha \Pr$: we simply take
$\alpha = \mathbf{p}(S)$ and $\Pr = \mathbf{p}/\alpha$.)
%joe20
%We say that a set $C$ of sub-probability measures is
%\emph{downward-closed}, if whenever $\mathbf{p} \in C$ and
A set $C$ of sub-probability measures is
\emph{downward-closed} if, whenever $\mathbf{p} \in C$ and
$\mathbf{q}\leq \mathbf{p}$, then $\mathbf{q}\in C$. 
We get a unique set of sub-probability measures in our representation
%joe20
%theorem, if we restrict to sets that are convex, downward-closed,
%closed, and contains at least one (proper) probability measure. 
theorem if we restrict to sets that are convex, downward-closed,
closed, and contain at least one (proper) probability measure. 
%joe20
(The latter requirement corresponds to having $\alpha_{\Pr} = 1$ for
some $\Pr \in \cP^+$.)
%sam22
For convenience, we will call a set \emph{regular} if it is convex,
downward-closed, and closed. 

We identify each set of weighted probabilities $\cP^+$ with 
%joe20
%a set of sub-probability measures $C(\cP^+)$, defined by 
%$$C(\cP^+) =  \{ \alpha_{} \Pr : (\Pr, \alpha_{\Pr})\in\cP^+, 0 \leq
the set of sub-probability measures 
$$C(\cP^+) =  \{ \alpha_{} \Pr : (\Pr, \alpha_{\Pr})\in\cP^+, 0 \leq
\alpha \leq \alpha_{\Pr}\}.$$ 
%joe20
%Note that $C(\cP^+)$ includes all 
%the points between the all-zero measure and each $\alpha_{\Pr}\Pr$. 
%the sub-probability measures between the all-zero measure and each
Note that if $(\alpha,\Pr) \in \cP^+$, then $C(\cP^+)$ includes all the
sub-probability measures between the all-zero measure and
$\alpha_{\Pr}\Pr$.  
%joe20*: this is not quite true; see below  I just cut this for now
\commentout{
We can also define an inverse mapping, $C^{-1}(\cdot)$, where $C^{-1}(Q)
= \{ (\sup_{\alpha : \alpha \cdot \Pr \in Q }\alpha, \Pr) : \Pr \in
\Delta(S)\}$. 
%joe20
%It is easy to check that one can go back and forth between the two
%representations $\cP^+$ and $C(\cP^+)$. 

\begin{proposition}
%joe20*: here's the bug: you get extra 0's
%For any weighted set of probabilities $\cP^+$, $C^{-1}(C(\cP^+)) =
%\cP^+$. 
For all weighted sets of probabilities $\cP^+$, $C^{-1}(C(\cP^+)) =
\cP^+ \cup \{(\Pr,0): \Pr \in \Delta(S)\}$. 
\end{proposition}
\begin{proof}
%joe20
%First, suppose  $(\alpha,\Pr) \in \cP^+$.
%We'll show that $(\alpha,\Pr) \in C^{-1}(C(\cP^+))$.
%Since  $(\alpha,\Pr) \in \cP^+$, $\alpha \Pr \in C(\cP^+)$.
%Furthermore, there is no $\alpha' > \alpha$ such that $\alpha' \Pr \in
%C(\cP^+)$.   Hence $(\alpha,\Pr) \in  C^{-1}(C(\cP^+))$.
First, suppose  that $(\alpha,\Pr) \in \cP^+$.
It follows that $\alpha \Pr \in C(\cP^+)$
and there is no $\alpha' > \alpha$ such that $\alpha' \Pr \in
C(\cP^+)$.   Hence, $(\alpha,\Pr) \in  C^{-1}(C(\cP^+))$.

%joe20*: and here's the bug in the proof.  There may be no \alpha^* such
%that (\alpha^*,\Pr) \in \cP^+.
%Next, suppose $(\alpha,\Pr) \not\in \cP^+$.
%We'll show that $(\alpha,\Pr) \not\in C^{-1}(C(\cP^+))$.
%Suppose $(\alpha^*,\Pr) \in \cP^+$ where $\alpha^* \in [0,1]$.
%We'll consider two cases. 
%First, suppose $\alpha^* > \alpha$. 
%Then $\alpha^* \Pr \in C(\cP^+)$ and thus $\alpha$ is not an upper bound
%for the set $\{ \alpha' : \alpha' \Pr \in C(\cP^+)\}$ and hence $\alpha
Next, suppose that $(\alpha,\Pr) \not\in \cP^+$ and
$(\alpha^*,\Pr) \in \cP^+$.  
If $\alpha^* > \alpha$, then 
$\alpha^* \Pr \in C(\cP^+)$, so $\alpha$ is not an upper bound
of the set $\{ \alpha' : \alpha' \Pr \in C(\cP^+)\}$ and $\alpha
\not\in C^{-1}(C(\cP^+))$. 
%joe20
%For the other case, suppose $\alpha^* < \alpha$. 
%Then $(\alpha + \alpha^*)/2$ is also an upper bound for the set $\{
%\alpha' : \alpha' \Pr \in C(\cP^+)\}$ (hence $\alpha$ cannot be the
If $\alpha^* < \alpha$, then
$(\alpha + \alpha^*)/2$ is also an upper bound of the set $\{
\alpha' : \alpha' \Pr \in C(\cP^+)\}$.  Thus, $\alpha$ cannot be the
$\sup$ and hence $\alpha \not\in C^{-1}(C(\cP^+))$. 
\end{proof}
}

%joe20
%The reasons why we need the set of sub-probability measures to be convex
%and closed are the same as the reasons why one needs the set of
%probability measures to be closed in the representation of MMEU. 
We need to restrict to closed and convex sets of sub-probability measures
to get uniqueness in the representation of MWER for much the same reason
that we need to restrict to closed and convex sets to get uniqueness in
the representation of MMEU. 
To see why convexity is needed, 
consider the delivery example and the expected regrets in
Table~\ref{tab:qcRegrets}, and the distribution $a \Pr_1 +
(1-a) \Pr_{10}$, for some $a\in (0,1)$. 
The weighted expected regret of any act with respect to $a \Pr_1 +
(1-a) \Pr_{10}$ is bounded above by the maximum weighted expected
regret of that act with respect to $\Pr_1$ and $\Pr_{10}$.  
%joe20
%Therefore, adding $a \Pr_1 + (1-a)\Pr_{10}$ to $\cP^+$, with any weight in
%$(0,1]$, yields another representation for the same MWER preferences.  
Therefore, adding $a \Pr_1 + (1-a)\Pr_{10}$ to $\cP^+$ for some weight 
%joe21: I don't think that the notion of ``MWER preferences'' is
%defined.  Do you want to define it as a family of preference orders
%\succeq_M that satisfies axioms 1-10? %sam21 changed
$a \in (0,1)$ does not change the resulting family of preferences.  
%joe20
%Thus convexity is necessary for uniqueness of the representation.
%Furthermore, the set needs to be closed, since if we have a set $C$ of
Similarly, we need to restrict to closed sets for uniqueness, since if
we start with a set $C$ of
sub-probability measures that is not closed, taking the closure of $C$
would result in the same family of preferences. 

%joe20
%When we want to examine convexity and closedness, it is much more
%natural to consider a set of sub-probability measures than a set of
%weighted probabilities.  For instance, 
While convexity is easy to define for a set of sub-probability measures,
there seems to be no natural notion of convexity for a set $\cP^+$ of
weighted probabilities.
%joe20
%while convexity is clearly defined for a set of sub-probability measures. 
Moreover, the requirement that $\cP^+$ is closed is different from the
requirement that $C(\cP^+)$ is closed.
%joe20
%and the former seems unreasonable.   
The latter requirement seems more reasonable.
For example, 
%joe20: you're implicitly assuming that $(Pr',0) \in \cP^+ for \Pr' \ne
%\Pr, without ever saying that.
%if $\cP^+$ consists of a single probability measure $\Pr$
%with weight $\alpha_{\Pr} =1$,  then there are sequences $\Pr_n
%\rightarrow \Pr$ with $\alpha_{\Pr_n}=0$, and yet $\alpha_{\Pr}=1$. 
%On the other hand, for this example $C(\cP^+)$ is clearly closed.
fix a probability measure $\Pr$, and let $\cP^+ = \{(1,\Pr)\} \cup
\{(0,\Pr'): \Pr' \ne \Pr\}$.  Thus, $\cP^+$ essentially consists of a single
probability measure, namely $\Pr$, with weight 1; all the weighted
probability measures $(0,\Pr')$ have no impact.  This represents the
uncertainty of an agent who is sure that that $\Pr$ is true
probability.  Clearly $\cP^+$ is not closed, since we can find a
sequence $\Pr_n$ such that $(0,\Pr_n) \rightarrow (0,\Pr)$, although
$(0,\Pr) \notin \cP^+$.  But $C(\Pr^+)$ is closed.

%joe20
%In the case of sets of sub-probability measures, we also need
%downward-closedness and the presence of a proper probability
%distribution to get uniqueness.  
Restricting to closed, convex sets of sub-probability measures does not
suffice to get uniqueness; we also need to require downward-closedness.
%joe20
%Downward-closedness is needed because if $\mathbf{p}$ is in $C$, then
This is so because if $\mathbf{p}$ is in $C$, then
adding any $\mathbf{q} \leq \mathbf{p}$ to the set leaves all regrets
unchanged. 
%joe20*: we may want some reweriting here, see my comments on uniqueness below
%Finally, the presence of a proper probability distribution is also
Finally, the presence of a proper probability measure is also
required, since for any $a\in (0,1]$, scaling each element in the set
$C$ by $a$ leaves the family of preferences unchanged. 

%sam21 added clarification
In summary, if we consider arbitrary sets of sub-probability measures,
then the set of sub-probability measures that represent a given family
of MWER preferences would be unique if we required the set to be regular and contain a probability measure. 

\begin{myTheorem}\label{thm:completeness}
For all $Y$, $U$, $S$, and $\cP^+$, the family of preference orders
$\succeq_{M,\cP^+}^{S,Y,U}$ 
satisfies Axioms \ref{axiom:T}--\ref{axiom:boundedness}.  
Conversely, if a family of
preference orders $\succeq_{M}$ on the acts in
$\Delta(Y)^S$  
satisfies Axioms \ref{axiom:T}--\ref{axiom:boundedness}, then there exist a  
a utility $U$ on $Y$ and a weighted set $\cP^+$ of probabilities 
%sam22 used 'regular'
%with $C(\cP^+)$ being regular, such that
%%sam22 used 'regular'
%joe22
%on $S$, with $C(\cP^+)$ being regular, such that
on $S$ such that $C(\cP^+)$ is regular and
$\succeq_M = \succeq_{M,\cP^+}^{S,Y,U}$.  Moreover, $U$ is unique up to
affine transformations, and $C(\cP^+)$ is unique %sam21 clarified uniqueness
in the sense that if $\cQ^+$ represents $\succeq_M$, and $C(\cQ^+)$ is
%joe21: don't we also need to require the existence of a probability
%measure in the set?  Bumping a level, it might be useful to define an
%abbreviation, like ``acceptable'' or ``well-behaved'' (feel free to
%make up you own word here) for a set of sub-probability measures
%that is closed, convex, downward-closed, and contains a probability
%measure.  If you make this definition, then your summary can say
%``there is a unique acceptable set of sub-probability measures that
%represents a given familiy of MWER preferences''
regular, then $C(\cQ^+) = C(\cP^+)$.
%joe20*: Sam, what does it mean that C(\cP^+) is unique?  Do you mean
%(a) that if \cQ^+ represents the preferences, then C(\cQ^+) = C(\cP^+)$, or
%(b) that if \cQ^+ represents the preferences *and C(\cP^+) is closed,
%convex, downward closd  and contains a proper probability meausre, then
%C(\cQ^+) = C(\cP^+)$.  I suspect that you mean the latter.  But then I
%think that you can simplify a bit.  c(\cQ^+) is bound to contain a
%proper probality, so you don't have to say that.  I think it suffices
%to say that if C(\cQ^+) is closed, convex, and downward-closed, then
%C(\cP^+) = C(cQ^+) (and we need to say explicitly that 
%this is what uniqueness means).  In our discussion of uniqueness, we
%can perhaps point out that if we were working with arbitrary sets of
%subprobability measures, then we would get uniquess if we required them
%to be closed, convex, downward-closed, and contain a probability
%measure. In any case, if you agree with this, you need to do a little
%rewriting. 
\end{myTheorem}
%(High level sketch)
%joe2:
%The verification of the `if' direction is
%straightforward. %sam:hopefully
%For the `only if'/ completeness direction, there are three key steps. 
%joe9
%Showing that $\succeq_{M,\cP^+}$ satisfies Axioms
Showing that $\succeq_{M,\cP^+}^{S,Y,U}$ satisfies Axioms
\ref{axiom:T}--\ref{axiom:boundedness} is fairly straightforward;
%joe15
we leave details to the reader.
%joe11
%For the converse direction, 
The proof of the converse is quite nontrivial, although it follows the lines
of the proof of other representation theorems.  We provide an outline of
the proof here; details can be found in the appendix. 

Using standard techniques, we can show that the axioms guarantee the
existence of a utility 
function $U$ on prizes that can be extended to lotteries in the obvious
way, so that  
$l^* \succeq (l')^*$ iff $U(l) \ge U(l')$. 
We then use techniques of Stoye \citeyear{Stoye2011}
to show that it suffices to get a representation theorem for a single
menu, rather than all menus: the menu consisting of all acts $f$ such
that $U(f(s)) \le 0$ for all states $s \in S$.
This allows us to use techniques 
%joe11
%similar to those used 
in the spirit of those used by
%joe11: added
by Gilboa and Schmeidler \citeyear{GS1989} 
to represent
(unweighted) MMEU.  However, there are technical difficulties that arise
from the fact that we do not have a key axiom that is satisfied by MMEU:
C-independence (discussed below).   
%joe11
%It turns out that the restriction to
%non-positive functions helps us deal with the lack of C-independence.
The heart of the proof involves dealing with the lack of C-independence;
%\shortv{We leave details of the proof to the full paper.}
%\fullv{We defer the proof to the appendix.}
%joe15: this is the full paper!
%\fullv{We leave details of the proof to the full paper.}
%\shortv{We defer the proof to the appendix.}
as we said, the details can be found in the appendix.

%joe15
\medskip

%joe2*: added some discussion, to put things in context.
It is instructive to compare Theorem~\ref{thm:completeness} to other
representation results in the literature.
Anscombe and Aumann \citeyear{AnscombeAumann1963} 
showed that the menu-independent versions of axioms
\ref{axiom:T}--\ref{axiom:arch} and \ref{axiom:IND} characterize SEU.  
%joe3
%The presence of Axiom~\ref{axiom:IND} (Independence) greatly simplifies
The presence of Axiom~\ref{axiom:IND} (menu-independent Independence)
greatly simplifies 
things.  
Gilboa and Schmeidler \citeyear{GilboaSchmeidler1989}
showed that axioms \ref{axiom:T}--\ref{axiom:AA}
%joe12: what's the C in C-independence stand for?
%together with one more axiom that they call C-independence characterizes
%sam15: corrected what C in C-independence stands for 
together with one more axiom that they call
\emph{Certainty-independence} characterizes 
MMEU.  Certainty-independence, or C-independence for short, is a weakening of independence (which, as we
observed, does not hold for MMEU), where the act $h$ is required to be a
%constant act.  Since MMEU is menu-independent, we %sam13
constant act.  Since MMEU is menu-independent, we
state it in a menu-independent way.
\begin{axiom}{(C-Independence)}\label{axiom:C-IND}
%\begin{align}\tag{AA}
If $h$ is a constant act, then 
%sam4 changed iff to <=>, seems to look better
%joe5: again
%$f\succeq g$ $\Leftrightarrow$ $pf+(1-p)h \succeq g + (1-p)h$. 
%joe12: typo!
%$f\succeq g$ iff $pf+(1-p)h \succeq g + (1-p)h$. 
$f\succeq g$ iff $pf+(1-p)h \succeq pg + (1-p)h$. 
%\end{align}
\end{axiom}

%sam3
%joe11: this seems like a distraction.  It make it seem like the variant
%is given by Axiom 12, but it's not
%A variation of C-Independence holds for MER.  For MER, we need to make the
%preferences menu-dependent, and assume that the relevant acts are in the
%menu. 
%sam3: Commented out stuff about Weak-C-Independence since we want C-Betweenness instead
\commentout{ Moreover, we need to assume that a constant act $\overline{h}$
dominates all other  
acts in the menu $M$ (that is, $\overline{h} \succeq_{\mid s} f$ for all
$f\in M$). 
%menu.  Moreover, we need to assume the act $f$ dominates all other acts
%in the menu $M$ (that is, $f(s) \succeq g(s)$ for all $g \in M$).
%Since MWER is a generalization of MER, one would expect that the axioms
%characterizing MWER form a strictly weaker set of constraints than those
%characterizing MER.  
%Indeed, if we add the following axiom to AX-MWER, based on a
%corresponding axiom in \cite{StoyeRegret}, we obtain a set of axioms
%that characterize MER.  
%joe2
%\begin{axiom}{(Independence for Menus with a dominating act)}\\
\begin{axiom}{(Weak C-Independence)}\label{axiom:WCID} 
%joe2*: you only want p \in (0,1), not [0,1]
%For any $p\in [0,1]$, for any constant act $P^*$, for any menu $M$
%containing $f,g,P^*$, and a weakly dominating act $\overline{f}\in M$,  
%sam2: corrected, want a constant dominating act
If a constant act $\overline{h}\in M$ weakly dominates $M$, then for any
constant act $h\in M$,  
%If $f$ weakly dominates $M$, $h$ is a constant act, and 
%sam2 changed: $f$, $g$, $pf + (1-p)h$, and $pg + (1-p)h$ are all in
%$M$, then 
%\begin{align}\tag{CID}
%\label{equ:WCID} 
$f\succeq_{M} g$ iff $pf+(1-p)h \succeq_{M} pg+(1-p)h$.
%\end{align}
\end{axiom}
%\ref{equ:CID} weakens C-Independence axiom for MMEU
%\cite{GilboaSchmeidler1989} in two ways: it is menu-dependent, and it
%requires a dominating act in the menu.  
%In a sense, it `counteracts' Ambiguity Aversion by limiting the extent
%to which mixtures of acts are preferred over non-mixtures.  
%joe2: is this the right Stoye reference 
%sam2: The story is a bit complicated   
The proofs in Stoye~\citeyear{Stoye2011} can be readily adapted to show
that Axioms~\ref{axiom:T}--\ref{axiom:INA}
together with Axiom~\ref{axiom:WCID} characterize MER.

\begin{table}
	\centering
		\begin{tabular}{|l|l|l|l|l|l|l|} \hline
&	$1$ broken cake		& $10$ broken cakes\\ \hline
		$\cont$ 		  															& 10,000 			& -10,000				\\ \hline 
		$\frac{1}{2}\cont + \frac{1}{2}\back  $			& 5,000 			& -5,000				\\ \hline 
			$\back$																		& 0						& 0 \\ \hline 
		$\frac{1}{2}\chec + \frac{1}{2}\back $				&	2,500.5			& -2499.5 \\ \hline 
					$\chec$  															&	5,001				& -4,999 \\ \hline 
		\end{tabular}
	\caption{Payoffs for the delivery problem, with an acts mixed
		with the constant act $\back$.}  
	\label{tab:QC56}
\end{table}
\begin{table}
	\centering
		\begin{tabular}{|l|l|l|l|l|l|l|} \hline
&	$1$ broken cake		& $10$ broken cakes\\ \hline
		$\cont$ 		  															&     0 			& 10,000				\\ \hline 
		$\frac{1}{2}\cont + \frac{1}{2}\back  $			& 5,000 			& 5,000				\\ \hline 
			$\back$																		& 10,000				& 0 \\ \hline 
		$\frac{1}{2}\chec + \frac{1}{2}\back $				&	7,499.5			& 2,499.5 \\ \hline 
					$\chec$  															&	4,999				& 4,999 \\ \hline 
		\end{tabular}
	\caption{Regrets for the delivery problem, with two acts mixed with the constant act $\back$.} 
	\label{tab:QC57}
\end{table}

%joe3*: It would be good to show that we also need it to be a *constant* act.
%The assumption that $\overline{h}$ dominates $M$ is critical in
The assumption that there is a constant act $\overline{h}$ that
dominates $M$ is critical in  Axiom~\ref{axiom:WCID}.
For example, recall that in the delivery problem, $\cont$ has
higher maximum expected regret ($10,000$) than $\chec$ ($4,999$). 
%joe3
%However, as seen in Tables~\ref{tab:QC56} and \ref{tab:QC56},
However, as seen in Tables~\ref{tab:QC56} and \ref{tab:QC57},
$\frac{1}{2}\cont + \frac{1}{2}\back  $ has \emph{lower} maximum expected
regret  
($5,000$) than $\frac{1}{2}\chec + \frac{1}{2}\back $
%joe8
%($7499.5$). Therefore, Axiom~\ref{axiom:WCID} is violated. 
($7499.5$), violating the version of Axiom~\ref{axiom:WCID} without the
assumption that the menu has a state-independent outcome distribution.

\begin{table}
	\centering
		\begin{tabular}{|l|l|l|l|l|l|l|} \hline
&	$1$ broken cake		& $10$ broken cakes\\ \hline
		$\sd$																			& 10,000			& 10,000	\\ \hline
		$\cont$ 		  															& 10,000 			& -10,000				\\ \hline 
		$\frac{1}{2}\cont + \frac{1}{2}\back  $			& 5,000 			& -5,000				\\ \hline 
			$\back$																		& 0						& 0 \\ \hline 
		$\frac{1}{2}\chec + \frac{1}{2}\back $				&	2,500.5			& -2,499.5 \\ \hline 
					$\chec$  															&	5,001				& -4,999 \\ \hline 
		\end{tabular}
	\caption{Payoffs for the delivery problem, with mixed acts and a dominating constant act.} 
	\label{tab:QC58}
\end{table}
\begin{table}
	\centering
		\begin{tabular}{|l|l|l|l|l|l|l|} \hline
&	$1$ broken cake		& $10$ broken cakes\\ \hline
		$\sd$																			& 10,000			& 10,000	\\ \hline
		$\cont$ 		  															&      0 			& 20,000				\\ \hline 
		$\frac{1}{2}\cont + \frac{1}{2}\back  $			& 5,000 			& 15,000				\\ \hline 
			$\back$																		& 10,000			& 10,000 \\ \hline 
		$\frac{1}{2}\chec + \frac{1}{2}\back $				&	7,499.5			& 12,499.5 \\ \hline 
					$\chec$  															&	4,999				& 14,999 \\ \hline 
		\end{tabular}
	\caption{Regrets for the delivery problem, with mixed acts and a dominating constant act.} 
	\label{tab:QC59}
\end{table}
%joe3
%Axoim~\ref{axiom:WCID} is no longer sound once we add weights. 
Although Axiom~\ref{axiom:WCID} is sound for unweighted expected regret, 
it is no longer sound once we add weights. 
%joe3: slow down a little
For example, suppose that we add a dominating constant act to the
delivery problem, as in Table~\ref{tab:QC58},
%joe3
so that the hypotheses of Axiom~\ref{axiom:WCID} hold.
Then the regrets will be those shown in Table~\ref{tab:QC59}.
Suppose 
%joe3
that
the weights on $\Pr_1$ and $\Pr_{10}$ are $1$ and $0.2$, respectively. 
Then $\cont$ has higher maximum weighted expected regret ($4,000$) than
$\chec$ ($4,999$). 
However, $\frac{1}{2}\cont + \frac{1}{2}\back  $	has \emph{lower} maximum weighted expected regret ($5000$) than 
$\frac{1}{2}\chec + \frac{1}{2}\back $	($7,499.5$).
Therefore Axiom~\ref{axiom:WCID} is violated.
} %sam3: end \commentout 
%sam3: Use this instead. Alternative to weak C-Independence (this is from StoyeRegret)
%joe3*: I have no idea why this is here.  There is no story.  How does
%it contribute to the paper?    If we were going to add this, it would
%need some story.  Cut it for now.
%joe11: added paragraph break and some glue

%joe12: we need a story here.  I don't understand where the name weak
%C-betweenness comes from.  If you can't explain it, cut it
As we observed, in general, we have Ambiguity Aversion 
%sam13 capitalized since it is capitalized earlier in the paper 
(Axiom~\ref{axiom:AA}) for regret.
%sam13 I think this isn't the right interpretation
%The following axiom gives conditions
%under which we do not have ambiguity aversion.
\emph{Betweenness} \cite{Chew1983} is a stronger notion than ambiguity
%joe16
%aversion, which states that if two acts are ranked indifferent, then any
%mixture of these two acts must also be ranked indifferent to these
aversion, which states that if an agent is indifferent between two acts,
then he must also be indifferent among all convex combinations of these
acts. 
While betweenness does not hold for regret, Stoye \cite{StoyeRegret}
%joe13
%gives a weaker version which does hold. 
gives a weaker version that does hold. 
%joe13:
%
A menu $M$ has \emph{state-independent outcome distributions} if the set 
$L(s) = \{y\in\Delta(Y): \exists  f\in M, f(s)=y\}$ is the same for all
states $s$.
%joe12: cut name for now for now
\begin{axiom}
%joe13: cut name; it doesn't add anything
%{(C-Betweenness When Outcome Distributions are State Independent)}
\label{axiom:WCB} %sam13 added name back
%\begin{axiom}{axiom:WCB} 
If $h$ is a constant act, and $M$ has state-independent outcome distributions,
then 
%joe12*: changed menu, and used \sim, for a better story
%$$ h\sim_M f \Rightarrow  pf + (1-p)h \preceq_M f.$$
$$ h\sim_M f \Rightarrow  pf + (1-p)h \sim_{M\cup\{pf+(1-p)h\}} f.$$
\end{axiom}
%joe13: cut
%The ``C'' in C-betweenness refers to the fact that $h$ in the axiom
%must be a constant act. 
%joe13: moved from below
The assumption that the menu has state-independent outcome distributions
is critical in  Axiom~\ref{axiom:WCB}.
%joe10
%In our full paper we give an example to show that the variant of Axiom
%sam14 Removed this statement since this is the full paper.
%In the full paper, we give an example showing that the variant of Axiom
%12 without the state-independent outcome distribution requirement does
%not hold.  

Stoye~\citeyear{StoyeRegret} shows that
Axioms~\ref{axiom:T}--\ref{axiom:INA} 
together with Axiom~\ref{axiom:WCB} characterize MER.\footnote{Stoye
actually worked with choice correspondences; see
Section~\ref{sec:conclusion}.}  
Non-probabilistic regret (which we denote REG) can be viewed as a
special case of MER, where $\cP$ consists of all distributions.
%joe3: slowed down
%However, it satisfies an additional axiom.  
This means that it satisfies all the axioms that MER satisfies.  As
Stoye \citeyear{Stoye2011} shows,  REG is characterized by
Axioms~\ref{axiom:T}--\ref{axiom:INA} and one additional axiom, which he
calls Symmetry.  We omit the details here.

\begin{table}
	\centering
		\begin{tabular}{|l|l|l|l|l|l|l|} \hline
&	\multicolumn{2}{|c|}{$1$ broken cake}		& \multicolumn{2}{|c|}{$10$ broken cakes}\\ \hline
& Payoff	& Regret & Payoff & Regret \\ \hline
		$\cont$ 		  															& 10,000 		&0			& -10,000	& 10,000			\\ \hline 
		$\frac{1}{2}\cont + \frac{1}{2}\back  $			& 5,000 		&5,000	& -5,000	& 5,000	\\ \hline 
		$\back$																			& 0					&10,000	& 0 			& 0\\ \hline 
		$\chec$  																		&	5,001			&4,999	& -4,999 	& 4,999 \\ \hline 
		\end{tabular}
	\caption{Payoffs and regrets for the delivery problem, with $\cont$ mixed with the constant act $\back$.} 
	\label{tab:QC66}
\end{table}

The assumption that the menu has state-independent outcome distributions
is critical in  Axiom~\ref{axiom:WCB}. 
%joe9
%For example, suppose we change the payoffs in the delivery problem, such
For example, suppose that we change the payoffs in the delivery problem so
that $\cont$ has the same maximum expected regret as $\back$ ($10,000$).  
However, as seen in Table~\ref{tab:QC66},
$\frac{1}{2}\cont + \frac{1}{2}\back  $ has lower maximum expected
regret ($5,000$) than $\cont$ ($10,000$),
%joe8  
%Therefore, Axiom~\ref{axiom:WCB} is violated.
showing that the variant of Axiom~\ref{axiom:WCB} without the
state-independent outcome distribution requirement does not hold.

\begin{table}
	\centering
		\begin{tabular}{|l|l|l|l|l|l|l|} \hline
&	\multicolumn{2}{|c|}{$1$ broken cake}		& \multicolumn{2}{|c|}{$10$ broken cakes}\\ \hline
& Payoff	& Regret & Payoff & Regret \\ \hline
		$\cont$ 		  															& 10,000 		&0			& -10,000	& 20,000			\\ \hline 
		$\frac{1}{2}\cont + \frac{1}{2}\back  $			& 5,000 		&5,000	& -5,000	& 15,000	\\ \hline 
		$\back$																			& 0					&10,000	& 0 			& 10,000\\ \hline 
		$\chec 1$  																	&	-5,000		&15,000	& 5,000 	& 5,000 \\ \hline 
		$\chec 2$  																	&	-10,000		&20,000	& 10,000 	& 0			 \\ \hline 
		\end{tabular}
	\caption{Payoffs and regrets for the delivery problem, with state-independent outcome distributions. }
	\label{tab:QC67}
\end{table}

%sam4
Although Axiom~\ref{axiom:WCB} is sound for unweighted minimax expected
regret, it is no longer sound once we add weights.  
%Although Axiom~\ref{axiom:WCB} is sound for unweighted expected regret,
%it is no longer sound once we add weights.  
%joe8
%For example, suppose that we modified the delivery problem such that all
For example, suppose that we modified the delivery problem so that all
states we care about have the same outcome distributions, 
%joe8
%so that the hypotheses of Axiom~\ref{axiom:WCB} hold. 
as required by Axiom~\ref{axiom:WCB}.
Then the payoffs and regrets will be those shown in Table~\ref{tab:QC67}.
Suppose that the weights on $\Pr_1$ and $\Pr_{10}$ are $1$ and $0.5$, respectively. 
Then $\cont$ has the same maximum weighted expected regret as $\back$ ($10,000$). 
However, $\frac{1}{2}\cont + \frac{1}{2}\back  $	has lower
maximum weighted expected regret ($7,500$) than $\cont$, 
%joe8
%Therefore Axiom~\ref{axiom:WCB} is violated.
showing that Axiom~\ref{axiom:WCB} with weighted probabilities does not hold.

%joe2*: Weak C-Independence must still be sound for the
%non-probabilistic case, because it can be viewed as a special case of
%the probabilistic case.  Don't we still need it for REG?
%sam2: Seems like we do not need it; Stoye excludes it from his
%axiomatization of REG.  
%I'm guessing Symmetry implies weak C-independence, although I haven't
%verified it. 
%joe3*: I wouldn't worry about it.

%joe3*: In retrospect, there is no need for this in our paper.
\commentout{
\begin{axiom}\label{axiom:symmetry}{(Symmetry)}
For any menu $M$, let $E,F \subseteq S$ be nonempty, disjoint events such that for any 
$f\in M$, $f$ is constant on $E$ as well as $F$. Define $f'$ by 
$$ f'(s) = \begin{cases}
 f(s'), s'\in E, \text{ if }s\in F \\
 f(s'),s' \in F, \text{ if }s\in E \\
 f(s) \text{ otherwise }
\end{cases}.$$
Let $M'$ be the menu generated by replacing every act $f\in M$ with $f'$.
Then $$ f\succeq_M g \Leftrightarrow f'\succeq_{M'} g'.$$
\end{axiom}
Symmetry ensures that the preference ordering does not impose prior beliefs by implicitly assigning different 
likelihoods to different events. 
Stoye \citeyear{Stoye2011} shows that REG is characterized by Axioms~\ref{axiom:T}--\ref{axiom:INA} and Axiom~\ref{axiom:symmetry}.
}
%joe3: \end{commentout}

%joe11: I would reinstate this
%Sam10 : removed to make room for UAI. Reasoning: only 1 column is due to our result...
Table~\ref{tab:axiomComp} describes the relationship between the axioms
characterizing the decision rules.

%joe17: the table is overly long for UAI style.  I made some changes
%below to deal with that.
\begin{table}
	\centering
		\begin{tabular}{|l|l|l|l|l|l|l|l|} \hline
%joe2
%&	Minimax Regret				& MER
&SEU&	REG				& MER
		& MWER  		& MMEU
		\\ \hline 
%joe17
%			Axioms 1-6,8-10 &\checkmark	
			Ax.~1-6,8-10 &\checkmark	
	& \checkmark 
		& \checkmark 	&	\checkmark&\checkmark
		\\ \hline 
			Ind	&\checkmark	& \checkmark
		& \checkmark	& \checkmark& 	 						\\ \hline
%joe2: moved up C-independence
  		C-Ind	& 	\checkmark&					 & 			 			& 					& \checkmark			\\ \hline
%joe17: I don't think we name this in the text
%			Weak C-Btwn &\checkmark	& \checkmark
			Ax.~12 &\checkmark	& \checkmark
			& \checkmark  &
			&
			\\ \hline  
			Symmetry		& \checkmark	&
		\checkmark
		&
		& 					&
		\\ \hline  
		\end{tabular}
%joe17
%	\caption{Characterizing axioms for several decision rules}
	\caption{Characterizing axioms for several decision rules.}
	\label{tab:axiomComp}
\end{table}

%joe2*: Mixture continuity is *stronger* (not weaker) than sequential
%continiuty.  Cut this
%\commentout{
%However, there is a subtlety concerning Sequential Continuity. 
%For Minimax Regret, MER, and MMEU, Sequential Continuity can be weakened to a property called Mixture Continuity while maintaining the characterization. 
%In these cases, strengthening Mixture Continuity to Sequential Continuity further guarantees a continuous utility function. 
%However, our current proof of MWER requires Sequential Continuity, rather than Mixture Continuity. 
%An alternate proof may allow Sequential Continuity to be weakened.
%}

%\footnote{In \cite{StoyeRegret}, an axiomatization of
%MER is provided for choice correspondences. 
%That is, instead of assuming that the agent exhibits a set
%$\succeq_{M}$ of preference relations, the agent reveals a choice
%$C(M)\subseteq M$ for every menu, representing the agent's most
%preferred set of acts if $M$ were the feasible acts.  
%Despite this difference, however, the axiomatization of MER can be
%easily adapted, with the help of the axiomatization of
%(distribution-free) minimax regret in \cite{Stoye2011}, for the setting
%that we are working with, where the agent reveals a set of menu-dependent
%preference relations. We provide this adapted axiomatization in our
%full paper.}   

%joe2
%\subsection{MWER with Likelihood Updating}
\section{Characterizing MWER with Likelihood Updating}\label{sec:charMWERL}

%joe2*: You need to SLOW DOWN here.  Tell a story before you rush into
%the technical results.  Set the scene!  Don't rush so quickly into the
%technical results.  I added a paragraph
%We show that by adding a single axiom to AX-MWER, we can characterize
%MWER with \emph{likelihood updating}. 
We next consider a more dynamic setting, where agents learn information.
For simplicity, we assume that the information is always a subset $E$ of the
state space.  If the agent is representing her uncertainty using a set
$\cP^+$ of weighted probability measures, then we would expect her to
update $\cP^+$ to some new set $\cQ^+$ of weighted probability measures,
%sam2: "and then MWER" is this intended or a typo?
%joe3
%and then MWER with $\cQ^+$.  In this section, we characterize what
and then apply MWER with uncertainty represented by$\cQ^+$.  In this
section, we characterize what 
happens in the special case that the agent uses likelihood updating, so
that $\cQ^+ = (\cP^+ \mid E)$.  

For this characterization, we assume
that the agent has a family of preference orders $\succeq_{E,M}$ indexed 
not just by the menu $M$, but by the information $E$.  Each preference
%joe3
%order $\succeq_{E,M}$ satisfies AX-MWER, since the agent makes decisions
%joe8: typo, I assume
%order $\succeq_{E,M}$ satisfies Axioms \ref{axiom:T}--\ref{axiom:INA},
order $\succeq_{E,M}$ satisfies Axioms \ref{axiom:T}--\ref{axiom:boundedness},
since the agent makes decisions 
after learning $E$ using MWER.  Somewhat surprisingly, all we need is
one extra axiom for the characterization;  
%joe2: back to your story. I removed the hyphen from M-DC; I think it
%looks better that way
we call this axiom MDC, for `menu-dependent dynamic consistency'. 

%joe2: first some notation
To explain the axiom, we need some notation.  As usual, we take $fEh$ to
be the act that agrees with $f$ on $E$ and with $h$ off of $E$; that is
$$fEh(s) = \left\{\begin{array}{ll}
f(s) &\mbox{if $s \in E$}\\
h(s) &\mbox{if $s \notin E$.} 
\end{array} \right.
$$
%joe2:
%In the quality-control example, the option to $\chec$ can be 
%thought of as the act that tests then (1) sells at $high$ price if the
%true state is one where there is only 1 broken cake, and (2) sells at
%$\back$ price if the true state is one where there are 10 broken cakes. 
In the delivery example, the act $\chec$ can be 
thought of as $(\cont)E(\back)$, where $E$ is the set of states where
there is only one broken cake.

%joe2
%Roughly, MDC says that if you knew
%the true state was in $E$, then you must also prefer $f$ over $g$ in a
%ignored and some other act $h$ is chosen.  %$f$ over $g$ in a modified decision problem where, if $E$ does not
%occur, your choice is 
%ignored and some other act $h$ will be chosen.  
%joe3
%Roughly speaking, MDC says that if you prefer $f$ to $g$ once you learn
Roughly speaking, MDC says that you prefer $f$ to $g$ once you learn
$E$ if and only if, for any act $h$, you also prefer $fEh$ to $gEh$
%sam2 iff => if and only if  maybe better? 
%$E$ iff, for any act $h$, you also prefer $fEh$ to $gEh$
before you learn anything.  
%Conversely, if you prefer $f$ over $g$ in the modified problem, then
%you must also prefer $f$ over $g$ if you knew the true state had to be
%in $E$.  
%joe3
%This seems reasonable, since knowing that the true state was in $E$ is
This seems reasonable, since learning that the true state was in $E$ is
conceptually similar to knowing that none of your choices matter
off of $E$.  

To state MDC formally, we need to be careful about the menus involved.
Let $MEh = \{fEh: f \in M\}$.  
%joe2*: we need this too
We can identify unconditional preferences with preferences conditional
on $S$; that is, we identify $\succeq_M$ with $\succeq_{S,M}$.
%joe2: moved this back
We also need to restrict the sets $E$ to which MDC applies.
Recall that conditioning using likelihood updating is undefined for an
event such that $\ucP(E) = 0$.
%joe2
%$\Pr(E)=0$ for all $\Pr\in\cP$.  
%sam added
That is,
$\alpha_{\Pr}\Pr(E)=0$ for all $\Pr\in\cP$.  
%joe2:
%As is commonly done, we define \emph{non-null} events as a sufficient
%condition such that \emph{likelihood updating} will be defined.  
As is commonly done, we capture the idea that conditioning on $E$ is
possible using the notion of a \emph{non-null} event.
%It is necessary to exclude null events from our representation result,
%since for a null event $E$, you would have f\succeq_E,M g but we won't
%have a MWERL representation for it 
\begin{definition}
%joe2
%An event is non-null if it is \textbf{not} the case that $\forall f,g\in
%L, M\supseteq \{f,g\}, fEg \sim_M f$. 
An event $E$ is \emph{null} if, for all $f, g \in \Delta(Y)^S$ and menus
$M$ with $fEg, g \in M$, we have $fEg \sim_M g$.   
\end{definition}

%joe2: premature
%Theorem \ref{thm:update} shows that agent preferences satisfy AX-MWER
%and MDC if and only if they can be represented by MER with
%likelihood updating.  
%Theorem \ref{thm:update} can be interpreted as follows: if your
%unconditional preferences are MER, and you believe MDC is reasonable,
%then you must behave as-if you are using likelihood updating.

%We will need some notation to state the MDC axiom.
%Let the act $fEh$ be the act that is $f$ if $E$ occurs, and is $h$ if
%$E$ doesn't occur. 
%$fEh$ can be thought of as a `contingency act' that will do $f$ or $h$
%depending on whether the true state is in $E$ or not.  
%However, we can define $fEh$ whether the resulting contingency act makes
%sense in the decision problem or not. 
%More precisely, $fEh(s) = f(s)$ if $s \in E$, and $h(s)$ otherwise.
%
%Moreover, we can modify an entire menu of acts to be some other act
%contingent on event $E$ not occurring.  
%We denote such a menu $\{gEh : g\in M\}$ by $MEh$.

\begin{description}
\item[MDC.]
For all non-null events $E$,
%\begin{align}\label{eq:mDC}\tag{MDC}
$f \succeq_{E,M} g$ iff $fEh \succeq_{MEh} gEh$ for some $h \in M$.%
\footnote{Although we do not need this fact, it is worth noting that 
the MWER decision rule has the property that  $fEh \succeq_{MEh} gEh$
for some act $h$ iff   $fEh \succeq_{MEh} gEh$ for all acts $h$.  Thus,
%joe3
%this property follows from AX-MWER.}
%joe19
%this property follows from Axioms \ref{axiom:T}--\ref{axiom:INA}.}  
this property follows from Axioms \ref{axiom:T}--\ref{axiom:boundedness}.}  
%\end{align} %sam4 actually follows from Independence, doesn't it?
%joe5*: It's not at all obvious to me how this follows from Independence
%as we use it here.  There's a notion of Event independence that Savage
%uses that says exactly this, but we're not using it
\end{description}
%joe2*: this is not ``an aside''.  It's a fact that must follow
%from AX-MWER if it's true.  
%As an aside, the right hand side holds for all $h\in M$ if and only if
%it holds for some $h\in M$.  
%This fact is apparent within the proof of Theorem \ref{thm:update}.
%joe2: moved this observation here
The key feature of MDC is that it allows us to reduce all the
conditional preference orders $\succeq_{E,M}$ to the unconditional order
$\succeq_M$, to which we can apply Theorem~\ref{thm:completeness}.  

%joe: you can call this an intuitive
%discussoin if you like, but it still needs some rewriting. 
%joe2*: cut all this.  We don't need equal initial weights.  The proof
%the MDC is sound should go into the proof of the theorem, and be a bit
%more formal.
\commentout{
%For ease of exposition, we will call MWER with equal initial weights and
%\emph{likelihood updating} MWERL.  
As an intuitive discussion, we look at why MWERL satisfies \ref{eq:mDC}. 
First let's consider the direction $f\succeq_{E,M} g \Rightarrow fEh \succeq_{MEh} gEh$. 
Suppose $f\succeq_{E,M} g$.
As observed earlier, MWERL actually compares regrets using the \emph{unconditional} probabilities ($\Pr(s|E)\cdot \Pr(E) = \Pr(s)$). 
Therefore, $f\succeq_{E,M} g$ means that the maximum expected regret associated with states in $E$ is lower for $f$ than for $g$. 
Since for the menu $MEh$, all feasible acts are identical outside of
$E$, it is not possible for any act to have any regrets outside of $E$
with respect to the menu $MEh$. 
Therefore, MWERL prefers $fEh$ to $gEh$ unconditionally. 

%joe1: Is it M_E h or MEh?  You're inconsistent.
%Next, consider the other direction, $fEh \succeq_MEh gEh, \exists
Next, consider the opposite direction.  
Suppose that, for some $h$, $fEh \succeq_{MEh} gEh$.
Since for the menu $MEh$, all feasible acts are identical outside of $E$, the only reason that MWERL prefers $fEh$ over $gEh$ is because $fEh$ has lower maximum expected regret on $E$ with respect to $\Pr$. 
Since MWERL essentially compares regrets using the unconditional probabilities $\Pr$, MWERL thus also must prefer $fEh$ over $gEh$ conditional on $E$. 
Furthermore, it is easily verified that $fEh\sim_{E,M} f$ and that the axioms guarantee that $\succeq_{E,M}$ is a weak order. 
Therefore we can conclude that $f\succeq_{E,M} g$.
}
%joe2: \end{commentout}

%Since we know that (\ref{eq:savcond1}) does not hold for ambiguity
%averse preferences, \ref{eq:mDC} may appear unbelievable at first
%sight. 
%However, note that the Axiom is a lot weaker than it seems -- the
%entire menu on the RHS is changed in such a way that all acts are equal
%off of event $E$. Intuitively, then, unconditionally there is $0$
%regret associated with $E^c$ and hence acts can be compared by their
%regrets on $E$ only. 
\begin{myTheorem} \label{thm:update}
%joe2*: rewrote theorem statement
%Let $\{\succeq_{E,M}\}_{M\subseteq L}$ be a family of binary relations
%on $M\subseteq L$.  
%For any non-null event $E\subseteq S$, the following are equivalent: 
%\begin{align*}
%&1. \text{ $\{ \succeq_{E,M}\}_{M\subseteq L}$ satisfies AX-MWER and \ref{eq:mDC}} \\
%&2. \forall M\subseteq L, f \succeq_{E,M} g \Leftrightarrow \\ 
%& \max_{\Pr\in\cP}\left( \alpha_{\Pr}\Pr(E) \sum_{s\in E}\left[
%\left(\max_{a^*\in M}u(a^*,s)-u(f,s)\right)\Pr(s|E)\right]\right) \\  
%& \leq \max_{\Pr\in\cP}\left( \alpha_{\Pr}\Pr(E) \sum_{s\in E}\left[
%\left(\max_{a^*\in M}u(a^*,s)-u(f,s)\right)\Pr(s|E)\right]\right) 
%\end{align*}
%Where $\cP\subseteq\Delta(S)$ and $u:X\rightarrow \R$ are those
%specified by the MWER preferences implied by AX-MWER; hence $\cP^+$ is
%a set of weighted probabilities, and $u$ is unique up to affine
%transformations. 
%joe9
%The family of preference orders $\succeq_{M,\cP^+\mid E}$ induced by a
%$\cP^+$ decision problem $(S,\Delta(Y),u,\cP^+)$ for events $E$ such
For all $Y$, $U$, $S$, and $\cP^+$, the family of preference orders
$\succeq_{\cP^+\mid E,M}^{S,Y,U}$ for events $E$ such that $\ucP(E) > 0$ 
satisfies Axioms \ref{axiom:T}--10 and MDC.
Conversely, if a family of preference orders $\succeq_{E,M}$ on the acts in $\Delta(Y)^S$ 
satisfies Axioms \ref{axiom:T}--10 and MDC, then there exists 
a utility $U$ on $Y$ and a weighted set $\cP^+$ of probabilities 
%joe22
%y
%on $S$, with $C(\cP^+)$ being regular, %sam22 used 'regular'
%such that 
on $S$ such that $C(\cP^+)$ is regular, and %sam23 fixed comment-out by mistake?
for all non-null $E$, $\succeq_{E,M} = \succeq_{\cP^+\mid E,M}^{S,Y,U}$.
Moreover, $U$ is unique up to 
affine transformations, and $C(\cP^+)$ is unique in the sense that if $\cQ^+$ represents $\succeq_{E,M}$, and $C(\cQ^+)$ is regular, then $C(\cQ^+) = C(\cP^+)$.
\end{myTheorem}
\begin{proof}
%joe2*: this proof needs to be rewritten.  First prove that MDC is sound.
%joe4: cut; this didn't help
%The main step is to note that the conditional (on $E$) regret of any act
%$f$ is precisely equal to the $E$ part of the unconditional regret of
%act $f$; hence if everything is equivalent off of $E$, unconditional
%regret is determined by conditional regret on $E$. A more detailed
%argument follows. 

%joe2: did some rewriting already
%Since the AX-MWER axioms are satisfied, unconditional preferences have a
%MWER representation, that is, for every $M\subseteq L$, 
%joe3
%Since $\succeq_M = \succeq_{S,M}$ satisfies AX-MWER, there must exist a
Since $\succeq_M = \succeq_{S,M}$ satisfies Axioms
\ref{axiom:T}--10, there must exist a 
weighted set $\cP^+$ of probabilities on $S$ 
%joe4
and a utility function $U$
such that $f\succeq_M g$ iff
%joe9
%$f \succeq_{M,\cP^+}^{(S,X,u)} g$.
$f \succeq_{M,\cP^+}^{S,Y,U} g$. %sam12
%\begin{align*} 
%f\succeq_M g \Leftrightarrow WREG_{\cP,u}(f,M) \leq WREG_{\cP,u}(g,M)
%\end{align*}
%
%Where $WREG_{\cP^+,u}(f,M)$ denotes the maximum weighted regret of act
%$f$, with respect to $\cP^+,u,M$: 
%\begin{align*}
%&WREG_{\cP^+,u}(f,M) = \\
%&\max_{\Pr\in\cP}\left( \alpha_{\Pr} \sum_{s\in S}\left[ \left(
%\max_{h\in M}u(h(s))- u(f(s))\right) \Pr(s) \right]\right) 
%\end{align*} where $\cP^*$ is a set of weighted probabilities, and $u$
%is unique up to affine transformations. 
%
%Similarly $WREG_{\cP^+\mid E,u}(f,M)$ denotes the conditional maximum
%weighted regret of act $f$, with respect to $\cP^+\mid E,u,M$: 
%\begin{align*}
%&WREG_{\cP^+\mid E,u}(f,M) = \\
%&\max_{\Pr\in\cP}(  \frac{\alpha_{\Pr}\Pr(E)}{\max_{\Pr'\in\cP}
%\alpha_{\Pr'}\Pr'(E)} \sum_{s\in S} \left( \max_{h\in M}u(h(s))-
%u(f(s))\right) \Pr(s) ) 
%\end{align*} 
%
%To simplify notation, similarly denote the conditional (on $E$) maximum
%regret of act $f$, with respect to menu $M$, as $REG_{\cP^+\mid E,
%u}(f,M)$. 
%
%\begin{align*}
%&REG_{\cP^+\mid E, u}(f,M) = \\
%&\max_{\Pr\in\cP}\left( \Pr(E)\sum_{s\in E}\left[ \left( \max_{h\in
%M}u(h(s))- u(f(s))\right) \Pr(s|E) \right]\right) 
%\end{align*}
%joe8: removed paragraph break.  You should try to avoid lots of short
%choppy paragraphys 
%
We now show that if $E$ is non-null, then $\ucP(E) > 0$, and 
$f\succeq_{E,M} g$ iff $f \succeq_{M,\cP^+\mid E}^{(S,X,u)} g$.  

%joe2*: don't just say ``one can verify''.  Verify!
%joe4: new argument
%Any event that is non-null must have positive
%probability under at least one 
%prior, thus ensuring that \emph{likelihood updating} is defined on all
%non-null events. 
For the first part, it clearly is equivalent to show that if $\ucP(E) =
0$, then $E$ is null.  
So suppose that
$\ucP(E) = 
0$.  Then $\alpha_{\Pr} \Pr(E)=0$ for all $\Pr\in\cP$.  This means that 
$\alpha_{\Pr} \Pr(s) = 0$ for all $\Pr \in \cP$ and $s \in E$.  Thus, for all
acts $f$ and $g$, 
$$\begin{array}{lll}
&\regret_{M,\cP^+}(f E g)\\
%joe17*: replaced \max by \sup, and used new abbreviation
%=& \max_{\Pr\in\cP} \left( \alpha_{\Pr} \sum_{s\in S}\Pr(s) \left(
%\max_{h\in M}u(h(s))- u(fEg(s))\right)\right)\\ 
%= &\max_{\Pr\in\cP}\left( \alpha_{\Pr}\sum_{s\in E} \Pr(s) \left(
%\max_{h\in M}u(h(s))- u(f(s))\right)\right)  \\ 
%&+ \left. \alpha_{\Pr} \sum_{s\in {E^c}}\Pr(s)  \left( \max_{h\in
%M}u(h(s))- u(g(s))\right)  \right) \\ 
%= &\max_{\Pr\in\cP}\left( \alpha_{\Pr}\sum_{s\in S} \Pr(s)  \left(
%%%sam4: changed s\in E to s\in S 
%  \max_{h\in M}u(h(s))- u(g(s))\right) \right)\\
=& \sup_{\Pr\in\cP} \left(\alpha_{\Pr} \sum_{s\in S}\Pr(s)
\regret_M(fEg,s)\right) \\ 
= &\sup_{\Pr\in\cP}\ \left(\alpha_{\Pr} \left( \sum_{s\in E} \Pr(s) 
\regret_M(f,s) \right.  \right)\\ 
&+ \left. \sum_{s\in {E^c}}\Pr(s)  \regret_M(g,s)   \right) \\ 
= &\sup_{\Pr\in\cP} \left(\alpha_{\Pr}\sum_{s\in S} \Pr(s)  
\regret_M(g,s)\right)\\
= &\regret_{M,\cP^+}(g). %sam4: should be g rather than f E g
\end{array}
$$
Thus, $f E g \sim_M g$ for all acts $f, g$ and menus $M$ containing
$fEg$ and $g$, which means that $E$ is null.  

%joe4*: cut all this; the argument above is simpler.
%\commentout{
%Note that 
%\begin{align}
%%&\exists f,g\in \Delta(Y)^S, M\supseteq \{f,g\}, fEg \succ_M g \\
%&fEg \succ_M g \\
%\Leftrightarrow & \max_{\Pr\in\cP}\left( \alpha_{\Pr} \sum_{s\in S}\left[ \left( \max_{h\in M}u(h(s))- u(fEg(s))\right) \Pr(s) \right]\right) \\
%& < \max_{\Pr\in\cP}\left( \alpha_{\Pr}\sum_{s\in S}\left[ \left( \max_{h\in M}u(h(s))- u(g(s))\right) \Pr(s) \right]\right)\\
%\Leftrightarrow &  \\
%&+ \left. \alpha_{\Pr} \sum_{s\in {E^c}}\left[ \left( \max_{h\in M}u(h(s))- u(g(s))\right) \Pr(s) \right]\right) \\
%& < \max_{\Pr\in\cP}\left( \alpha_{\Pr}\sum_{s\in S}\left[ \left( \max_{h\in M}u(h(s))- u(g(s))\right) \Pr(s) \right]\right)
%\end{align*}
%If $\alpha_{\Pr} \Pr(E)=0$ for all $\Pr\in\cP$ then the LHS and RHS would be identical, so it must be the case that $\alpha_{\Pr} \Pr(E)>0$ for some $\Pr\in\cP$.
%}

%On to the main point. 
%joe4: redundant
%As before, we use $MEh$ to denote the menu $\{ fEh : f\in M\}$. 

%joe2*: rewrite this using the notation I introduced, which should
%simplify it considerable.
%joe4
%It is straightforward to show that for any $f,h\in M$, non-null event $E$, 
%$\regret_{MEh,\cP^+}(fEh) = \ucP(E) \regret_{M,\cP^+\mid E}(f)$:
For the second part, we first show that if $\ucP(E) > 0$, then 
for all $f,h\in M$, we have that
$$\regret_{MEh,\cP^+}(fEh) = \ucP(E) \regret_{M,\cP^+\mid E}(f).$$
We proceed as follows: %sam4
%joe8: generally, things look better with array
%\begin{align*}
$$\begin{array}{lll}
&\regret_{MEh,\cP^+}(fEh) \\
%joe4: it looks better aligned this way
%&= \min_{\Pr\in\cP}\left( \alpha_{\Pr} \sum_{s\in S}\left[ \left( 
%joe17*: using new notation.  Also, you used min.  I assume that you
%meant max (and it should be sup)
%= &\min_{\Pr\in\cP}\\ %sam12 changed to fit in UAI column
= &\sup_{\Pr\in\cP}
%joe17
%&\left( \alpha_{\Pr} \sum_{s\in S} \Pr(s) \left
%\max_{\{gEh : g\in M\}}u(gEh,s)- u(fEh,s)\right)  \right) \\
\left(\alpha_{\Pr} \sum_{s\in S} \Pr(s) 
\regret_{MEh}(fEH,s) \right)\\
%joe4
%&= 
%joe17
%=& \min_{\Pr\in\cP} \\ %sam12
%&\left( \alpha_{\Pr}\Pr(E) \sum_{s\in E}\Pr(s \mid E)  
%\left( \max_{g\in M}u(g,s)- u(f,s) \right)  \right. \\
=& \sup_{\Pr\in\cP} 
\left( \alpha_{\Pr}\Pr(E) \sum_{s\in E}\Pr(s \mid E)  
\regret_M(f,s)\right. \\
%joe17
%	&\left.+ \alpha_{\Pr}\sum_{s\in E^c} \Pr(s) \left( u(h,s)- u(h,s)
%  \right)   \right) \\ 
	&\left.+ \alpha_{\Pr}\sum_{s\in E^c} \Pr(s) \regret_{\{h\}}(h,s)
	\right)\\
%joe4: you can't put text like this in the middle of a string of
%equalities; in any case, it's OK without it.
%&\text{Note the $E^c$ part of the regret is $u(h,s)-u(h,s)=0$ }\\
%joe4
%&= 
=& 
%joe17
%\min_{\Pr\in\cP} \\ %sam12
%&\left( \alpha_{\Pr}\Pr(E) \sum_{s\in E}\Pr(s|E)  \left( \max_{g\in
%M}u(g(s))- u(f,s) \right)  \right)\\ 
\sup_{\Pr\in\cP} 
\left( \alpha_{\Pr}\Pr(E) \sum_{s\in E}\Pr(s|E)  \regret_M(s,f)\right)\\
= &\sup_{\Pr\in\cP} 
\left( \ucP(E) \alpha_{\Pr,E} \sum_{s\in
E}\Pr(s|E)  \regret_M(f,s) \right)\\
&\mbox{[since $\alpha_{\Pr,E} = \sup_{\{\Pr'\in\cP : \Pr'\mid E =
\Pr\mid E\}} \frac{\alpha_{\Pr'}\Pr'(E)}{\ucP(E)}$]}\\ 
%\min_{\Pr\in\cP} \\ %sam12
%&\left( \alpha_{\Pr}\Pr(E) \sum_{s\in E}\Pr(s|E)  \left( \max_{g\in
%M}u(g(s))- u(f,s) \right)  \right)\\ 
%= &\min_{\Pr\in\cP} \\ %sam12
%& \ucP(E) \left( \alpha_{\Pr \mid E} \sum_{s\in
%E}\Pr(s|E)  \left( \max_{g\in M}u(g(s))- u(f,s) \right)  \right)\\
%&\mbox{[since $\alpha_{\Pr \mid E} = \frac{\alpha_{\Pr}\Pr(E)}{\ucP(E)}$]}\\
=& \ucP(E) \cdot \regret_{M,\cP^+\mid E}(f). 
%joe8
%\end{align*}
\end{array}$$
%joe8: unecessary now
%by definition of   $\cP^+\mid E$, 
%since the updated weight of $\Pr$ is $\alpha_{\Pr}\Pr(E)/\ucP(E)$.

%joe4
%Thus we have, for all $h\in M$,
Thus, for all $h\in M$,
\begin{align*}
%joe4
%&fEh \succeq_{MEh} gEh \\
%$\Leftrightarrow 
& \regret_{MEh,\cP^+}(fEh)  \leq \regret_{MEh,\cP^+}(gEh) \\
%joe5
%\Leftrightarrow& \ucP(E)\cdot \regret_{M,\cP^+\mid E}(f) \leq
\mbox{ iff }& \ucP(E)\cdot \regret_{M,\cP^+\mid E}(f) \leq
\ucP(E)\cdot \regret_{M,\cP^+\mid E}(g)\\ 
%joe4
%\Leftrightarrow& \regret_{M,\cP^+\mid E}(f) \leq \cdot
%\regret_{M,\cP^+\mid E}(g) 
%joe55
%\Leftrightarrow& \regret_{M,\cP^+\mid E}(f) \leq 
\mbox{ iff }& \regret_{M,\cP^+\mid E}(f) \leq 
\regret_{M,\cP^+\mid E}(g).
\end{align*}
%joe4
It follows that the order induced by $\cP^+$ satisfies MDC.

%joe4
%So MDC is true if and only if 
%joe20*: it may not be unique
%Moreover, if \ref{axiom:T}--10 and MDC hold, then for the
Moreover, if \ref{axiom:T}--10 and MDC hold, then for a
weighted set $\cP^+$ that represents $\succeq_{M}$, we have
$$\begin{array}{ll}
%joe4
%f \succeq_{E,M} g &\Leftrightarrow \forall h\in M, fEh \succeq_{MEh} gEh\\ 
%&\Leftrightarrow \regret_{M,\cP^+\mid E}(f) \leq \regret_{M,\cP^+\mid
%E}(g)
&f \succeq_{E,M} g \\
%joe5
%\Leftrightarrow &\mbox{ for some } h\in M, fEh \succeq_{MEh} gEh\\ 
%\Leftrightarrow &\regret_{M,\cP^+\mid E}(f) \leq \regret_{M,\cP^+\mid
\mbox{ iff } &\mbox{ for some } h\in M, fEh \succeq_{MEh} gEh\\ 
\mbox{ iff } &\regret_{M,\cP^+\mid E}(f) \leq \regret_{M,\cP^+\mid
E}(g),
\end{array}
$$
%Which shows also that MDC is sound for $\succeq_{M,\cP^+\mid
%E}^{(S,X,u)}$.
as desired.
%joe20*: it seems to me that you still have to argue uniquness, don't
%you?  It might follow easily from uniqueness in the unconditional case,
%but you have to say something.

%sam21 : addressing uniqueness 
Finally, the uniqueness of $C(\cP^+)$ follows from
Theorem~\ref{thm:completeness}, which says that the family
$\succeq_{S,M}$ of preferences is already sufficient to guarantee the
uniqueness of $C(\cP^+)$. 
\end{proof}

%joe2*: rewrote, to make it a more general discussion of related work.
\commentout{
Recall that MDC stands for menu-dependent dynamic consistency. 
The term dynamic consistency (DC) refers to the property 
\begin{align}\label{equ:DC}
fEh\succeq gEh \Leftrightarrow f \succeq_E g\tag{DC}
\end{align} from conditional subjective expected utility preferences. 
%Savage's Independence axiom, which can be stated as, for all acts $f,g,h,h'$, 
%\begin{align}\tag{IND}\label{eq:independence}
%fEh \succeq gEh \Leftrightarrow fEh' \succeq gEh'
%\end{align}
%implies that the definition of conditional preference as 
%\begin{align}\tag{DC}\label{eq:savcond}
%f \succ_E g \Leftrightarrow fEh \succ gEh, \forall h\in A
%\end{align} is a preference relation. 
The essence of the condition requires that the unconditional preferences are respected after the event $E$ is observed. 
%It suffices to consider the cases $fEh \succ gEh$, where the two acts are identical off of $E$, because [TODO]. 
%joe1
%DC is also discussed in, e.g., \cite{Epstein1993}.
DC is also discussed by Epstein and Schneider \citeyear{Epstein1993}.

It is interesting that the statement of MDC is almost identical to DC,
except MDC specifies the menus on each side, while DC does not. 
DC does not need to specify menus because the order $\succeq$
represented by subjective expected utility maximization is
menu-independent.  
}
%joe2: \end{commentout}

Analogues of MDC have appeared in the literature before in the context
of updating preference orders.  In particular, Epstein and Schneider
\citeyear{Epstein1993} discuss a menu-independent version of MDC, although
they do not characterize updating in their framework.  
Sinischalchi \citeyear{Siniscalchi2011} also uses an analogue of MDC in his
axiomatization of measure-by-measure updating of MMEU. 
%Furthermore, it is interesting to compare MDC with the axiom in
%Siniscalchi's 
%\cite{Siniscalchi2011} axiomatization of measure-by-measure updating on
%MMEU preferences.   
Like us, he starts with an
%joe8: reordered a bit
%axiomatization for unconditional preferences, and adds an axiom that he
%calls \emph{constant-act dynamic consistency} (CDC) to extend the 
%axiomatization to conditional preferences.  
%sam10 changed since removed CDC
axiomatization for unconditional preferences, and adds an axiom called  
\emph{constant-act dynamic consistency} (CDC), 
%joe10
somewhat analogous to MDC,
to extend the
%axiomatization for unconditional preferences, and adds the following
%axiom, which he 
%calls \emph{constant-act dynamic consistency} (CDC), to extend the 
axiomatization of MMEU to deal with conditional preferences.  
\commentout{
It would be good if, for any event $E$ and any weighted probabilities
$\cP^+$, in $\cP^+ |E$ there is always some distribution with weight
$1$. 
However, as the following example illustrates, even if $\cP^+$ has weight $1$ on some probability distribution, $\cP^+|E$ might not.
\begin{example}
Let $S= \{ s_1,s_2, s_3\}$ and let $E=\{s_1, s_2\}$. 
Let $\succeq_M$ be represented as minimizing weighted expected regret with respect to a utility function, and a set of weighted probabilities, $\cP^+$, where $$\alpha^{\succeq}_{\Pr} = 
\begin{cases}
1 \text{ if } \Pr(s_3) = 1\\
\Pr(s_1|E) \text{ if } \Pr(s_2 |E) \neq 0 \\
0 \text{ otherwise }
\end{cases}.$$
Now consider the updated weighted probabilities $\cP^+ \mid E$.
By definition, $\alpha^{\succeq}_{\Pr \mid E} = \sup_{\{\Pr'\in\cP :
\Pr'\mid E = \Pr\mid E\}} \frac{\alpha^{\succeq}_{\Pr'}\Pr'(E)}{\ucP(E)}$.
In this example, $\ucP(E) = \sup\{\alpha^{\succeq}_{\Pr}\Pr(E): \Pr \in \cP\} = \sup \{ \Pr(s_1|E)) : \Pr\in\cP \} = 1$.
However, there is no $\Pr$ such that $\alpha_{\Pr}\Pr(E) = \Pr(s_1|E) \Pr(E) = 1$, and thus no $\Pr|E$ that has $\alpha_{\Pr|E}= 1$.
\end{example}
}

%joe2: cut; I don't think this adds anything
\commentout{
Theorem \ref{thm:update} shows that any set of conditional preference
relations $\succeq_{E,M}$ satisfies AX-MWER and \ref{eq:mDC} if and only
if these conditional preferences can be represented as MWERL.  
In other words, if the unconditional preferences $\succeq_{S,M}$ satisfy
AX-MWER, then the agent behaves `as-if' she were a MWER agent. 
If the conditional preferences also satisfy \ref{eq:mDC}, then the agent behaves `as-if' she uses \emph{likelihood updating} to update her beliefs. 
Moreover, any MWERL preferences satisfy AX-MWER and \ref{eq:mDC}. 
}
% Removed because I don't think this is consequentialism
%One of the standard, desirable properties of a dynamic decision model is \emph{Consequentialism}(see e.g. \cite{SiniscalchiResponse}): \begin{align}\tag{CON}\label{equ:CON}
%f\sim_{E,M} fEh, \forall M
%\end{align}
%\emph{Consequentialism} says that preferences are independent of things like forgone opportunities, past events, unrealized events, and intended future actions.
%MDC and $monotonicity$ implies immediately that that conditional preferences represented by MWER with \emph{likelihood updating} satisfies \emph{Consequentialism}.

%\begin{proof}
%To see that Consequentialism holds, note that by \ref{eq:mDC},
%$f\succeq_{E,M} g\Leftrightarrow fEh\succeq_MEh gEh$ for all $h\in M$.  
%Thus, for any $M$ containing $f$ and $fEh$, $f\sim_{E,M} fEh \Leftrightarrow fEh \sim_MEh fEh$, where the right hand side can be deduced from monotonicity of $\succeq_M$.
%\end{proof}

%joe11*: I would cut this section.  We really have very little to say.
%I think it will hurt our case, rather than helping it.  (I made some
%changes below anyway, for the full paper.)

\section{Dynamic Inconsistency}
%joe8*: lots of rewriting in this section.  I don't think that the story
%came out clearly at all.
There is an important issue when one attempts to apply MWER with
likelihood updating to dynamic decision problems. 
If you want to execute a plan, at every step you'll need to stick with that plan and execute the corresponding part of the plan. 
However, after following the initial steps of an ex-ante optimal plan, a MWER agent may no longer wish to adhere to the plan. 
%joe8
%In such a situation, the agent is said to have dynamically inconsistent
In such a situation, the agent is said to have \emph{dynamically inconsistent}
preferences. 
%joe11: we need to start out with the big picture
Dyanmic inconsistency is well known to hold for regret.  
%joe11: moved up
Indeed, as Epstein and Le Breton \citeyear{Epstein1993} show,
dynamic inconsistency arises for any non-Bayesian approach to decision making
(i.e., any approach that does not involve maximizing expected utility)
that satisfies certain minimal assumptions.  Not surprisingly, it arises
for MWER as well.  In the rest of this section, we discuss the problem
and some standard approaches to dealing with it, and illustrate some
subtleties that arise in dealing with it in the context of MWER.

%joe8: 
To understand the problem in the context of regret,
consider the two-stage decision problem of having dinner, represented as
a decision tree in Figure \ref{fig:inconsistency_restaurant_1}. 
Solid circles denote decision points, and empty circles denote points where nature reveals information to the agent.
The decision tree also includes information about what states are considered possible at each node. 
The set of states considered possible at the root is always the entire state space, and nature's actions at each nature decision point partitions the set of possible states. 
%A decision tree is defined recursively. 
%Each outcome at a particular state and terminal time is itself a tree. 
%an action available in node (t,$\omega$) maps states in the cake of the partition 

\begin{figure}[h]
	\centering
%joe8: the figure is a bit big
%\includegraphics[width=0.80\textwidth]{inconsistency_restaurant_1} 
\includegraphics[width=0.60\textwidth]{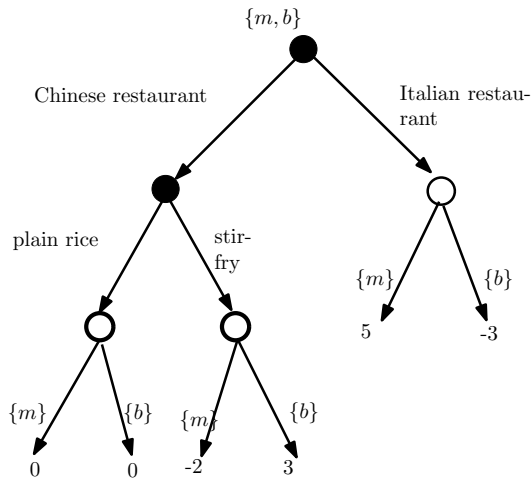} 
	\caption{\label{fig:inconsistency_restaurant_1}Dynamic
	inconsistency example. 
	%Sam10 : removed line to make more room. stated in main text 
	%Empty circles are nature's moves where information is revealed.
	}\label{fig1}
\end{figure}

First, you have to choose between a Chinese restaurant and an Italian
restaurant.  
Once you have arrived at a particular restaurant, you cannot change your mind and go to another; so in the second stage you must order something from the menu at the chosen restaurant. 
Your utility is a combination of how much you enjoy the food, and whether you get an allergic reaction. 
Initially, you know that there are two possible states: you must be either allergic to MSG (state $m$) or to basil (state $b$), but not both. 
Assume that all Italian foods will have traces of basil, and Chinese
%joe8
%stir-fry has MSG but plain rice does not (however, you do not enjoy
%eating plain rice, hence the utility of ordering rice is $0$).  
stir-fry has MSG but plain rice does not.  However, you do not enjoy
eating plain rice, so the utility of ordering rice is $0$.

Suppose that you make decisions using the minimax regret decision rule,
%joe8: I have no idea what this is
%We will use the textbook method of reducing plans to acts.
viewing a plan (i.e., a strategy) as an act.
%joe8: 
%If we compute the regrets for each plan, we'll find that going to the
A straightforward computation shows that, \emph{ex ante}, going to the
Chinese restaurant and ordering plain rice has the lowest regret ($5$).  
%joe8
%However, note that if you did go to the Chinese restaurant, the choice
%of going to the Italian restaurant is now irrelevant, and while you
%decide on what to order at the Chinese restaurant, you'll see that the
%regret of ordering stir-fry is lower ($2$) than that of ordering rice
%($3$).  
However, if you go to the Chinese restaurant, the choice
of going to the Italian restaurant is now irrelevant.  If we now
compute regret with respect to the menu of the two remaining choices,
then the regret of ordering stir-fry is lower ($2$) than that of
ordering rice ($3$).  
%joe8: I'm not sure why this is ``inevitable''.
%Inevitably, you will order the stir-fry.
%Therefore, there is no way you can actually go to a Chinese restaurant
%and order plain rice, even though it was initially the optimal plan. 
You thus end up ordering the stir-fry.
The plan of going to the Chinese restaurant and ordering plain rice
cannot be carried out.

%joe8
%This is an example of dynamic inconsistency, a common problem for many
%non Bayesian (expected utility maximization) decision rules.  
%Dynamic consistency is the notion that the plan considered optimal at a
%given point in the decision process is also optimal at any preceding
%point in the process, as well as any future point that is reached with
%positive probability \cite{Siniscalchi2011}.  
%When dynamic consistency fails to hold, an ex-ante optimal plan may not
%be favored by future instances of the agent, and hence be impossible to
%carry out.  
More generally, dynamic consistency requires that the plan
considered optimal \emph{ex ante} continues to be considered optimal at
any later time.  
%joe11: shortened
%As Epstein and Le Breton \citeyear{Epstein1993} show, under minimal
%assumptions, such dynamic
%inconsistency arises for any non-Bayesian approach to decision making
%(i.e., any approach that does not involve maximizing expected utility).
As we said earlier, Epstein and Le Breton \citeyear{Epstein1993} show
that dynamic inconsistency will arise for essentially all non-Bayesidan
decision rules.
A standard approach for dealing with this lack of dynamic consistency in
the literature is to consider `sophsticated' agents, who are aware of
the potential for dynamic inconsistency, and thus use backward induction
to determine the feasible plans.  
In the restaurant example, a sophisticated
agent believes correctly that she will prefer stir-fry over rice, once
she is at the Chinese restaurant.  Therefore, she no longer considers
%joe8
%going to the Chinese restaurant and ordering plain rice as a viable plan.  
%Therefore, the only remaining options are going to the Italian
going to the Chinese restaurant and ordering plain rice a viable plan.  
The only feasible options are going to the Italian 
restaurant, or having stir-fry at the Chinese restaurant.

%joe8*: I don't think that this is quite the right story.
%However, a subtlety exists in sophistication when preferences over plans
%are menu-dependent. 
%In the example, how should the agent rank the two remaining plans, given
%that the plan to have rice is deemed infeasible?  
%In other words, how should the agent compute the regrets of each of the
%remaining plans?  
%There are two obvious answers: either (1) use the original ordering,
%which implies that regrets are computed with respect to the original
%set of available plans; or (2) compute new regret values with
%infeasible plans excluded from the menu.  
A subtlety arises when trying to apply backward induction to
menu-dependent decision rules: which menu do we use when comparing the
viable plans?  For example, in the restaurant example, do we use the
menu consisting of all three plans, or the menu consisting of just the
viable plans.  It turns out not to matter in this example---with respect
to both menus, going to the Italian restaurant minimizes regret.
%joe8: cut; this is overkill
%In this case, including or excluding the plain rice does not change the
%regrets of the two remaining plans.  
%Between these two plans, going to the Italian restaurant has a maximum
%regret of $6$, which is lower than that of having stir-fry ($7$).  
%As a result, the sophisticated agent would choose to go to the Italian
%restaurant.  
%
%However, in general the regrets of the remaining plans may be different
%depending on whether the `eliminated' plans are included in the regret
%computation.  
%Perhaps it is more natural for the agent to use the larger menu, with
%the infeasible plans included, since the agent can still imagine
%following these plans. 
%However, it can also be argued that the agent should not regret against
%plans that he knows could not be carried out. 
%Moreover, one can further ask why the agent is not regretting against
%other plans that she can imagine, that aren't even in the original
%problem description. 
%Hayashi \citeyear{Hayashi2009} uses the latter assumption for minimax
%expected regret agents in optimal stopping problems. 
However, in general, the choice of menu can matter.  Hayashi
\citeyear{Hayashi2009} uses the menu of viable plans in computing for
minimax  expected regret agents in optimal stopping problems, 
but it seems to us that both choices (and perhaps others) can be
justified.  

%Under this assumption, backward induction involving menu-dependent
%preferences is different from that involving menu-independent
%preferences, in the sense that the backward induction solution depends
%on ... ?? 

%joe8: trying to get a more coherent stroy
%An assumption that needs to be made, regardless of menu dependence of
%preferences, is the behavior of the agent when two or more choices are
%considered indifferent.  
A second subtlety that arises when considering sophisticated agents:
%joe11
%what choice do they make when they are indifferent between two plans.
What choice do they make when they are indifferent between two plans?
Sinischalchi \citeyear{Siniscalchi2011}  axiomatizes \emph{consistent
planning} (with menu-independent preferences over plans), which augments 
backward induction with a tie-breaking assumption. 
%joe8
%This tie-breaking assumption in consistent planning is that the agent
%can commit to a plan as long as each stage of the plan is considered to
This tie-breaking assumption in consistent planning allows an agent
to commit to a plan as long as each stage of the plan is considered to
be one of the best at each local decision node.  

%joe8*: now it really starts to meander
%As Siniscalchi points out, preferences over plans are in general not
%enough to reveal the agent's sophistication. 
%Therefore, Siniscalchi considers preferences over a set of decision
%trees, where a plan is a special case of a decision tree. 
In order to axiomatize consistent planning, Siniscalchi must assume that
the agent has preferences that are more general than preferences over
plans.  Rather, the agent must be assumed to have preferences over
\emph{decision trees} (such as that in Figure~\ref{fig1}).  Plans are the
% sam8 
%special case of decision trees with no branching; we can identify a
special case of decision trees with no branching at decision nodes; we can identify a
decision tree with a set of plans (essentially, the branches in the
decision tree).
Sophistication is captured by an axiom that says, roughly, that the
agent is indifferent between a decision tree and the same tree with a
non-optimal (based on backward-induction) plan removed.  
%joe8
%This approach is in line with the menu choice literature dating back to
%Kreps \citeyear{kreps79}.  
%Preferences over decision trees is analogous to preferences over menus
%of choices. 
Preferences over decision trees are similar in spirit to preferences
over menus \cite{kreps79}.

If we try to apply Siniscalchi's approach to regret, we encounter
further difficulties.  In a menu-independent setting, we can compare two
decision trees by comparing the best plans in each decision tree (if we
identify a decision tree with a set of plans).  But once menus become
relevant, we must decide what menu to use when making this comparison.  
It is not clear which menu to choose.  What we really have here are
menus over menus; it is not even clear how to apply regret in this
setting.  Defining and axiomatizing consistent planning in a
regret-based setting remains an open problem.

\section{Conclusion}\label{sec:conclusion}
We proposed an alternative belief representation using \emph{weighted
  sets of probabilities}, and described a natural approach to updating
in such a situation and a natural approach to determining the weights.   
%joe8
%We also showed how weighted sets of probabilities can be used with MER
%to obtain the MWER decision rule, and provided an axiomatization that
We also showed how weighted sets of probabilities can be combined with regret
to obtain a decision rule, MWER, and provided an axiomatization that
characterizes static and dynamic preferences induced by MWER.  

%joe4
We have considered preferences indexed by menus here.
Stoye~\citeyear{Stoye2011} used a different framework: \emph{choice
functions}.  A choice function maps every finite set $M$ of acts to a subset
$M'$ of $M$.  Intuitively, the set $M'$ consists of the `best' acts in
$M$.  Thus, a choice function gives less information than a preference
order; it gives only the top elements of the preference order.  
The motivation for working with choice functions is that an agent can reveal
his most preferred acts by choosing them when the menu is offered.  In a
menu-independent setting, the agent can reveal his whole preference order;
to decide if $f \succ g$, it suffices to present the agent with a choice
among $\{f,g\}$.  However, with regret-based choices, the menu matters;
the agent's most preferred choice(s) when presented with $\{f,g\}$ might
no longer be the most preferred choice(s) when presented with a larger
menu.
Thus, a whole preference order is arguably not meaningful with
regret-based choices.
Stoye~\citeyear{Stoye2011} provides a representation theorem for MER where the
axioms are described in terms of choice functions.  The axioms that we
have attributed to Stoye are actually the menu-based analogue of his
axioms.  We believe that it should be possible to provide a
characterization of MWER using choice functions, although we have not
yet proved this.

%joe4*: rewrote
%joe8: rewrote and shortened; it is actually not what we do.  
\commentout{
Another issue that must be dealt with when MWER is combined with likelihood
updating is \emph{dynamic inconsistency}.  It is not hard to give
examples of settings where an agent that uses MWER combined with
likelihood updating starts out with one plan, and then changes his
mind.  It is not hard to construct examples where an agent decides on a
plan that says that if he learns $E$, he should perform act $f$, but then
when he actually learns $E$, he performs $f'$ instead.  Such dynamic
inconsistency arises with MMEU when using measure-by-measure updating as
well.  
Siniscalchi \citeyear{Siniscalchi2011} proposes an approach to dealing with
dynamic consistency in the context of MMEU combined with
measure-by-measure updating by using backward induction to decide which
action to take.  As we show in the full paper, this approach can be
applied to ensure dynamic consistency for MWER combined with likelihood
updating as well. 
}
Finally, we briefly considered the issue of dynamic consistency and
consistent planning.  As we showed, making this precise in the context
of regret involves a number of subtleties.  We hope to return to this
issue in future work.

%Recall that the axiom that characterizes \emph{likelihood updating} on MER preferences, MDC, stands for menu-dependent dynamic consistency. 
%As the name implies, MDC is a weakening of dynamic consistency. 

%joe16: moved the references after the appendix, which is somewhat more
%standard 

\appendix

\section{Proof of Theorem~\protect{\ref{thm:completeness}}}

%joe10
%As the soundness of the axioms is straightforward to prove, we only show
%joe15
%As the soundness of the axioms is straightforward to prove, we prove
%only completeness of the axioms.  
%That is, we show that
We show here that
if a family of menu-dependent preferences $\succeq_M$ satisfies
axioms 1-10, then $\succeq_M$ can be represented as minimizing expected
regret with respect to a set of weighted probabilities and a utility
function. 
% TODO should have a several paragraph long 'map' of the proof
%joe10:
%The completeness proof is simplified by Lemmas \ref{lem:vonNeumann} to
%\ref{lem:stoye}. 
%joe10
Since the proof is somewhat lengthy and complicated, 
%sam14 added 'a' 
%we split it into several steps, each in separate subsection.
we split it into several steps, each in a separate subsection.

\subsection{Simplifying the Problem }
%joe10: moved from above; you just gave an outline for this part.
Our proof starts in much the same way as the proof by
Stoye~\citeyear{Stoye2011} of a 
representation theorem for regret.
Lemma \ref{lem:vonNeumann} guarantees the existence of a utility
%joe8
%function $U$ on prizes, that can be extended to lotteries in the obvious
%way such that  
function $U$ on prizes that can be extended to lotteries in the obvious
way, so that  
$l^* \succeq (l')^*$ iff $U(l) \ge U(l')$. 
In other words, preferences over all constant acts are represented by the maximization of $U$ on the corresponding lotteries that the constant acts map to. 
Lemma~\ref{lem:vonNeumann} is a consequence of standard results.
Our menus are arbitrary sets of acts, as opposed to convex hulls of a
finite number of acts in \cite{Stoye2011}; 
Lemma~\ref{lem:cases} shows that Stoye's technique can be adapted to
work when menus are arbitrary sets of acts. 
%joe10
%Lemma~\ref{lem:stoye} is a reproduction of the observations made by
%Stoye \citeyear{Stoye2011} (proof of Theorem 1 (iv)).  
Finally, following Stoye \citeyear{Stoye2011}, we 
reduce the proof of existence of a minimax
%joe8
%weighted regret representation for the family $\succeq_M$, to the proof
weighted regret representation for the family $\succeq_M$ to the proof
of existence of a minimax weighted regret representation for a
%joe8
%\emph{single} preference ordering $\succeq$.  
single menu-independent preference ordering $\succeq$
(Lemma~\ref{lem:stoye}).

%joe9: instead of subsections, we need lemmas, and a bit of a story
%\subsubsection{Getting a utility function for constant acts (lotteries)}
%As we said, our goal here is to show that we can reduce to considering
%one menu, rather than a family of menus.  Our first step is to prove a
%standard result, showing that we can find a utility function on prizes
%that represents the order.

\begin{lemma}\label{lem:vonNeumann}
%joe10
%If Axioms 1-3, 5, 7,8 hold, then there exists a nonconstant 
If Axioms 1-3, 5, 7, and 8 hold, then there exists a nonconstant 
%joe8: moved
%(and unique up to positive affine transformations) 
function $U : X\rightarrow \R$,
%joe8: moved here
unique up to positive affine transformations,
%joe8
%such that for all constant acts $l^*$ and $(l')^*$, and menus $M$, 
such that for all constant acts $l^*$ and $(l')^*$ and menus $M$, 
$$ l^* \succeq_M (l')^* \Leftrightarrow \sum_{ 
%joe8: added a little space
%\{y : l^*(y)>0 \}} l^*(y) U(y) \geq \sum_{ \{y : (l')^*(y)>0 \}}
%joe9: the thype here is wrong; l^* maps states to lotteies; y is not a
%state 
%\{y :\, l^*(y)>0 \}} l^*(y) U(y) \geq \sum_{ \{y :\, (l')^*(y)>0 \}};
%(l')^*(y) U(y) .$$ 
\{y :\, l^*(y)>0 \}} l(y) U(y) \geq \sum_{ \{y :\, l'(y)>0 \}}
l'(y) U(y) .$$ 
\end{lemma}
\begin{proof}
%joe8
%The family of preferences $\succeq_M$, where $M$ is any menu, are all
%just the same preference order by using menu-independence for constant
%acts. 
By menu independence for constant acts, the family of preferences
$\succeq_M$ all agree when restricted to constant acts.
The lemma then follows from standard results (see, e.g.,
%joe8: using joe.bib version
%\cite{Kreps1988}), since menu-independence for constant acts, combined
\cite{Kreps}), since menu-independence for constant acts, combined
with independence, gives the standard independence (substitution) axiom
from expected utility theory. 
\end{proof}

%joe8
%As is commonly done, from $U$ we define $u$ to be be the expected
As is commonly done, given $U$, we define $u(l) = \sum_{\{y:\, l(y) > 0\}}
l(y) U(y)$.  Thus, $u(l)$ is the expected utility of lottery $l$.
We extend $u$ to contsant acts by taking $u(l^*) = u(l)$.
%joe9
%That is, for all constant acts $l^*$ and $(l')^*$ and menus $M$
%containing $l^*$ and $(l')^*$, $u(l^*)\geq u((l')^*) \Leftrightarrow
%l^*\succeq_M (l')^*$. 
Thus, Lemma~\ref{lem:vonNeumann} says that, for all menus $M$, 
$l^* \succeq (l')^*$ iff $u(l^*) \ge u(l')$.  
%joe9*: move dfrom below (where it was out of place
If $c$ is the utility of some lottery, let $l^*_c$ be a 
constant lottery that $u(l^*_c) = c$.  The following is now
immediate. We state it as a lemma so that we can refer to it later.
\begin{lemma}\label{property1}
$u(l^*_c) \ge u(l^*_{c'})$ iff $l^*_c \succeq l^*_{c'}$;
similarly, $u(l^*_c) =u(l^*_{c'})$ iff $l^*_c \sim l^*_{c'}$, and
$u(l^*_c) > u(l^*_{c'})$ iff $l^*_c \succ l^*_{c'}$.
\end{lemma}

%joe9*: you need a story here.  Why doesn't Stoye have the same problem
%that we have?  
%\subsubsection{Reducing to a single menu}
%Here we reuse an observation made by Stoye \citeyear{Stoye2011} to
%further simplify our problem: it suffices to find a representation for a
%single preference order on nonpositive acts \citeyear{Stoye2011}, 
\medskip

The key step in showing that we can reduce to a single menu is to show
that, roughly speaking, for each menu, there exists a menu-dependent
function $g_M$ such that $u(g_M(s)) = -\sup_{f \in M} u(f(s))$.  
Stoye \citeyear{Stoye2011} proved a similar result, but he assumed that
all menus were obtained by taking the convex hull of a finite set of
acts.  Because we allow arbitrary bounded menus, this result is not
quite true for us.  For example, suppose that the range of $u$ is
$(-1,\infty]$.  Then there may be a menu $M$ such that $\sup_{f \in M}
u(f(s)) = 5$, so $-\sup_{f \in M} u(f(s)) = -5$.  But there is no
act $g$ such that $u(g(s)) = -5$, since $u$ is bounded below by $-1$.
The following weakening of this result suffices for our purpose.

%joe8*: Sam, this is clearly false once you've fixed u.  Even in the 
%case that utilities are bounded above and below, you need to modify the
%utility so that the range is [-1,1].  I stopped making changes here.  I
%think you need to state a clear lemma, making the quantification clear.
%But even if you modify the utility so that the range is (-\infty,1] in
%case 3, the lemma is false as stated.  Not surprisingly, the proof of
%case 3 does not seem to show exactly what you claimed.  I would like
%you to state a lemma clearly, and prove it.

\begin{lemma}\label{lem:cases}
There exists a utility function $U$ such that for every menu $M$, 
%joe9: too complicated, and g is overloaded
%there exists a functional $t: \Delta(Y)^S \Rightarrow \Delta(Y)^S$ such
%\succeq_{t(M)} t(g)$. 
%sam14 added 'such'
there exists $\epsilon \in (0,1]$ and constant act $l^*$ such 
that for all $f,g \in M$, $f \succeq_M g \Leftrightarrow t(f)
\succeq_{t(M)} t(g)$, where 
%sam14 replaced .' by ',' 
%joe14: removed , altogether
%$t$ has the form $t(f) = \epsilon f + (1-\epsilon)l^*$, 
$t$ has the form $t(f) = \epsilon f + (1-\epsilon)l^*$ 
%joe9
%for some  $\epsilon >0$ and act $g$, and $t(M)=\{t(f) : f \in M\}$. 
%Moreover, there exists an act $g_{t(M)}$ that satisfies 
%$$ \forall s\in S, u(g_{t(M)}(s) ) = -\sup_{f\in t(M)}u(f(s)). $$
%sam14
and $t(M)=\{t(f) : f \in M\}$. 
Moreover, there exists an act $g_{t(M)}$ such that 
$u(g_{t(M)}(s) ) = -\sup_{f\in t(M)}u(f(s))$ for all $s \in S$.
\end{lemma}
\begin{proof}
The nontriviality and monotonicity axioms imply there must exist prizes
$x$ and $y$ such that $U(x) > U(y)$.   
%joe9: out of place
%It is convenient to note that for any $c$ which is the utility of some
%act, there exists a constant act $l^*_c$ such that $u(l^*_c) = c^*$. 
%joe9
%We consider four possible cases.
We consider four cases.

Case 1: 
%joe9*: why do we need l^*_0? In fact, why do we need l^*_c?
%The set of utilities is bounded above and below. 
%Then we can take the range of utilities to be $[-1,1]$.
%Thus there are prices $x$ and $y$ such that $U(x)=1$ and $U(y)=-1$. 
%$l^*_0$ can be defined by $\frac{1}{2}x + \frac{1}{2}y$. 
%For any $c\in [-1,1]$, $l^*_c$ can be defined by taking convex
%combinations of $l^*_0$ and $x$ or $y$.  
%So $t$ is the identity function and for every menu $M$, $g_M$ always exists.
The range of $U$ is bounded above and below.  Then we can rescale so
that the range of $U$ is $[-1,1]$.
Thus, there must be prizes $x$ and $y$ such that $U(x)=1$ and $U(y)=-1$. 
For all $c \in [-1,1]$, there must be a prize $x'$ that is a convex
combination of $x$ and $y$ such that $u(x') = c$, so we can clearly define
a function $g_M$ such that, for all $s \in S$, we have 
$u(g_{M}(s) ) = -\sup_{f\in M}u(f(s))$.  
Furthermore, we know that such a $g_M$ exists because it can be formed as an act which maps each state to an appropriate lottery over the prizes $x$ and $y$. 
More generally, we know that an act with a certain utility profile exists if its utility for each state is within the range of $U$.
This fact will be used in the other cases as well.

Thus, in this case we can take $t$ to be the identity (i.e., $\epsilon = 1$).

Case 2: 
%joe9
%The utilities are unbounded above and below. 
The range of $U$ is $(-\infty,\infty)$.  
%joe9: not quite
%Then $l^*_c$ can be defined in the same way as in Case 1, and for every
%menu $M$, $g_M$ must exist. 
Again, for all $c \in (\infty,\infty)$, there must exist a prize $x$
such that $u(x) = c$.  Since menus are assumed to be bounded above, we can
again define the required function $g$ and take $\epsilon = 1$.

Case 3: 
%joe9
%The utilities are bounded above, and unbounded below. 
%Then we can assume that the maximum utility is $1$, so that $l^*_c$
%exists for all $c \leq 1$. 
The range of $U$ is bounded above and unbounded below.
Then we can assume without loss of generality that the range is
%joe14
%$(\infty,1]$, and for all $c$ in the range, there is a prize $x$ such
$(-\infty,1]$, and for all $c$ in the range, there is a prize $x$ such
that $u(x) = c$.
%joe9
%For every menu $M$, $\epsilon >0$, and all acts $f$ and $g$ in $M$, by
%the independence axiom,  
%joe14: uncommented the next two lines
For all menus $M$, $\epsilon >0$, and acts $f,
g \in M$, by  Independence,   we have that
$$ f \succeq_M g \Leftrightarrow \epsilon f + (1-\epsilon) l_1^* \succeq_{\epsilon M + (1-\epsilon) l_1^* } \epsilon g + (1-\epsilon) l_1^*.$$
%joe9
%We can always choose $\epsilon$ sufficiently small such that for every
%joe14
%We can find an $\epsilon > 0$ sufficiently small that for all
There exists an $\epsilon > 0$ such that for all
$s\in S$,  
%$\max_{f \in M} \epsilon u(f(s)) + (1-\epsilon) u(l_1^*(s)) \geq -1$. 
%sam8
%joe9
%$1 \geq \sup_{f \in M} \epsilon u(f(s)) + (1-\epsilon) u(l_1^*(s)) \geq
%-1$.  Therefore, we can let $t(f) = \epsilon f + (1-\epsilon) l_1^*$,
%and $g_{t(M)}$ always exists.
%sam14
%joe14: unnecessary; the sentence starts ``We can find'' (now rewritten
%to ``There exists''
%for this choice of $\epsilon$, 
%For this choice of $\epsilon$, let $t(f) = \epsilon f + (1-\epsilon) l_1^*$,
$$1 \geq \sup_{f \in M} \epsilon u(f(s)) + (1-\epsilon) \geq -1.$$  
%joe14
%Therefore, we can let $t(f) = \epsilon f + (1-\epsilon) l_1^*$, and
%clearly there exists an act $g_{t(M)}$ such that 
Let $t(f) = \epsilon f + (1-\epsilon) l_1^*$.
Clearly there exists an act $g_{t(M)}$ such that 
$ u(g_{t(M)}(s) ) = -\sup_{f\in t(M)}u(f(s))$ for all $s \in S$.

Case 4: 
%joe9
%The utilities are unbounded above, and bounded below. 
The range of $U$ is bounded below and unbounded above.
By the upper-boundedness axiom, every menu has an upper bound on its
utility range.  
Therefore, for every menu $M$, $\epsilon >0$, and all acts $f$ and $g$
%joe9
%in $M$, by the independence axiom,  
in $M$, by Independence,
$$ f\succeq_M g \Leftrightarrow \epsilon f + (1-\epsilon) l_{-1}^*
\succeq_{\epsilon M + (1-\epsilon) l_{-1}^*} \epsilon g + (1-\epsilon)
l_{-1}^*. $$ 
%joe9
%We can always choose $\epsilon$ sufficiently small such that for every
%joe14
%We can find $\epsilon >0$ sufficiently small that for all
There exists $\epsilon >0$ such that for all
$s\in S$, 
%joe14: displayed
%$\sup_{f \in M} \epsilon u(f(s)) + (1-\epsilon) u(l_{-1}^*(s) ) \leq 1$. 
$$\sup_{f \in M} \epsilon u(f(s)) + (1-\epsilon) u(l_{-1}^*(s) ) \leq 1.$$ 
%joe9
%Therefore, we can let $t(f) =  \epsilon f + (1-\epsilon) l_{-1}^*$ for a
%sufficiently small $\epsilon$, such that $t(M)$ has utilities within
%$[-1,1]$, and $g_{t(M)}$ always exists. 
%joe14
%For this choice of $\epsilon$, let $t(f) =  \epsilon f + (1-\epsilon)
Let $t(f) =  \epsilon f + (1-\epsilon)
l_{-1}^*$.  Again, it is easy to see that $g_{t(M)}$ exists.
\end{proof}

%joe9*; this is important
In light of Lemma~\ref{lem:cases}, we henceforth assume that the utility
function $u$ derived from $U$ is such that its range is either
$(-\infty,\infty)$, $[-,1,1]$, $(-\infty,1]$, or $[-1,\infty)$.  In
any case, its \emph{range always includes }$[-1,1]$.

%joe9
%Next we will establish notation for acts and utility acts. 
Before proving the key lemma, we establish some useful notation for acts
%joe10*; have we defined ``utility act''?
and utility acts.  
Given a utility act $b$, let $f_b$, the act corresponding to $b$, be the
%joe9: If we make the change I suggested above, such an act will always
%exist.) 
%act $f$ such that $f(s) = l_{b(s)}$, if such an act exists
act such that $f_b(s) = l_{b(s)}$, if such an act exists.
%joe9*: this is certainly not true in general
%For future steps in our proof, it is convenient to note that $f_b$
%exists	for all $(-1)^* \leq b \leq 0^*$. 
Conversely, let $b_f$, the utility act corresponding to the act $f$, be
%joe9
%defined by $b_f(s) = U(f(s))$.   
defined by taking $b_f(s) = u(f(s))$.   
Note that monotonicity implies that if $f_b = g_b$, then $f \sim_M g$
%joe9
%for any $M$. 
for all menus $M$. 
That is, only utility acts matter.
%joe9
%If $c$ is a real, we will let $c^*$ be the constant utility act such
If $c$ is a real, we take $c^*$ to be the constant utility act such
that $c^*(s) = c$ for all $s \in S$. 

\begin{lemma}\label{lem:stoye} Let $M^*$ be the menu
consisting of all acts $f$ such that $(-1)^* \leq b_f \le 0^*$.  
%joe9*: we've never talked about (U,\cP^*) representing an ordering;
%rather.  If you're going to use this notion, you have to define it.
%Then $(U,\cP^+)$ represent $\succeq_{M^*}$ (i.e., $\succeq_{M^*} =
%\succeq_{M^*,\cP^+}^{S,X,U})$ iff $(U,\cP^+)$ represent $\succeq_M$ for
Then $(U,\cP^+)$ represents $\succeq_{M^*}$ (i.e., $\succeq_{M^*} =
\succeq_{M^*,\cP^+}^{S,X,U})$ iff $(U,\cP^+)$ represents $\succeq_M$ for
all menus $M$. 
\end{lemma}

\begin{proof}
%joe9
%We use the same general arguments as in \cite{Stoye2011}.
Our arguments are similar in spirit to those of Stoye \citeyear{Stoye2011}.

%joe9
By Lemma~\ref{lem:cases}, there exists $t$ such that $t(f) = \epsilon f
+ (1-\epsilon)h$ for a constant function $h$ such that 
\begin{align*}
%f\succeq_M g \Leftrightarrow t(f) \succeq_{t(M)} t(g); 
f\succeq_M g \mbox{ iff }  t(f) \succeq_{t(M)} t(g); 
\end{align*}
%joe9
%and such that $g_{t(M)}$ always exists.
moreover, for this choice of $t$, the act $g_{t(M)}$ defined in
Lemma~\ref{lem:cases} exists.
%joe9: this seems to be more than convenience.
%For convenience, in the following discussion we will refer to $t(M)$ as $M$.

%joe9*: Reordered argument.
By Independence,
$$t(f) \succeq_{t(M)} t(g) \mbox{ iff } 
\frac{1}{2}t(f)+\frac{1}{2}g_{t(M)} \succeq_{\frac{1}{2}t(M)
+\frac{1}{2}g_{t(M)}} \frac{1}{2}t(g)+\frac{1}{2}g_{t(M)}. 
$$
%\end{align*}

Let $M^*$ be the menu that contains all acts with utilities in $[-1,0]$.
By {INA}, we know that for all acts $f$ and $g$, and menus $M$
%joe9
for which $g_M$ is defined, we have
\begin{align*}
f\succeq_M g
\mbox{ iff }
\frac{1}{2}f+\frac{1}{2}g_M \succeq_{M^*} \frac{1}{2}g+\frac{1}{2}g_M.
\end{align*}
%joe9
%$M^*$ is the menu that contains all acts with nonpositive utilities. 
%This is because all such acts are never strictly optimal with respect to
This is because acts of the form $\frac{1}{2}f+\frac{1}{2}g_M$
are never strictly optimal with respect to
the menu $\frac{1}{2}M + \frac{1}{2}g_M$.  
At every state there must be some act in $\frac{1}{2}M+\frac{1}{2}g_M$
that has utility $0$ (namely, the mixture that involves the act
$\argmax_{f\in M}u(f(s))$. 

Thus, 
$$f\succeq_M g \mbox{ iff } 
\frac{1}{2}t(f)+\frac{1}{2}g_{t(M)} \succeq_{M^*} \frac{1}{2}t(g)+\frac{1}{2}g_{t(M)}. 
$$

%joe9
%The argument is completed by noting that independence and INA also hold
%for the MWER representation. 
%
%However, since the MWER representation also satisfies independence and
Since the MWER representation also satisfies Independence and
INA, we know that for all menus $M$, and acts $f$ and $g$ in $M$, 
%joe9*: up to now we written ^{S,X,u}, not ^{S,X,u}.  We could use
%{S,X,U} instead of talking about the u determined by U, but then this
%change has to be made uniformly.
$$ f \succeq_{M,\cP^+}^{S,X,U} g \Leftrightarrow t(f)
\succeq_{t(M),\cP^+}^{S,X,U} t(g)  \Leftrightarrow \frac{1}{2}t(f) +
\frac{1}{2}g_{t(M)} \succeq_{M^*,\cP^+}^{S,X,U} \frac{1}{2}t(g) +
\frac{1}{2}g_{t(M)}.$$ 

Therefore, to show that $\succeq_M$ has a MWER representation with
%joe9: 
respect to $(U,\cP^+)$, it suffices to show that $\succeq_{M^*}$ has a
MWER representation with respect to $(U,\cP^+)$.  
\end{proof}

%joe9
%Due to Lemma \ref{lem:stoye}, from this point on, we will drop the menu
In the sequel, we drop the menu
subscript when we refer to the family of preferences,  
and just write $\succeq$ (to denote $\succeq_{M^*}$);
%joe15
%by Lemma~\ref{lem:stoye}, this can be done without loss of generality.
by Lemma~\ref{lem:stoye}, it suffices to consider $\succeq_{M^*}$.

%joe9: we're not extending u
%\subsection{Extending $u$ to nonconstant acts}
%joe9
%\subsubsection{Defining $I$}
\subsection{Defining a functional on utility acts}

As we said, Stoye \citeyear{StoyeRegret} also started his proof of a
representation theorem for MER by reducing to a single preference order
$\succeq_{M^*}$.  He then noted that, 
the expected regret of an act $f$ with respect to a probability
$\Pr$ and menu $M^*$ is just the negative of the expected utility of
$f$.  Thus, the 
worst-case expected regret of $f$ with respect to a set $\cP$ of
probability measures is the negative of the worst-case expected utility
of $f$ with respect to $\cP$.  Thus, it sufficed for Stoye to show that 
$\succeq_{M^*}$ had an MMEU representation, which he did by showing that
$\succeq_{M^*}$ satisfied Gilboa and Schmeidler's \citeyear{GS1989}
axioms for MMEU, and then appealing to their representation theorem.

%joe15: 
%In our case, our  $\succeq_{M^*}$ does not satisfy the C-independence
%for MMEU.
This argument does not quite work for us, because now $\succeq_{M^*}$
does not satisfy the C-independence axiom.  (This is because
our preference order
$\succeq_{M^*}$ is based on \emph{weighted} regret, not regret.)
%joe15
%However, we can still use similar techniques as those used in
%\citeyear{GS1989} for MMEU  
However, we can get a representation theorem for weighted regret by
using some of the techniques used by Gilboa and Schmeidler to get a
representation theorem for MMEU, appropriately modified to deal with
lack of C-independence.
%joe15
%Namely, we define a functional $I$ on utility acts, where the ranking
%of ranking of any two utility acts depend on the value that $I$ assigns
%to them.  
Specifically, like Gilboa and Schmeidler, we define a functional $I$ on
utility acts such that the preference order on utility acts is
determined by their value according to $I$ (see Lemma~\ref{lem:Iprop}). 
%joe15: but we don't want a representation for I.  In any case, it
%doesn't help to say that ``I has certain properties''
%By showing that $I$ satisfies certain properties and using a separation
%theorem for Banach spaces, one can then show that $I$ has the desired
%representation. 
%We will see later that our case requires different arguments than those
%in \citeyear{GS1989}. 
Using $I$, we can then determine the weight of each probability in
$\Delta(S)$, and prove the desired representation theorem.

%joe15: redundant; we said it above, and we say it agein in the next
%paragraph. 
%In this section, 
%we define a functional $I$ mapping utility acts to the reals such that
%$f \succeq g$ iff $I(b_f) \ge I(b_g)$.  
Recall that $u$ represents $\succeq$ on constant acts, and that only
utility acts matter to $\succeq$.   
%joe9:
%
The space of all utility acts is the Banach space $\cB$ of real-valued
functions on $S$. 
Let $\cB^-$ be the set of nonpositive
functions in $\cB$, where the function $b$ is nonpositive if $b(s) \le 0$ for all $s
\in S$.    
%joe9

%Below, we will define a functional $I$ on utility acts such that
We now define a functional $I$ on utility acts in $\cB^-$ such that
%joe9*: actually, this isn't what you show
%$I(b)\geq I(b') \Leftrightarrow f_b\succeq f_{b'}$. 
for all $f,g$ with $b_f, b_g \in \cB^-$, we have
$I(b_f)\geq I(b_g)$ iff $f\succeq g$. 
%joe9
%For convenience, we first define the set 
Let
$$R_f = \{\alpha': l_{\alpha'}^* \succeq f\}.$$
%joe9: 
%In the case where there exists acts for all utilities, for utility acts
%$b$ we let  
If $0^* \ge b \ge (-1)^*$, then $f_b$ exists, and we define
$$I(b) = \inf(R_{f_b}).$$ 
%This determines $I(b)$ for all $b\in \cB^-$.
%If the range of $u$ is a subset of $[-1,\infty)$, then 
%joe9*: 
%In the case where utilities are bounded, 
%then for utility acts $b$ such that 
%$b\geq (-1)^*$ (i.e., $b(s) \ge -1$ for all $s \in S$), define
%$$I(b) = \inf(R_{f_b}).$$ 
For the remaining $b\in \cB^-$, we extend $I$ by homogeneity.
%joe9: ||b|| hasn't been defined
Let $||b|| = |\min_{s \in S}b(s)|$.  Note that if $b \in \cB^-$, then
$0^* \ge b/||b|| \ge (-1)^*$, so we define
$$I(b) =
||b|| I(b/||b||).$$
%joe9: added
%It is easy to check that the two approaches agree for $b \ge (-1)^*$;
%that is, if $b \ge (-1)^*$, then 

%joe9*: As near as I can tell, only the first property is used.  If so,
%there's no point cluttering up the paper with the others.  Moreover,
%the property had *nothing* to do with I; it's a consequence of Lemma
%1.I moved it back 
%The definition results in several useful properties.  For example, 
\commentout{
One useful consequence of this definition is that,
for $\alpha $ and $\beta$ in the range of $u$, we have 
\begin{align}
 \alpha > \beta \text{ if and only if } l^*_\alpha \succ l^*_\beta.
%joe9
% \text{ (and the two acts exist)}. 
\label{property1} 
\end{align}
(\ref{property1}) follows from monotonicity and the definition of $l^*_c.$

\begin{align}
 \text{For any act } f, \text{ if } \alpha \in R_f \text{ and } \beta
%joe9:
% \geq \alpha \text{ (and $l^*_\beta$ exists) then } \beta\in
\geq \alpha \text{ and $\beta$ is in the range of $u$, then } \beta\in
 \geq \alpha \text{ then } \beta\in
 R_f. \label{property2} 
 \end{align}
(\ref{property2}) follows from transitivity and the above property.

\begin{align}
\text{For } q\in (0,1), ql^*_0 + (1-q)l^*_c \sim l^*_{(1-q)c}.
\label{property3} 
\end{align}
This is because the constant act on the left hand side has the same
utility profile as the constant act on the right hand side. 
Similarly, for $q\in (0,1)$, $ql^*_{-1} + (1-q)l^*_c \sim l^*_{(1-q)c-q}$. 
}
%joe9: \end{commentout}

%joe9*: we need to make this a lemma.  I'm pretty sure that what you
%said is not what you meant.  It's not enough that f has nonpositive
%utility; we must have b_f \in \cB^-
%Now we show that, for any act $f$ with nonpositive utility, $f \sim
%l_{I(b_f)}^*$.  
%Suppose by way of contradiction, that 
\begin{lemma}\label{lem:key}
If $b_f \in \cB^-$, then $f \sim l_{I(b_f)}^*$.  
\end{lemma}
\begin{proof}
Suppose that $b_f \in \cB^-$ and, by way of contradiction, that 
$l_{I(b_f)}^* \prec f$.
If $f\sim l_0^*$, then it must be the case that $I(b_f)=0$, since
%joe9: you need to slow down a little here
%$I(b_f)\leq 0$ by definition of $\inf$ and $I(b_f)\geq \epsilon$ for
%all $\epsilon < 0$ by property~(\ref{property1}) and transitivity. 
$I(b_f)\leq 0$ by definition of $\inf$, and $f \sim l_0^* \succ
l_{\epsilon}^*$ for all $\epsilon < 0$ by Lemma~\ref{property1},
so $I(b_f) > \epsilon$ for all $\epsilon < 0$.
Therefore, $f \sim l_{I(b_f)}^*$. 
%joe9
%Otherwise, $l_0^* \succ f \succ l_{I(b_f)}^*$. 
Otherwise, since $b_f \in \cB^-$, by monotonicity, we must have $l_0^*
\succ f$, and thus 
$l_0^* \succ f \succ l_{I(b_f)}^*$. 
By mixture continuity, there is some $q\in (0,1)$ such that 
$ q\cdot l_0^* + (1-q) \cdot l_{I(b_f)}^* \sim l_{(1-q)I(b_f)} \prec f
%joe9
%$, contradicting $I(b)$ being the greatest lower bound of $R_f.$ 
$, contradicting the fact that $I(b)$ is the greatest lower bound of $R_f.$ 

If, on the other hand, $l^*_{I(b_f)} \succ f$, then $l^*_{I(b_f)} \succ
f \succeq l^*_{\underline{c}}$ for some $\underline{c}\in \R$. 
If $f \sim l^*_{\underline{c}}$ then it must be the case that $I(b_f)=\underline{c}$. 
$I(b_f) \leq \underline{c}$ since $l^*_{\underline{c}}\succeq
l^*_{\underline{c}}$, and $I(b_f) \geq \underline{c}$ since for all $c'
< \underline{c}$, $l^*_{c'} \prec f \sim l^*_{\underline{c}}$. 

Otherwise, $l^*_{I(b_f)} \succ f \succ l^*_{\underline{c}}$, and by mixture continuity, there is some $q\in (0,1)$ such that 
$q\cdot l^*_{I(b_f)} + (1-q) l^*_{\underline{c}} \succ f$.
Since $qI(b_f) + (1-q)\underline{c} < I(b_f)$, this contradicts 
%joe9
%$I(b_f)$ being a lower bound of $R_{f}$. 
the fact that $I(b_f)$ is a lower bound of $R_{f}$. 
Therefore, it must be the case that $l^*_{I(b_f)} \sim f$.
\end{proof}

%you can say that f
%\succeq g iff I(f_u) \ge I(g_u) (so that is the precise sense in which I
%represents \succeq), and for constant acts l, we can take I(l) = u(l)
%(for the u already defined on constant acts).  This makes precise the
%first paragraph of A.2.3.
%joe9: yet again, this should be a lemma
%Note that 
We can now show that $I$ has the required property.
\begin{lemma}\label{lem:Iprop}
For all acts $f,g$ 
%joe9*: you need this
such that $b_f, b_g \in \cB^-$,
%joe9*: strengthened
%if $f \succeq g$, then $I( b_f ) \ge I( b_g )$.
 $f \succeq g$ iff $I( b_f ) \ge I( b_g )$.
\end{lemma}
\begin{proof}
%This follows from the fact that 
Suppose that $b_f, b_g \in \cB^-$. By Lemma~\ref{lem:key}, 
$l^*_{I(b_f)} \sim f$ and $g \sim  l^*_{I(b_g)}$.  Thus, 
$f \succeq g$ iff  $l^*_{I(b_f)} \succeq l^*_{I(b_g)}$, and by
Lemma~\ref{property1}, 
$l^*_{I(b_f)} \succeq l^*_{I(b_g)}$ iff 
$I(b_f)\geq I(b_g)$. 
\end{proof}

%(The standard separation
%theorems apply to Banach spaces.  I suspect that they can be extended to
%closed subspaces of Banach spaces; we should check.  Just to be safe, I
%made some minor technical modifications so that things work for $\cB$.
%Note that $\cB^-$ itself is not a Banach space.  To be a Banach
%space over some field $F$, you have to be able to multiply any vector in
%the space by an element of $F$ to get another element in the space.
%Here $F$ is the real numbers.  You can't multiply a   negative-valued
%function by a negative number to get another negative-valued function.)

In order to invoke a standard separation result for Banach spaces, we
extend the definition of $I$ to the Banach space $\cB$.  
We extend $I$ to
$\cB$ by taking $I(b) = I(b^-)$ for $b \in \cB-\cB^-$, where for all $b \in \cB$, $b^-$ is defined as
$$b^-(s) = \begin{cases}
b(s) \text{, if } b(s) \le 0, \\
0, \text{ if } b(s) > 0.\end{cases}$$
Clearly $b^- \in \cB^-$ and $b=b^-$ if $b \in \cB^-$.
%joe9*: Do we want to prove that f \succeq g still implies that I(b_f)
%\get I(b_g)?  We shouldn't prove it unless we need it.

%joe9
%\subsubsection{Properties of $I$ on $\cB$}
%joe9
%We'll show that the assumed properties of $\succeq$ guarantee that on
%$, $I$ maps constant functions in $\cB^-$) to their values, is
%(positive) homogeneous of degree 1, and is monotonic, continuous, and
%superadditive. 
We show that the axioms guarantee that $I$ has a number of standard properties.
%sam15 added 
Since we have artificially extended $I$ to $\cB$, our arguments require
more cases than those in \cite{GS1989}.  
%joe15
(We remark that such an ``artificial'' extension seem unavoidable in our
setting.)  
%joe15: it's not the functional that satisfies C-independence, but the
%preference order
Moreover, we must work harder to get the result that we want.
We need different arguments from that for MMEU \citeyear{GS1989}, since
the preference order induced by MMEU satisfies C-independence, while 
our preference order does not.

%joe9
\begin{lemma}\label{lem:Iproperties}
\begin{enumerate}
%joe9*: you should reserve \forall for formulas
%\item $I(c^*)=c, \forall $c^* \in \cB^-$
\item[(a)] If $c \le 0$, then $I(c^*)=c$.
\item[(b)] $I$ satisfies positive homogeneity: if $b \in \cB$ and $c > 0$,
then $I(cb) = cI(b)$. 
\item[(c)] $I$ is monotonic: if $b, b' \in \cB$ and $b \ge b'$, then $I(b)
\ge I(b')$.
\item[(d)] $I$ is continuous: if $b, b_1, b_2, \ldots \in \cB$, and $b_n
\rightarrow b$, then $I(b_n) \rightarrow I(b)$.
\item[(e)] $I$ is superadditive: if $b, b' \in \cB$, then $I(b+b') \ge
I(b) + I(b')$.
\end{enumerate}
\end{lemma}
\begin{proof} For part (a), 
%joe9: rewrote argument
%joe9: rewrote proof.
If $c$ is in the range of $u$,  then it is immediate from the defintion
of $I$ and Lemma~\ref{property1} that $I(c^*) = c$.  
If $c$ is not in the range of $u$, then
since $[-1,0]$ is a subset of the range of $u$, we must have $c < -1$,
and by definition of $I$, we have $I(c^*) = |c| I(c^*/|c|) = c$.

%joe9: the structure of the proof isn't clear.  In particular, it isn't
%clear the way you've written it why we need to consider a special case 
%\item $I$ satisfies positive homogeneity on $\cB$: \\
%joe9: rewrote proof
%We'll first consider the case where act utilities are unbounded. 
%Consider some $b \in \cB^-$ and $0 < c \leq 1$.
%	We have to show that $I(cb)=c I(b)$.
%We've already established that $f_b \sim l^*_{I(b)}$. 
%	By independence of $\succeq$, we know that 
%	\begin{align*}
%	f_b \sim l^*_{I(b)} \Leftrightarrow c f_b + (1-c) l_0^* \sim c
%l^*_{I(b)} + (1-c) l_0^*. 
%	\end{align*} 
For part (b), first suppose that $||b|| \le 1$ and $b \in \cB^-$ (i.e.,
$0^* \ge b \ge (-1)^*$).   Then there exists an act $f$ such that $b_f = b$.  
By Lemma~\ref{lem:key}, $f \sim l^*_{I(b)}$. 
We now need to consider the case that $c \le 1$ and $c > 1$ separately.
If $c \le 1$, by Independence,
$c f_b + (1-c) l_0^* \sim c l^*_{I(b)} + (1-c) l_0^*$. 
%joe9: We haven't defined utility profile
%Note that the utility profile of $c f_b + (1-c) l_0^*$ is equal to
%$cb$. Moreover, the utility profile of $c l^*_{I(b)} + (1-c) l_0^*$ is
%$(cI(b))^*$. 
By Lemma~\ref{lem:Iprop}, $I(b_{c f_b + (1-c) l_0^*}) =  I(b_{c l^*_{I(b)}
+ (1-c) l_0^*})$.  It is easy to check that  $b_{c f_b + (1-c) l_0^*} =
cb$, and $b_{c l^*_{I(b)}} + (1-c) l_0^* = cI(b)^*$.  Thus, $I(cb) =
I(cI(b)^*)$.  By part (a), $I(cI(b)^*) = cI(b)$.  Thus, $I(cb) = cI(b)$, as
desired.  

If $c > 1$, there are two subcases.  If $||cb|| \le 1$, since $1/c < 1$,
by what we have just shown $I(b) = I(\frac{1}{c}(cb)) =
\frac{1}{c}I(cb)$.  Crossmultiplying, we have that $I(cb) = cI(b)$, as
desired.  And if $||cb||>1$, by definition, $I(cb) = ||cb||
I(bc/||cb||) = c||b||I(b/||b||)$ (since $bc/||cb|| = b/||b||$).
Since $||b|| \le 1$, by what we have shows $I(b) = I(||b|| (b/||b||) =
||b||I(b/||b||)$, so $I(b/||b||) = \frac{1}{||b||} I(b)$.  Again,
it follows that $I(cb) = cI(b)$.

Now suppose that $||b|| > 1$.  Then $I(b) = ||b|| I(b/||b||)$.
Again, we have two subcases.
If $||cb|| > 1$, then $$I(cb) = ||cb|| I(cb/||cb||) = c||b|| I(b/||b||) 
= cI(b).$$
And if $||cb|| \le 1$, by what we have shown for the case $||b|| \le 1$,
$$I(b) = I(\frac{1}{c} (cb))
= \frac{1}{c}I(cb),$$ so again $I(cb) = cI(b)$.

%joe9: cut all this
\commentout{
	By (\ref{property1}), $I( (cI(b))^*) = I(cb)$. 
	However, $I( (cI(b))^* ) = cI(b)$ by property 1.
	Therefore, we have
	\begin{align*}
I(c b) = c I(b)
	\end{align*}
	for $0 < c \leq 1$. 
	The case $c > 1$ reduces to the case of $0 < c\leq 1$. 
	To see how, consider any $c > 1$ and $b \in \cB^-$.
	$I(b) = I(\frac{1}{c} c b ) = \frac{1}{c} I(cb)$, therefore $I(cb)=cI(b)$.
	
	Next we consider the case where act utilities are bounded. 
	There are three sub-cases where $cb \geq (-1)^*$. 
	For the case $0 < c \leq 1$ and $b \geq (-1)^*$, by the same application of the independence axiom, we can conclude that $I(cb)=cI(b)$.
	For the case where $c > 1$, $b \geq (-1)^*$, note that $I(b) = I(\frac{1}{c} cb) = \frac{1}{c} I(cb)$, so $I(cb)=cI(b)$. 
	For the remaining case where $0 < c \leq 1$ and $b \not\geq (-1)^*$, it is also the case that $I(b) = I(\frac{1}{c} c b) = \frac{1}{c} I(cb)$, so again $I(cb)=cI(b)$. 
	
	Since this previous case covers all the cases where $cb \geq (-1)^*$, we'll now consider the case where $cb \not\geq (-1)^*$. 
	Recall from the definition of $I$ that $I(c b) = ||c b|| I( c b/ ||c b||)$, which is equal to $||c b|| I( c b/ ||c b||) = c I(b)$ by the above case.
	Hence $I(c b) = c I(b)$.
	
	Now for $b\in \cB\backslash \cB^-$, consider any $b\in \cB-\cB^-$. Note that for $c>0$, $I(cb)=I(cb^-)$, where 
$$cb^-(s) = \begin{cases}
cb(s) \text{, if } b(s) \le 0, \\
0, \text{ if } b(s) > 0\end{cases}.$$
Since $cb^-\in \cB^-$, $I(cb^-)=cI(b^-)$. Therefore $I(cb)=I(cb^-)=cI(b^-)=cI(b)$, as desired.
%\item $I(c^*)=c$ not satisfied.
	
\item $I$ is monotonic: \\
}
%joe9: \end{commentout}
For part (c), first note that if $b, b' \in \cB^-$.  If $||b|| \le 1$ and
$|b'|| \le -1$, then the acts $f_b$ and $f_{b'}$ exist.  Moreover, since
$b \ge b'$, we must have $(f_b(s))^* \succeq (f_{b'})^*(s)$ for all
states $s \in S$. Thus, by Monotocity, $f_b \succeq f_{b'}$.   
If either $||b|| > 1$ or $||b'|| > 1$, let $n = \max(||b||,||b'||)$.
Then $||b/n|| \le 1$ and $||b'/n|| \le 1$.  Thus, $I(b/n) \ge I(b'/n)$,
by what we have just shown.  By part (b), $I(b) \ge I(b')$.
Finally, if either $b \in \cB - \cB^-$ or $b' \in \cB - \cB^-$, 
note that if $b \ge b'$, then 
$b^- \ge (b')^-$.  By definition, $I(b) = I(b^-)$ and $I(b') =
I(b')^-$; moreover, $b^-, (b')^- \in \cB^-$.  Thus, by the argument
above, $I(b) \ge I(b^-)$.

%joe10*: argument rewritten; please check
%$I$ is continuous :\\
For part (d), 
%note that if 
%$b_n \rightarrow b$, then it must be the case that
%$b_n^- \rightarrow b^-$, since $| b^-_n(s) - b^-(s) | \leq | b_n(s) -
%b(s) |$ for all $s$.  Since $I(b^-) = I(b)$, it suffices to consider the
%case that $b \in \cB^-$.
note that if $b_n \rightarrow b$, then for all $k$, there exists
$n_k$ such that $b_n - (1/k)^* \le b_n \le b_n + (1/k)^*$ 
%sam14 
for all $n\geq n_k$.
Moreover, by
the monotonicity of $I$ (part (c)), we have that $I(b - (1/k)^*) \le
I(b_n) \le I(b + (1/k)^*)$.   Thus, it suffices to show that 
$I(b - (1/k)^*) \rightarrow I(b)$ and that $I(b + (1/k)^*) \rightarrow
I(b)$. 

To show that $I(b - (1/k)^*) \rightarrow I(b)$, we must show that  for
all $\epsilon > 0$, there exists $k$ such that $I(b- (1/k)^*) \ge I(b) -
\epsilon$.  By positive homogeneity (part (b)), we can assume without
loss of generality that $||b - (1/2)^*|| \le 1$ and that $||b||
\le 1$.  Fix $\epsilon > 0$.  If $I(b - (1/2)^*) \ge I(b) - \epsilon$, then
we are done.  If not, then $I(b) > I(b)- \epsilon  > I(b - (1/2)^*) $.
Since $||b|| \le 1$ and $||b-(1/2)^*|| \le 1$, $f_b$ and $f_{b-(1/2)^*}$
exist.  Moreover, by Lemma~\ref{lem:Iprop},
$f_b \succ f_{(I(b) - \epsilon)^*} \succ f_{b-(1/2)^*} $.
By mixture continuity, for some $p\in (0,1)$, we have
$ pf_b  + (1-p) f_{(b-(1/2)^*} \succ f_{(I(b) - \epsilon)^*} $.
It is easy to check that $b_{p f_b  + (1-p) f_{b-(1/2)^*}} = b -
(1-p)(1/2)^*$.  Thus, by Lemma~\ref{lem:Iprop}, 
$f_{b-(1-p)(1/2)^*} \succeq f_{(I(b)-\epsilon)^*}$, and 
$I(b - (1-p)1/2)^*) > I(b) - \epsilon$.  Choose $k$ such that $1/k <
(1-p)(1/2)$.  
Then $I(b-(1/k)^*) \ge I(b - (1-p)1/2)^*) > I(b) - \epsilon$,
as desired.

%joe14
%Then argument that $I(b + (1/k)^*) \rightarrow I(b)$ is similar and left
The argument that $I(b + (1/k)^*) \rightarrow I(b)$ is similar and left
to the reader.  

%We will first consider the case where utilities are unbounded. 
%First we consider the requirement that $I(b_n) \leq I(b)+\epsilon$. 
%By monotonicity of $I$, if $I(b+\epsilon^*) \leq I(b) + \epsilon$, then
%we're done.  
%Otherwise, it must be the case that $f_{b+\epsilon^*} \succ
%f_{(I(b)+\epsilon)^*} \succ f_b$. 

%\item $I$ is superadditive:\\
For part (e), 
first suppose that $b, b' \in \cB^-$.
If $||b||, ||b^-|| \le 1$, and $I(b), I(b') \ne 0$, 
consider $\frac{b}{I(b)}$ and $\frac{b'}{I(b')}$. 
Since $I( \frac{b}{I(b)} ) = I(\frac{b'}{I(b')}) = 1$, 
%for any acts $f,f'$ such that $u\circ f = \frac{b}{I(b)} $ and $u\circ
%f' =  \frac{b'}{I(b')}$, we have $f\sim f'$. 
it follows from Lemma~\ref{lem:key} that 
$f_{\frac{b}{I(b)}} \sim f_{\frac{b'}{I(b')}}$.
By Ambiguity Aversion, for all $p\in (0,1]$, $p f_{\frac{b}{I(b)}} +
(1-p) f_{\frac{b'}{I(b')}} \succeq f_{\frac{b}{I(b)}}$. 
Thus, $I( \frac{I(b)}{I(b) + I(b') } \frac{b}{I(b)} + \frac{I(b')}{I(b)
+ I(b') }\frac{b'}{I(b')} ) \geq I( \frac{b}{I(b)} ) = I(
\frac{b'}{I(b')}) = 1 $.  
Hence, $ I( b + b' ) \geq I( b ) + I( b' )$.  

If $b, b^- \in \cB^-$ and either $||b|| > 1$ or $||b'|| > 1$, and both
$I(b) \ne 0$ and $I(b') \ne 0$, then the result easily follows by positive
homogeneity (property (b)).

If $b, b^- \in \cB-$ and either $I(b)=0$ or $I(b') = 0$, 
let $b_n = b- \frac{1}{n}^*$ and $b'_n = b' - \frac{1}{n}^*$.  Clearly
$||b_n|| > 0$, $||b'_n|| > 0$, $b_n \rightarrow b$, and $b_n'
\rightarrow b_n'$.  By our argument above, $I(b_n + b_n') \ge I(b_n) +
I(b_n')$ for all $n \ge 1$.  The result now follows from continuity.  
	
%sam14
%Finally, if either $b \in \cB^ - \cB$ or $b' \in \cB - \cB^-$, observe
Finally, if either $b \in \cB - \cB^-$ or $b' \in \cB - \cB^-$, observe
that 
	$$ (b+b')^-(s) \begin{cases} 
	= b^-(s) + b'^-(s) ,\text{ if } b(s) \leq 0,b'(s) \leq 0 \\
	= b^-(s) + b'^-(s) ,\text{ if } b(s) \geq 0 ,b'(s) \geq 0 \\
	\geq b^-(s) + b'^-(s) ,\text{ if } b(s)>0 ,b'(s) \leq 0 \\
%joe10
%	\geq b^-(s) + b'^-(s) ,\text{ if } b(s)\leq 0,b'(s) > 0
%	\end{cases}.
	\geq b^-(s) + b'^-(s) ,\text{ if } b(s)\leq 0,b'(s) > 0.
	\end{cases}
	$$
	Therefore, $(b+b')^- \geq b^- + b'^-$. 
	Thus, $I( b + b' ) = I( (b+b')^- ) \geq I(b^- + b'^-) $ by
%joe10
%	monotonicity of $I$ on $\cB^-$, and  
the monotonicity of $I$, and  
$I(b^- + b'^-) \geq I(b^-) + I(b'^-)$ by superadditivity of $I$ on
$\cB^-$.  	Therefore, $I(b+b') \geq I(b) + I(b')$.
\end{proof}

%joe10
%\subsection{$I$ is representable by maximizing $NWREG$}

\subsection{Defining the weights}
%sam15 add comparison to GS1989 proof. 
%joe15: this was misplaced.  GS didn't have weights at all.
%We need different arguments from that for MMEU \citeyear{GS1989}, since
%the functional $I$ in \citeyear{GS1989} satisfies C-independence, while
%our $I$ does not. 

%joe10:
%joe15: 
%In this section, we use $I$ and $U$ to define a weight 
In this section, we use $I$ to define a weight 
$\alpha_{\Pr}$ for each probability $\Pr \in \Delta(S)$.  The heart of the
proof involves showing that the resulting set $\cP^+$ so determined gives
us the desired representation.

Given a set $\cP^+$ of
weighted probability measures, for $b \in \cB^-$, define
$$\NWREG(b) =  \inf_{\Pr \in \cP} \alpha_{\Pr} (\sum_{s \in S} b(s) \Pr(s)).$$
Note that $\NWREG$ is the negative of the weighted regret when the menu
%joe10
%has utility frontier $0^*$.   
%sam14 if the menu consists of strictly negative b's then this might not
%be the negative of the weighted regret 
%is contained in $\cB^-$.
is $\cB^-$. 
%We want to show that $\succeq$ is representable by the maximization of $\NWREG$. 
%(I'm ignoring utility here. What I really should be considering is  u
%\circ f.) 
%joe10
%
%For future reference, denote the unweighted negative regret as
Define
$$\NREG(b) =  \inf_{\Pr \in \cP} \sum_{s \in S} b(s) \Pr(s).$$
%joe10
%And also define 
and
$$\NREG_{\Pr}(b) =  \sum_{s \in S} b(s) \Pr(s) = E_{\Pr}b.$$

% Do not need since already defined the properties for I above
%$\NWREG$ satisfies the following five properties for
%$f, g,\gamma^* \in \cB^-$:
%\begin{enumerate}
%\item  $\NWREG(\gamma^*) = \gamma$.
%\item $\NWREG(\gamma f) = \gamma \NWREG(f)$ if $\gamma \ge 0$ (positive
%homogeneity) 
%\item $\NWREG(f+g) \ge \NWREG(f) + \NWREG(g)$ (superadditivity).
%%sam1: I think it's the other direction of inequality
%%\item If $f \ge g$, then $\NWREG(f) \le \NWREG(g)$ (monotonicity).
%\item If $f \ge g$, then $\NWREG(f) \geq \NWREG(g)$ (monotonicity).
%\item $\NWREG$ is continuous.
%\end{enumerate}

%It seems that (1) + (2) + (5) suffice to replace positive affine
%homogeneity: $\NWREG(\beta f + \gamma^*) = \beta \NWREG(f) + \gamma$, which
%is true in the unweighted case, but false in the weighted case.

%joe10
%Given a function $I$ on $\cB^-$ that satisfies properties (1)--(5) above,
%we want to show that there exists a set of weighted probability measures
%on $\cB^-$ such that $I(b) = \NWREG(b)$.  

% As is well known, we
%can use $I$ to define a utility on lotteries.

%We will define the weights as follows. For each probability $\cP$, let  
For each probability $\Pr \in \Delta(S)$, define 
%joe1*: fixed,
%$$\alpha_\Pr = \inf{\alpha: alpha \alpha \NREG_\Pr(f) \ge I(f) 
%\begin{align} \label{eqn:alpha} \alpha_{\Pr} = \sup\{\alpha: \alpha \NREG_{\Pr}(b) \ge I(b)
% sam7: change to max  %sam18 change back...
\begin{align} \label{eqn:alpha} \alpha_{\Pr} = \sup\{\alpha \in \R: \alpha \NREG_{\Pr}(b) \ge I(b) 
\mbox{ for all $b \in \cB^-$}\}.\end{align}
%joe1*: added next four lines
%joe10
%Note that $\alpha_{\Pr} \ge 0$ for all distribution $\Pr$, since $0 \ge
Note that $\alpha_{\Pr} \ge 0$ for all distributions $\Pr \in \Delta(S)$,
since $0 \ge 
I(b)$ for $b \in \cB^-$ (by monotonicity); and $\alpha_{\Pr} \le
1$, since $\NREG_{\Pr}((-1)^*) = I((-1)^*) = -1$ for all distributions
$\Pr$.  Thus, $\alpha_{\Pr} \in [0,1]$.
%joe14*: Sam, the following is all true, but I don't see where we use
%it.  If we don't use it, there's no point in saying it; it just
%clutters the presentation.  If we do use it, reinstate it, but make
%clear where it's used.
%Moreover, the set $\{\alpha \in [0,1]: \alpha \NREG_{\Pr}(b) \ge I(b)
%\mbox{ for all $b \in \cB^-$}\}$ is closed, since if $\alpha_n
%\NREG_{\Pr}(b) \ge I(b) $ for $n\in\N$, and $\alpha_n\rightarrow
%\alpha$, then $\alpha \NREG_{\Pr}(b) \ge I(b) $.  
%Therefore, the set $\{\alpha \in [0,1]: \alpha \NREG_{\Pr}(b) \ge I(b)
%\mbox{ for all $b \in \cB^-$}\}$ is compact for every $\Pr$. 
%joe14*: added the following, which is important
Moreover, it is immediate from the definition of $\alpha_{\Pr}$ that
$\alpha_{\Pr} \NREG_{\Pr}(b) \ge   I(b)$ for all $b \in \cB^-$.  
The next lemma shows that there exists a probability $\Pr$ where we have
equality.

% Sam: Removed the claim 

%To prove that a weighted minimax regret representation exists, we must
%show two things: 
%joe14: I don't think we've defined the notion of representing I.  (If
%you want to reinstate the sentence, we need to say ``shows'')
%The next lemma show that the weighted negative regret represents $I$:
\begin{lemma}\label{lem:weights}
\begin{itemize}
\item[(a)] For some distribution  $\Pr$,  we have $\alpha_{\Pr} = 1$.  
\item[(b)] For all $b \in \cB^-$, there exists $\Pr$ such that $\alpha_{\Pr}
\NREG_{\Pr}(b) = I(b)$.
\end{itemize}
\end{lemma}
\begin{proof}
%joe10
%Both proofs use a standard separation result:  If $U$ is an open convex
The proofs of both part (a) and (b) use a standard separation result:
If $U$ is an open convex subset of $\cB$, and $b \notin U$, then there
is a linear functional $\lambda$ that separates $U$ from $b$, that is,
$\lambda(b') > \lambda(b)$ for all $b' \in U$.   We proceed as follows

%joe10
%\subsubsection{Some $\alpha_{\Pr} =1$}
%For the first claim, 
For part (a),
we must show that for some $\Pr$, for all $b\in
\cB^-$, $\NREG_{\Pr}(b) \ge I(b)$.    
Since $\NREG_{\Pr}(b) = E_{\Pr} b$, it suffices to show that $E_{\Pr}(b)
\ge I(b)$ for  all $b \in \cB^-$.

Let $U = \{b' \in \cB: I(b') > -1\}$.  $U$ is open (by continuity of $I$), and convex (by
positive homogeneity and superadditivity of $I$), and $(-1)^* \notin U$.  Thus,
there exists a linear functional $\lambda$ such that
$\lambda(b') > \lambda((-1)^*)$ for $b' \in U$.  
%joe10
%This $\lambda$ turns out to be the $E_{\Pr}$ we want, and will allow us
%to construct the ${\Pr}$ we need. 

%joe10
We want to show that $\lambda$ is a positive linear functional, that is,
that $\lambda(b) \ge 0$ if $b \ge 0^*$.
Since $0^* \in U$, and
$\lambda(0^*) = 0$, it follows that $\lambda((-1)^*) < 0$.  Since
$\lambda$ is linear, we can assume
without loss of generality that $\lambda((-1)^*) = -1$.  
%joe10: try to avoid using ``any''; use ``some'', ``all'', or cut it
%altogether 
%Thus, for any $b'\in \cB^-$, $I(b') > -1$ implies $\lambda(b') > -1$. 
Thus, for all $b'\in \cB^-$, $I(b') > -1$ implies $\lambda(b') > -1$. 
%joe10
%We use this fact to show that $\lambda$ is nonnegative for $b$ that
%takes on only nonnegative values.  
%Consider some $c > 0$ and $b' \geq 0^*$.  
%We know $I(cb') = I(0^*) = 0 > -1$. 
(The fact that $I(cb') = I(0^*)$ follows from the definition of $I$ on
elements in $\cB - \cB^-$.)  
Suppose that $c > 0$ and $b' \geq 0^*$.  
From the definition of $I$, it follows that $I(cb') = I(0^*) = 0 > -1$. 
So $c\lambda(b') = \lambda(cb') > -1$, so $\lambda(b') > -1/c$.
Since this is true for all $c > 0$, it must be the case that
$\lambda(b') \ge 0$.  
Thus, $\lambda$ is a positive functional.  

Define the probability distribution  $\Pr$ on $S$ by taking $\Pr(s) =
\lambda(1_{s})$.  To see that $\Pr$ is indeed a probability
distribution, note that since $1_{s} \ge 0$ and $\lambda$ is positive,
we must have $\lambda(1_{s}) \ge 0$.  Moreover,
$\sum_{s \in S} \Pr(s) = \lambda(1^*) = 1$.  In addition,  for all $b' \in
\cB$, we have 
$$\lambda(b') = 
\sum_{s \in S} \lambda(1_{s}) b'(s) = \sum_{s \in S} \Pr(s) b'(s) =
E_{\Pr}(b').$$

Next note that, for $b \in \cB^-$,
\begin{equation}\label{eq1}
\mbox{for all $c<0$, if $I(b) > c$, then $\lambda(b) > c$}.
\end{equation}  
For if $I(b) > c$, then $I(b/|c|) > -1$ by positive homogeneity, so
$\lambda(b/|c|) > -1$ and $\lambda(b) > 
c$. The result now follows. For if $b \in \cB^-$, then $I(b) \le I(0^*) =
0$ by monotonicity.  Thus, if $c < I(b)$, then $c < 0$, so, by (\ref{eq1}),
$\lambda(b) > c$. Since $\lambda(b) > c$ whenever $I(b) > c$, it follows that
$ E_{\Pr}(b) = \lambda(b) \ge  I(b)$, as desired.  

%joe10
%\subsubsection{For some $\Pr$,  $\alpha_{\Pr} \NREG_{\Pr}(b) = I(b)$}
%The proof of the second claim, that for every $b\in \cB^-$ there is some
%$\Pr$ that gives $b$ the correct negative regret to match $I(b)$, is
%similar in spirit. 
The proof of part (b) is similar to that of part (a).  We want to show
that, given $b \in \cB^-$, there exists $\Pr$ such that $\alpha_{\Pr}
\NREG_{\Pr}(b) = I(b)$.  
%joe10*: I don't see why this is true.  Rewrote
%If $I(b) = 0$, then it's easy: any distribution $\Pr$ will work.  
First supose that $||b|| \le 1$.  If $I(b) = 0$, then there must exist
%sam15 fixed typo
some $s$ such that $b(s) = 0$, for otherwise there exists $c < 0$
%some $s$ such that $I(b) = 0$, for otherwise there exists $c < 0$
such that $b \le c^*$, so $I(b) \le c$.  If $b(s) = 0$, let %sam15 fixed typo: changed \epsilon to c
$\Pr_s$ be such that $\Pr_s(s) = 1$.  Then $\NREG_{\Pr_s}(b) = 0$, so
(b) holds in this case.

If $||b|| \le 1$ and $I(b) < 0$, 
%It suffices to consider only $b \ge (-1)^*$. 
%This is because if the case for $b\||b||$, where $||b||=\max_{s\in
%%S}|b(s)|$, holds, then for some $\Pr$, 
%$$  \alpha_{\Pr} E_{\Pr}\frac{b}{||b||} = I\left( \frac{b}{||b||}\right),
%$$
%and the case for $b$ follows from homogeneity of $I$.
let $U = \{b': I(b') > I(b)\}$.  Again, $U$ is open and convex,
and $b \notin U$, so there exists a linear functional $\lambda$ such that 
$\lambda(b') > \lambda(b)$ for $b' \in U$.  
Since $0^* \in U$ and
$\lambda(0^*) = 0$, we must have $\lambda(b) < 0$.  
%joe2:
%Thus, we 
%can assume without loss of generality that $\lambda(f) = -1$.  
Since $(-1)^* \leq b$, $(-1)^*$ is not in $U$, and therefore we also have $\lambda((-1)^*) < 0$.  Thus, we
%joe10
%can assume without loss of generality that $\lambda((-1)^*) = -1$ and
can assume without loss of generality that $\lambda((-1)^*) = -1$, and
hence $\lambda((1)^*) = 1$.    
%joe1
%Thus, $I(g)/(-I(f)) > -1$ implies $\lambda(g) > -1$.  The same argument
%joe2: cut
%Thus, $I(g) > I(f)$ implies $\lambda(g) > -1$.  
The same argument as above shows that $\lambda$ is positive:
%joe10
%for any $c>0$ and $b'\geq 0^*$, $I(cb') = 0$ as before. 
%Since $I(b)<0$, $I(cb') > I(b)$ and hence $cb' \in U$ and $\lambda(cb')
for all $c>0$ and $b'\geq 0^*$, $I(cb') = 0$ as before. 
Since $I(b)<0$, it follows that $I(cb') > I(b)$, so $cb' \in U$
and $\lambda(cb') 
> \lambda(b) \geq \lambda((-1)^*) = -1$. 
%joe10
%Thus as before, for all $c>0$, $b' \geq 0^*$, $\lambda(b') >
%\frac{-1}{c}$ and thus $\lambda$ is a positive functional. 
Thus, as before, for all $c>0$, $b' \geq 0^*$, $\lambda(b') >
\frac{-1}{c}$, so $\lambda$ is a positive functional. 

Therefore, $\lambda$ determines a
probability distribution $\Pr$ such that, for all $b' \in \cB^-$, we have 
$\lambda(b') = E_{\Pr}(b')$.  This, of course, will turn out to be the
desired distribution.  To show this, we need to show that $\alpha_{\Pr}
%joe1*: simplified
%= I(f)/\NREG_{\Pr}(f) = I(f)/E_{\Pr}(f)$.  For then we clearly have 
%$\alpha_{\Pr} \NREG_{\Pr}(f) = I(f)$, as desired.
%joe2: back to what was there before
%-1$ by assumption, it then follows that
= I(b)/\NREG_{\Pr}(b)$.
%joe1*: rewrote this paragraph.  There were lots of bugs (although the
%idea was right!). 
%joe2*: undid change, correcting error you spotted; lots of rewriting in
%this paragraph
%Clearly $\alpha_{\Pr} \ge I(f)/\NREG_{\Pr}(f)$, since if $\alpha <
%I(f)/\NREG_{\Pr}(f)$, then $\alpha I(f) < \NREG_{\Pr}(f)$.  Thus, it
Clearly $\alpha_{\Pr} \le I(b)/\NREG_{\Pr}(b)$, since if $\alpha >
I(b)/\NREG_{\Pr}(b)$, then $\alpha \NREG_{\Pr}(b) < I(b)$ (since
$\NREG_{\Pr}(b) = \lambda(b) < 0$).  To show that
%Clearly $\alpha_{\Pr} \le -I(f)$, since if $\alpha >
%- I(f)$, then $\alpha \NREG_{\Pr}(f) = -\alpha < I(f)$.  To show that
$\alpha_{\Pr} \ge I(b)/\NREG_{\Pr}b$, we must 
show that $(I(b) /\NREG_{\Pr}(b))\NREG_{\Pr}(b') \ge I(b')$ for
all $b' \in \cB^-$.  Equivalently, we must show that 
$I(b) \lambda(b')/\lambda(b) \ge I(b')$ for all $b' \in \cB^-$. 

Essentially the same argument used to prove (\ref{eq1}) also shows
%joe2*: here's the big change
%\mbox{if $I(g)/(-I(f)) > c$ and $c < 0$, then $\lambda(g) > c$}.
$$\mbox{for all $c>0$, if $I(b') > cI(b)$, then $\lambda(b') > c \lambda(b)$}.$$
%\end{equation}  
In particular, if $I(b') > cI(b)$, then by positive homogeneity, $\frac{I(b')}{c} > I(b)$, so $\frac{b'}{c}\in U$, and $\lambda(\frac{b'}{c}) > \lambda(b)$ and hence $\lambda(b') > c\lambda(b)$.

Thus, if $I(b')/(-I(b)) > c$ and $c < 0$, then $I(b') > -c I(b)$, and hence $\lambda(b')/(-\lambda(b)) > c$.
It follows that $\lambda(b')/(-\lambda(b)) \ge I(b')/(-I(b))$ for all $b'
\in \cB^-$.
%joe1*: simplified
%Since $\lambda(f) = -1$, it follows that $\lambda(g) \ge
%(I(g)/I(f))\lambda(f)$. This is exactly what we want.
Thus, $I(b) \lambda(b')/\lambda(b) \ge I(b')$ for all $b' \in \cB^-$, as
required. 

Finally, if $||b|| > 1$, let $b' = b/||b||$.  By the argument above,
there exists a probability measure $\Pr$ such that
$\alpha_{\Pr}\NREG_{\Pr}(b/||b||) = I(b/||b||)$.  Since $\NREG_{\Pr}(b/||b||)
= \NREG_{\Pr}(b)/||b||$, and $I(b/||b||) = I(b)/||b||$, we must have that 
$\alpha_{\Pr}\NREG_{\Pr}(b) = I(b)$.
\end{proof}

%joe14: added
We can now complete the proof of Theorem~\ref{thm:completeness}.
%sam14
By Lemma~\ref{lem:weights} and the definition of $\alpha_{\Pr}$, for all
$b\in \cB^-$,  
\begin{align}
\label{equ:defI} I(b) = 
%joe14: what's \cP here?  Isn't it all distributions?  If so, you have
%to say that.  I also added the next line.  Finally, why not write -b
%instead of 0-b
%\inf_{\Pr \in \cP }\left( \alpha_{\Pr}\sum_{s\in S}b(s)\Pr(s) \right) =
%\sup_{\Pr \in \cP }\left( \alpha_{\Pr}\sum_{s\in S}(0-b(s))\Pr(s)
& \inf_{\Pr \in \Delta(S)} \alpha_{\Pr} \NREG(b) \\
\notag &= \inf_{\Pr \in \Delta(S) }\left( \alpha_{\Pr}\sum_{s\in S}b(s)\Pr(s)
\right)\\ 
\notag &= \sup_{\Pr \in \cP }\left( -\alpha_{\Pr}\sum_{s\in
S}b(s)\Pr(s) \right) .
\end{align}
%joe14
%Recall that Lemma~\ref{lem:Iprop} says, for all acts $f,g$ such that
Recall that, by Lemma~\ref{lem:Iprop}, for all acts $f,g$ such that
$b_f, b_g \in \cB^-$,  
$f \succeq g$ iff $I( b_f ) \ge I( b_g )$.
%joe14
%Thus we have $f \succeq g$ iff $\sup_{\Pr \in \cP }\left(
%\alpha_{\Pr}\sum_{s\in S}(0-u(f(s)))\Pr(s) \right) 
Thus, $f \succeq g$ iff $$\sup_{\Pr \in \Delta(S) }\left(
-\alpha_{\Pr}\sum_{s\in S}u(f(s))\Pr(s) \right) 
%joe14
%\leq \sup_{\Pr \in \cP }\left( \alpha_{\Pr}\sum_{s\in
\leq \sup_{\Pr \in \Delta(S) }\left( -\alpha_{\Pr}\sum_{s\in
S}u(g(s))\Pr(s) \right).$$ 
%joe14*: SLOW DOWN here.  Added next sentence
Note that, for $f \in M^* = \cB^-$, we have $\regret_{M^*,\Pr}(f) = \sup(-
u(f(s)) \Pr(s)$, since $0^*$ dominates all acts in $M^*$. 
%joe14
%In other words, $\succeq = \succeq_{M^*,\cP^+}^{S,Y,U}$.
Thus, $\succeq = \succeq_{M^*,\cP^+}^{S,Y,U}$, where $\cP^+ =
\{(\Pr,\alpha_{\Pr}: \Pr \in \Delta(S)\}$.
By Lemma~\ref{lem:stoye}, this means $(U,\cP^+)$ represents $\succeq_M$
for all menus $M$, as required. 
	
%joe14:
%Maximality of $\cP^+$ is guaranteed by the definition of $\alpha_{\Pr}$. 
We have already observed that $U$ is unique up to affine
transformations, so it remains to show that $\cP^+$ is maximal.  This
follows from the definition of $\alpha_{\Pr}$. 
If $\succeq_M = \succeq_{M,(\cP')^+}^{S,Y,U}$, and $(\alpha',\Pr) \in
(\cP')^+$, then we claim that
$ \alpha' \in \{\alpha \in \R: \alpha \NREG_{\Pr}(b) \ge I(b) \mbox{ for all $b \in \cB^-$}\}$.
%joe14
%If this were not true, then there would be some $b\in \cB^-$ with $||b||
If not, there would be some $b\in \cB^-$ with $||b||
\leq \frac{1}{2}$, such that 
$ \alpha' \NREG_{\Pr}(b) < I(b)$, which, by the definition of
$\prec^{S,Y,U}_{M^*,(\cP')^*}$, means 
%joe14
that
$ l^*_{-1} \prec^{S,Y,U}_{M^*,(\cP')^+} f_b \prec^{S,Y,U}_{M^*,(\cP')^+} l^*_{I(b)}$. 
%joe14
%However, recall that  $I(b_f) = \inf \{ \gamma : l^*_{\gamma}
Recall that  $I(b_f) = \inf \{ \gamma : l^*_{\gamma}
\succeq_{M^*} f \}$. 
%joe14
%Moreover, we know that $\prec^{S,Y,U}_{M^*,(\cP')^+}$ satisfies the
%mixture continuity axiom,  
%which means there exists $p \in (0,1)$ such that $f_b
Moreover, since $\prec^{S,Y,U}_{M^*,(\cP')^+}$ satisfies the
Mixture Continuity,
there exists some $p \in (0,1)$ such that $f_b
\prec^{S,Y,U}_{M^*,(\cP')^+} p l^*_{-1} + (1-p) l^*_{I(b)}
\prec^{S,Y,U}_{M^*,(\cP')^+} \prec^{S,Y,U}_{M^*,(\cP')^+} l^*_{I(b)}$.
%joe14
%which contradicts the definition of $I(b)$. 
This contradicts the definition of $I(b)$. 
Therefore, $ \alpha' \in \{\alpha \in \R: \alpha \NREG_{\Pr}(b) \ge I(b)
%joe14
%\mbox{ for all $b \in \cB^-$}\}$ and hence $\alpha' \leq \alpha_{\Pr}$. 
\mbox{ for all $b \in \cB^-$}\}$, and hence $\alpha' \leq \alpha_{\Pr}$. 

\subsection{Uniqueness of Representation}

In the preceding sections, we have shown that if a family of menu-dependent preferences $\succeq_M$ satisfies
axioms $1-10$, then $\succeq_M$ can be represented as minimizing
weighted expected 
regret with respect to a canonical set 
%joe20
%of weighted probabilities, $\cP^+$, and a utility function. 
$\cP^+$ of weighted probabilities and a utility function. 
We now want to show uniqueness.  

In this section, we show that the canonical set of weighted
probabilities we constructed, when viewed as a set of subnormal
%joe22: typo
%probability measures, is regular , and includes %sam22 used 'regular'
probability measures, is regular and includes %sam22 used 'regular'
%joe20
%at least one proper probability distribution. 
at least one proper probability measure. 
%joe20
%Moreover, such a set of subnormal probability measures provides a
%provides a \emph{unique} representation for a family of preferences
Moreover, this set of sub-probability measures 
is the only regular set that induces a %sam22 used 'regular'
family of preferences 
$\succeq_M$ that
satisfies 
axioms $1-10$.
Our uniqueness result is analogous to the uniqueness results of Gilboa and Schmeidler \cite{GilboaSchmeidler1989}, who show that the convex, closed, and non-empty set of probability measures in their representation theorem for MMEU is unique.

By Lemma~\ref{lem:stoye}, it suffices to consider the preference relation $\succeq_{M^*}$.
The argument is based on two lemmas: the first lemma says that the
canonical set of sub-probability measures is regular; %sam22 used 'regular'
and the second lemma says that a set of sub-probability
measures representing $\succeq_{M^*}$ that is regular %sam22 used 'regular'
and contains at least one proper probability measure
is unique.  
The proof of this second lemma, like the proof of uniqueness in Gilboa and Schmeidler \cite{GilboaSchmeidler1989}, uses a separating hyperplane theorem to show the existence of acts on which two different representations must `disagree'.
However, a slightly different argument is required in our case, since our acts in $M^*$ must have utilities corresponding to nonpositive vectors in $\R^{|S|}$.

\begin{lemma}
Let $\cP^+$ be the canonical set of weighted probability measures representing $\succeq_{M^*}$. 
The set $C(\cP^+)$ of sub-probability measures is regular.
\end{lemma}
\begin{proof}
%joe20
%It is useful to note that by definition, $\mathbf{p} \in C(\cP^+)$ if
It is useful to note that, by definition, $\mathbf{p} \in C(\cP^+)$ if
and only if  
\begin{align*}
E_{ \mathbf{p}}(b) \ge I(b) \mbox{ for all $b \in \cB^-$}
\end{align*}
%joe20: added
(where expectation with respect to a subnormal probability measure is
defined in the obvious way).

%joe20
%We first show that $C(\cP^+)$ is downward-closed, since it is the
%simplest to show.  Suppose $\mathbf{p} \in C(\cP^+)$. 
%We need to show that $\mathbf{q} \in C(\cP^+)$ for all subnormal
%probability distributions $\mathbf{q}$ such that $\mathbf{q}(s) \leq
%\alpha \Pr(s), \forall s\in S$. 
%sam22 added reminder for definition of regular
Recall that a set is regular if it is convex, closed, and downward-closed.
We first show that $C(\cP^+)$ is downward-closed.
Suppose that $\mathbf{p} \in C(\cP^+)$ and $\mathbf{q} \le \mathbf{p}$
(i.e., $\mathbf{q}(s) \leq \alpha \Pr(s)$ for all $s\in S$. 
%joe20
%
%Thus consider any such $\mathbf{q}$. 
Since $\mathbf{p} \in C(\cP^+)$, 
%joe20
%we know that 
%\begin{align*}
%E_{ \mathbf{p}}(b) \ge I(b) \mbox{ for all $b \in \cB^-$}.
%\end{align*}
$E_{ \mathbf{p}}(b) \ge I(b)$ for all $b \in \cB^-$.
%joe20
%However, this means that 
%\begin{align*}
%E_{ \mathbf{q}}(b) \ge I(b) \mbox{ for all $b \in \cB^-$},
Since $\mathbf{q} \le \mathbf{p}$ and, if $b \in cB^-$, we have $b \le
0^*$, it follows that
$E_{ \mathbf{q}}(b) \ge E_{ \mathbf{p}}(b) \ge I(b)$ for all $b
\in \cB^-$, 
and thus $\mathbf{q}\in C(\cP^+)$.

%joe20
%Next we show that $C(\cP^+)$ is closed.
%Let $\mathbf{p} = \lim_{n\to \infty} \mathbf{p}_n$, where each
To see that $C(\cP^+)$ is closed,
let $\mathbf{p} = \lim_{n\to \infty} \mathbf{p}_n$, where each
$\mathbf{p}_n \in C(\cP^+)$. 
%joe20: it's easier to do a direct proof.
%Suppose for contradiction that $\mathbf{p} \notin C(\cP^+)$. 
%Then it must be the case that for some $b\in \cB^-$ and $\epsilon > 0$,
%$$ \sum_{s\in S }\left( \mathbf{p}(s) b(s)\right) < I(b) - \epsilon .$$
%However, since $S$ is finite, we can find $N$ such that for all $n\geq N$, 
%$|  \sum_{s\in S }\left( \mathbf{p}_n(s) b(s)\right) - \sum_{s\in S
%}\left( \mathbf{p}(s) b(s)\right) | \leq \frac{\epsilon}{2}$,
%contradicting the fact that $\mathbf{p}_n \in C(\cP^+)$. 
Since $\mathbf{p}_n \in C(\cP^+)$ it must be the case that 
$E_{\mathbf{p}_n}(b) \ge I(b)$ for all $b \in \cB^-$.  By the continuity
of expectation, it follows that
$E_{\mathbf{p}}(b) \ge I(b)$ for all $b \in \cB^-$.  Thus, 
$\mathbf{p} \in C(\cP^+)$.

To show that $C(\cP^+)$ is convex, 
%joe20: unnecessarily complicated
%consider any improper probability
%distribution $\mathbf{p}_t = t (\alpha_1 \Pr_1) + (1-t) (\alpha_2
%\Pr_2)$, where $\alpha_1 \Pr_1 \in C(\cP^+)$ and $\alpha_2 \Pr_2 \in
%C(\cP^+)$. 
%To show that $\mathbf{p}_t \in C(\cP^+)$, it suffices to show that 
%$$E_{ \mathbf{p}_t}(b) \ge I(b) \mbox{ for all $b \in \cB^-$}.
%$$
%To see that this holds, note that for any $b \in \cB^-$, we have
%$$E_{ \mathbf{p}_t}(b) = t  \alpha_1 E_{\Pr_1} (b) +(1-t) \alpha_2
%E_{\Pr_2} (b) \ge I(b), 
%$$ since $\alpha_1 E_{\Pr_1} (b) \geq I(b)$ and $\alpha_2  E_{\Pr_2}
%(b) \geq I(b)$. 
suppose that $\mathbf{p}, \mathbf{q} \in C(\cP^+)$.  Then 
$E_{ \mathbf{p}}(b) \ge I(b)$ and 
$E_{ \mathbf{q}}(b) \ge I(b)$ 
for all $b \in \cB^-$.  It easily follows that for all $a \in (0,1)$, 
$E_{ a\mathbf{p} + (1-a) \mathbf{q}}(b) \ge I(b)$ for all $b \in
\cB^-$.
Thus, $a\mathbf{p} + (1-a) \mathbf{q} \in C(\cP^+)$.
\end{proof}

\begin{lemma}
A set of sub-probability measures representing $\succeq_{M^*}$ that is
regular, %sam22 used 'regular'
and has at least one proper probability
%joe20
%distribution, is necessarily unique. 
measure is unique. 
\end{lemma}
\begin{proof}
Suppose for contradiction that there exists two regular sets of
subnormal %sam22 used 'regular' 
probability distributions, $C_1$ and $C_2$, that 
%joe20
represent $\succeq_{M^*}$ and have at least one proper probability measure.  
%joe20
%Further suppose that both $C_1$ and $C_2$ represent $\succeq_{M^*}$.
%We'll use Minkowski's separating hyperplane theorem to derive a
%contradiction.  

First, without loss of generality, let $\mathbf{q} \in C_2 \backslash C_1$.
%joe20
%We will actually look at an extension of $C_1$ that is closed downwards
We actually look at an extension of $C_1$ that is downward-closed
in each component to $-\infty$. 
%joe20
%That is, we let $\overline{C}_1 = \{ \mathbf{p} \in \R^{|S|} \mid
%\exists \mathbf{p'} \in C_1, \mathbf{p}(s) \leq \mathbf{p'}(s) \forall
%s\in S \}$. 
Let $\overline{C}_1 = \{ \mathbf{p} \in \R^{|S|} :
\mathbf{p} \leq \mathbf{p'} \}$.  Note an element $\mathbf{p}$ of
$\overline{C}_1$ may not be subnormal probability measures; we do not
require that $\mathbf{p}(s) \ge 0$ for all $s \in S$.
Since $\overline{C}_1$ and $\{\mathbf{q}\}$ are closed, convex, and
disjoint, and $\{\mathbf{q}\}$ is compact, the separating hyperplane
%joe20: Sam, you may want to add a reference here;  perhaps Rockafellar
%joe21: you still need a reference
theorem \cite{Rockafellar} says that there exists $\theta \in \R^{|S|}$ and $c\in \R$ such
that  
\begin{align}\label{equ:separating}
\theta \cdot \mathbf{p}> c \text{ for all } \mathbf{p}\in \overline{C}_1 \text{, and } \theta \cdot \mathbf{q} < c.
\end{align}
By scaling $c$ appropriately, we can assume that $|\theta(s)| \leq 1$
for all $s\in S$.  
Now we argue that it must be the case that $\theta(s) \leq 0$ for all
%joe20
%$s\in S$ (so that $\theta$ would correspond to the utility profile of
$s\in S$ (so that $\theta$ corresponds to the utility profile of
some act in $M^*$). 
%joe20
%Suppose, for the purpose of establishing a contradiction, that
Suppose that
$\theta(s') > 0 $ for some $s'\in S$. 
By (\ref{equ:separating}), $\theta \cdot \mathbf{p} > c \text{ for all } \mathbf{p}\in \overline{C}_1$. 
However, consider $\mathbf{p^*} \in \overline{C}_1$ defined by 
\begin{align*}
\mathbf{p^*}(s) = 
\begin{cases}
0 \text{, if } s\neq s' \\
\frac{-|c|}{\theta(s)} \text{, if } s = s'.
\end{cases}
\end{align*}
Clearly, $\theta \cdot \mathbf{p^*} \leq c $, contradicting (\ref{equ:separating}).
Thus it must be the case that  $\theta(s) \leq 0$ for all $s\in S$.

%Now, we will show that $C_1$ and $C_2$ cannot both represent the same
%preference relation $\succeq_{M^*}$. 
%Consider $\theta$ given by the separating hyperplane theorem, and let
Consider the $\theta$ given by the separating hyperplane theorem, and let
$f$ be an act such that $u\circ f = \theta$.  
By continuity, $f \sim_{M^*} l^*_d$ for some constant act $l^*_d$. 
%joe20: removed paragraph break
%
%By the assumption that $C_1$ and $C_2$ both represent $\succeq_{M^*}$,
%and that each of $C_1$ and $C_2$ contains proper probability
%distributions, we know that  
Since $C_1$ and $C_2$ both represent $\succeq_{M^*}$,
and $C_1$ and $C_2$ both contain a proper probability
measure, 
\begin{align*}
\min_{\mathbf{p} \in C_1} \mathbf{p}\cdot (u\circ f) = \min_{\mathbf{p} \in C_1} \mathbf{p}\cdot (u\circ l^*_d) = d = \min_{\mathbf{p} \in C_2} \mathbf{p} \cdot (u\circ f) .
\end{align*}
%joe20: removed paragraph break
%
%However, by (\ref{equ:separating}), we know that 
However, by (\ref{equ:separating}), 
\begin{align*}
\min_{\mathbf{p}\in C_1} \mathbf{p}\cdot (u\circ f)  > c > \min_{\mathbf{p} \in C_2} \mathbf{p} \cdot (u\circ f), 
\end{align*}
which is a contradiction. 

\end{proof}

%\subsection{Likelihood Updating Preserves Canonicity}

%If the unconditional preferences $\succeq_M$ satisfy axioms $1-10$, and the conditional preferences $\succeq_{E,M}$ are related to $\succeq_M$ by MDC, that is, 
%$f \succeq_{E,M} g \text{ iff } fEh \succeq_{MEh} gEh$ for some $h\in M$,
%then since axioms $1-10$ hold for $\succeq_{MEh}$, $\succeq_{M,E}$ also satisfy axioms $1-10$.
%In particular, this means that $\succeq_{E,M}$ also has a canonical representation with weight $1$ on some distribution.

%Not surprisingly, if we start off with canonical weighted probabilities $\cP^+$ representing $\succeq_M$, then likelihood updating results in the canonical weights for the representation of $\succeq_{E,M}$. 
%The reasoning is simple: since likelihood updating of $\cP^+$ to $\cP^+|E$ preserves the convexity, closedness, downward-closedness, and the inclusion of a proper probability measure, of $C(\cP^+|E)$, $\cP^+|E$ is necessarily the canonical set of weighted probabilities for $\succeq_{E,M}$.

\bibliographystyle{abbrv}
%joe8: added joe.bib
\bibliography{joe,awareness}

\end{document}